\documentclass{llncs}
\usepackage{amssymb}
\usepackage{graphicx}

\usepackage{url}
\urldef{\mails}\path|meizeh@cs.technion.ac.il|

\usepackage{ntheorem}
\theoremseparator{.}
\newtheorem{obs}{Observation}

\newtheorem{cor}{Corollary}

\newcommand{\myparagraph}[1]{\par\smallskip\par\noindent{\bf{}#1:~}}
\newcommand{\comment}[1]{}

\usepackage{algorithm}
\usepackage{algorithmic}
\newcommand{\alg}[1]{\mbox{\sf #1}}  

\usepackage{amsopn}

\DeclareMathOperator{\image}{image}

\begin{document}

\mainmatter

\title{Mixing Color Coding-Related Techniques (Extended Abstract)}

\author{Meirav Zehavi}

\institute{Department of Computer Science, Technion IIT,
 Haifa 32000, Israel\\
\mails}

\maketitle

\begin{abstract}
Narrow sieves, representative sets and divide-and-color are three breakthrough color coding-related techniques, which led to the design of extremely fast parameterized algorithms. We present a novel family of strategies for applying mixtures of them. This includes: (a) a mix of representative sets and narrow sieves; (b) a faster computation of representative sets under certain separateness conditions, mixed with divide-and-color and a new technique, ``balanced cutting''; (c) two mixtures of representative sets, iterative compression and a new technique, ``unbalanced cutting''. We demonstrate our strategies by obtaining, among other results, significantly faster algorithms for {\sc $k$-Internal Out-Branching} and {\sc Weighted 3-Set $k$-Packing}, and a framework for speeding-up the previous best deterministic algorithms for {\sc $k$-Path}, {\sc $k$-Tree}, {\sc $r$-Dimensional $k$-Matching}, {\sc Graph Motif} and {\sc Partial~Cover}.
\end{abstract}

\section{Introduction}

A problem is {\em fixed-parameter tractable (FPT)} with respect to a parameter $k$ if it can be solved in time $O^*(f(k))$ for some function $f$, where $O^*$ hides factors polynomial in the input size. The color coding technique, introduced by Alon et al.~\cite{colorcoding}, led to the discovery of the first single exponential time FPT algorithms for many subcases of {\sc Subgraph Isomorphism}. In the past decade, three breakthrough techniques improved upon it, and led to the development of extremely fast FPT algorithms for many fundamental problems. This includes the combinatorial divide-and-color technique \cite{divandcol}, the algebraic multilinear detection technique \cite{multilineardetection,appmultilinear,williamskpath} (which was later improved to the more powerful narrow sieves technique \cite{integerweights,bjo10}), and the combinatorial representative sets technique \cite{representative}.

Divide-and-color was the first technique that resulted in (both randomized and deterministic) FPT algorithms for weighted problems that are faster than those relying on color coding. Later, representative sets led to the design of deterministic FPT algorithms for weighted problems that are faster than the randomized ones based on divide-and-color. The fastest FPT algorithms, however, rely on narrow sieves. Unfortunately, narrow sieves is only known to be relevant to the design of {\em randomized} algorithms for {\em unweighted} problems.\footnote{More precisely, when used to solve weighted problems, the running times of the resulting algorithms have {\em exponential} dependencies on the length of the input~weights.}

We present novel strategies for applying these techniques, combining the following elements (see Section \ref{section:strategies}).

\renewcommand{\labelitemi}{$\bullet$}
\renewcommand{\labelitemii}{$-$}

\begin{itemize}
\item Mixing narrow sieves and representative sets, previously considered to be two {\em independent} color coding-related techniques.

\item Under certain ``separateness conditions'', speeding-up the best known computation of representative sets.

\item Mixing divide-and-color-based preprocessing with the computation in the previous item, speeding-up any standard representative sets-based algorithm.

\item Cutting the universe into small pieces in two special manners, one used in the mix in the previous item, and the other mixed with a non-standard representative sets-based algorithm to improve its running time (by decreasing the size of the partial solutions it computes).
\end{itemize}

{\noindent To demonstrate our strategies, we consider the following well-studied problems.}

\myparagraph{$k$-Internal Out-Branching ($k$-IOB)}Given a directed graph $G=(V,E)$ and a parameter $k\in\mathbb{N}$, decide if $G$ has an out-branching (i.e., a spanning tree with exactly one node of in-degree 0) with at least $k$ internal nodes.

\myparagraph{Weighted $k$-Path}Given a directed graph $G=(V,E)$, a weight function $w: E\rightarrow\mathbb{R}$, $W\!\in\mathbb{R}$ and a parameter $k\in\mathbb{N}$, decide if $G$ has a simple directed path on exactly $k$ nodes and of weight at most $W$.

\myparagraph{Weighted $3$-Set $k$-Packing ($(3,k)$-WSP)}Given a universe $U$, a family $\cal S$ of subsets of size 3 of $U$, a weight function $w\!: {\cal S}\!\rightarrow\!\mathbb{R}$, $W\!\in\!\mathbb{R}$ and a parameter $k\!\in\!\mathbb{N}$, decide if there is a subfamily ${\cal S}'\!\subseteq\! {\cal S}$ of $k$ disjoint sets and weight at~least~$W$.

\myparagraph{$P_2$-Packing}Given an undirected graph $G=(V,E)$ and a parameter $k\in\mathbb{N}$, decide if $G$ has $k$ (node-)disjoint simple paths, each on 3 nodes.

\medskip

The {\sc $k$-IOB} problem is NP-hard since it generalizes {\sc Hamiltonian Path}. It is of interest, for example, in database systems \cite{outbranchpatent}, and for connecting cities with water pipes \cite{kISPbounddeg}. Many FPT algorithms were developed for {\sc $k$-IOB} and related variants (see, e.g., \cite{kIOB49k,thesis11,kISP8k,kIOB16k,kIOB2klogk,4kist,kISP24klogk,esarepresentative,ipec13}). We solve it in deterministic time $O^*(5.139^k)$ and randomized time $O^*(3.617^k)$, improving upon the previous best deterministic time $O^*(6.855^k)$ \cite{esarepresentative} and randomized time $O^*(4^k)$ \cite{thesis11,ipec13}. To this end, we establish a relation between out-trees (directed trees with exactly one node, the root, of in-degree 0) that have many leaves and paths on 2 nodes. This shows how certain partial solutions to {\sc $k$-IOB} can be completed efficiently via a computation of a maximum matching in the underlying undirected graph.

We also present a unified approach for speeding-up standard representative sets-based algorithms; thus, our approach can be used to improve upon the best known deterministic algorithms for many problems (since numerous problems were solved using color coding). In particular, our approach can be used to modify the previous best deterministic algorithms (that already rely on the best known computation of representative sets) for the classic {\sc $k$-Path}, {\sc $k$-Tree}, {\sc $r$-Dimensional $k$-Matching ($(r,k)$-DM)}, {\sc Graph Motif with Deletions (GM$_\mathrm{D}$)} and {\sc Partial Cover (PC)} problems, including their weighted variants, which run in times 
$O^*(2.6181^k)$ \cite{productFam,esarepresentative}, $O^*(2.6181^k)$ \cite{productFam,esarepresentative}, $O^*(2.6181^{(r-1)k})$~\cite{fsttcs13}, $O^*(2.6181^{2k})$ \cite{mfcs14} and $O^*(2.6181^k)$~\cite{esarepresentative}, to run in times $O^*(2.5961^k)$, $O^*(2.5961^k)$, $O^*(2.5961^{(r-1)k})$, $O^*(2.5961^{2k})$ and  $O^*(2.5961^k)$, respectively.
We demonstrate our approach using {\sc Weighted $k$-Path}, since obtaining an $O^*(2^k)$ time deterministic algorithm for it is a major open problem.

In the past decade, {\sc $(3,k)$-WSP} and {\sc $(3,k)$-SP} enjoyed a race towards obtaining the fastest FPT algorithms that solve them (see \cite{bjo10,impdetmatpac,chenalgorithmica2004,divandcol,fellowsbook,fellowsalgorithmica2008,koutis2005,liutamc2007,wangcocoon2008,chenipec,multilineardetection,wangtamc2008,corrmatchpack}). We solve {\sc $(3,k)$-WSP} in deterministic time $O^*(8.097^k)$, significantly improving upon $O^*(12.155^k)$, which is both the previous best running time of an algorithm for {\sc $(3,k)$-WSP} and the previous best running time of a deterministic algorithm for {\sc $(3,k)$-SP} \cite{corrmatchpack}. The {\sc $P_2$-Packing} problem is a subcase of {\sc $(3,k)$-SP}, for which specialized FPT algorithms are given in \cite{p2packdet,p2packrand,p2packraible,p2pack2006}. Feng et al.~\cite{p2packrand} give a randomized $O^*(6.75^k)$ time algorithm for {\sc $P_2$-Packing},\footnote{{\sc $(3,k)$-SP}, generalizing {\sc $P_2$-Packing}, is solvable in randomized time $O^*(3.327^k)$~\cite{bjo10}.} for which Feng et al.~\cite{p2packdet} give a deterministic version that runs in time $O^*(8^{k+o(k)})$. We observe that the algorithms of \cite{p2packdet,p2packrand} can be modified to solve {\sc $P_2$-Packing} in deterministic time $O^*(6.75^{k+o(k)})$. Then, we give an alternative algorithm that solves {\sc $P_2$-Packing} in deterministic time $O^*(6.777^k)$.

\myparagraph{Organization} Section \ref{section:techniques} contains information on color coding-related techniques. Then, Section \ref{section:strategies} presents an overview of our strategies.
Finally, Section~\ref{section:disrep} contains the proof a central component of our second strategy (that is, a computation of {\em generalized} representative sets), which might be of independent interest; due to lack of space, other technical details are deferred to the appendix (Section \ref{section:strategies} contains the relevant pointers).\footnote{The appendix is also available on \alg{http://arxiv.org/abs/1410.5062}.}


\section{Color Coding-Related Techniques}\label{section:techniques}

In this paper, we use a known algorithm based on narrow sieves as a black box; thus, we avoid describing this technique. We proceed by giving a brief description of divide-and-color, followed by a more detailed one of representative sets.

\myparagraph{Divide-and-Color} Divide-and-color is based on recursion; at each step, we color elements randomly or deterministically. In our strategies, we are interested in applying only one step, which can be viewed as using color coding with only two colors. In such a step, we have a set $A$ of $n$ elements, and we seek a certain subset $A^*$ of $k$ elements in $A$. We partition $A$ into two (disjoint) sets, $B$ and $C$, by coloring its elements. Thus, we get the problem of finding a subset $B^*\subseteq A^*$ in $B$, and another problem of finding the subset $C^*\!=\!A^*\setminus B^*$ in $C$. The partition should be done in a manner that is both efficient and results in an easier problem, which does not necessarily mean that we get two independent problems (of finding $B^*$ in $B$ and $C^*$ in $C$). Deterministic applications of divide-and-color often use a tool called an $(n,k)$-universal set \cite{splitter}. We need its following generalization:

\begin{definition}\label{dfn:universal_set}
Let ${\cal F}$ be a set of functions $f:\{1,2,\ldots,n\}\rightarrow \{0,1\}$. We say that ${\cal F}$ is an $(n,k,p)$-universal set if for every subset
$I\subseteq \{1,2,\ldots,n\}$ of size $k$ and a function $f':I\rightarrow\{0,1\}$ that assigns '1' to exactly $p$ indices, there is a function $f\in{\cal F}$ such that for all $i\in I$, $f(i)=f'(i)$.
\end{definition}

{\noindent The next result (of \cite{representative}) asserts that small universal sets can be computed fast.}
\begin{theorem}\label{theorem:splitter}
There is an algorithm that, given integers $n,k$ and $p$, computes an $(n,k,p)$-universal set ${\cal F}$ of size $O^*({k\choose p}2^{o(k)})$ in deterministic time $O^*({k\choose p}2^{o(k)})$.
\end{theorem}

\myparagraph{Representative Sets} We first give the  definition of a representative family, and then discuss its relevance to the design of FPT algorithms. We note that a more general definition, not relevant to this paper, is given in \cite{representative,marx09}.

\begin{definition}\label{def:repfam}
Given a universe $E$, a family ${\cal S}$ of subsets of size $p$ of $E$, a function $w:{\cal S}\rightarrow\mathbb{R}$ and $k\in\mathbb{N}$, we say that a subfamily $\widehat{\cal S}\subseteq{\cal S}$ {\em max (min) $(k-p)$-represents} $\cal S$ if for any pair $X\in{\cal S}$ and $Y\subseteq E\setminus X$ such that $|Y|\leq k-p$, there is $\widehat{X}\in\widehat{\cal S}$ disjoint from $Y$ such that $w(\widehat{X})\geq w(X)$ $(w(\widehat{X})\leq w(X))$.
\end{definition}

Roughly speaking, Definition~\ref{def:repfam} implies that if a set $Y$ can be extended to a set of size at most $k$ by adding a set $X\in \cal S$, then it can also be extended to a set of the same size by adding a set $\widehat{X}\in\widehat{\cal S}$ that is at least as good as $X$. The special case where $w(S)=0$, for all $S\in{\cal S}$, is the unweighted version of the definition.

Plenty FPT algorithms are based on dynamic programming, where after each stage, the algorithm computes a family $\cal S$ of sets that are partial solutions. At this point, we can compute a subfamily $\widehat{\cal S}\!\subseteq\!{\cal S}$ that represents $\cal S$. Then, each reference to ${\cal S}$ can be replaced by a reference to $\widehat{\cal S}$. The representative family $\widehat{\cal S}$ contains ``enough'' sets from ${\cal S}$; therefore, such replacement preserves the correctness of the algorithm. Thus, if we can efficiently compute representative families that are small enough, we can substantially improve the running time of~the~algorithm.

The {\em Two Families Theorem} of Bollob$\acute{\mathrm{a}}$s \cite{bollobas65} implies that for any universe $E$, a family ${\cal S}$ of subsets of size $p$ of $E$ and a parameter $k$ ($\geq p$), there is a subfamily $\widehat{\cal S}\subseteq {\cal S}$ of size ${k \choose p}$ that $(k-p)$-represents ${\cal S}$. Monien \cite{monien85} computed representative families of size $\sum_{i=0}^{k-p}p^i$ in time $O(|{\cal S}|p(k-p)\sum_{i=0}^{k-p}p^i)$, and Marx \cite{marx06} computed such families of size ${k \choose p}$ in time $O(|{\cal S}|^2p^{k-p})$. Recently, Fomin et al.~\cite{representative} introduced a powerful technique that enables to compute representative families of size ${k \choose p}2^{o(k)}\log\! |E|$ in time $O(|{\cal S}|(k/(k-p))^{k-p}2^{o(k)}\log\! |E|)$. We~need~the following tradeoff-based generalization of their computation, given in \cite{productFam,esarepresentative}:

\begin{theorem}\label{theorem:esarep}
Given a fixed $c\!\geq\! 1$, a universe $E$, a family ${\cal S}$ of subsets of size $p$~of~$E$, a function $w\!:{\cal S}\rightarrow\mathbb{R}$ and a parameter $k\in\mathbb{N}$, a subfamily $\widehat{\cal S}\subseteq{\cal S}$ of size $\frac{(ck)^{k}}{p^p(ck-p)^{k-p}}2^{o(k)}\log |E|$ that max (min) $(k-p)$-represents $\cal S$ can be found in time $\displaystyle{O(|{\cal S}|(ck/(ck\!-\!p))^{k-p}2^{o(k)}\log |E| + |{\cal S}|\log|{\cal S}|)}$.
\end{theorem}

\section{Our Mixing Strategies}\label{section:strategies}

In this section, we give an overview of our mixing strategies in the context of the problems mentioned in the introduction, including references to relevant appendices. The first strategy builds upon the approach used in \cite{representative} to solve {\sc Long Directed Cycle}, and it does not involve a mixture of techniques, but it is relevant to this paper since the second strategy builds upon it; the other strategies are novel and involve mixtures of techniques.


\bigskip
\myparagraph{Strategy I} Our deterministic {\sc $k$-IOB} algorithm follows the strategy in~Fig.~\ref{fig:strategies}(I). The first reduction (of \cite{kIOB49k}) allows to focus on finding a small out-tree rather than an out-branching, while the second allows to focus on finding an even smaller out-tree, but we need to find it along with a set of paths on 2 nodes (Appendix~\ref{section:kiobred}). The second reduction might be of independent interest (indeed, the paper \cite{spotting} applies our reduction). We then use a representative sets-based procedure in a manner that does not directly solve the problem, but returns a family of partial solutions that are trees. We try to extend each partial solution to a solution in polynomial time via a computation of a maximum matching (Appendix~\ref{section:kiobdet}). We thus obtain the following theorem.

\begin{theorem}\label{theorem:kiobDet}
Relying on Strategy I, {\sc $k$-IOB} can be solved in deterministic time $O^*(5.139^k)$.
\end{theorem}

{\noindent Generally, this strategy is relevant to the design of deterministic algorithms as follows. We are given a problem which can be reduced to two suproblems whose solutions should be disjoint. If we could disregard the disjointness condition, the first subproblem is solvable using representative sets, and the second subproblem is solvable in polynomial time. Then, we modify the representative sets-procedure to compute a family of solutions (rather than one solution), such that for any set that might be a solution to the second subproblem, there is a solution in the family that is disjoint from this set. Thus, we can iterate over every solution $A$ in the family, remove its elements from the input, and attempt to find a solution $B$ to the second subproblem in the remaining part of the input---if we succeed, the combination of $A$ and $B$ should solve the original problem.}

\begin{figure}[ht!]
\centering
\includegraphics[scale=0.7]{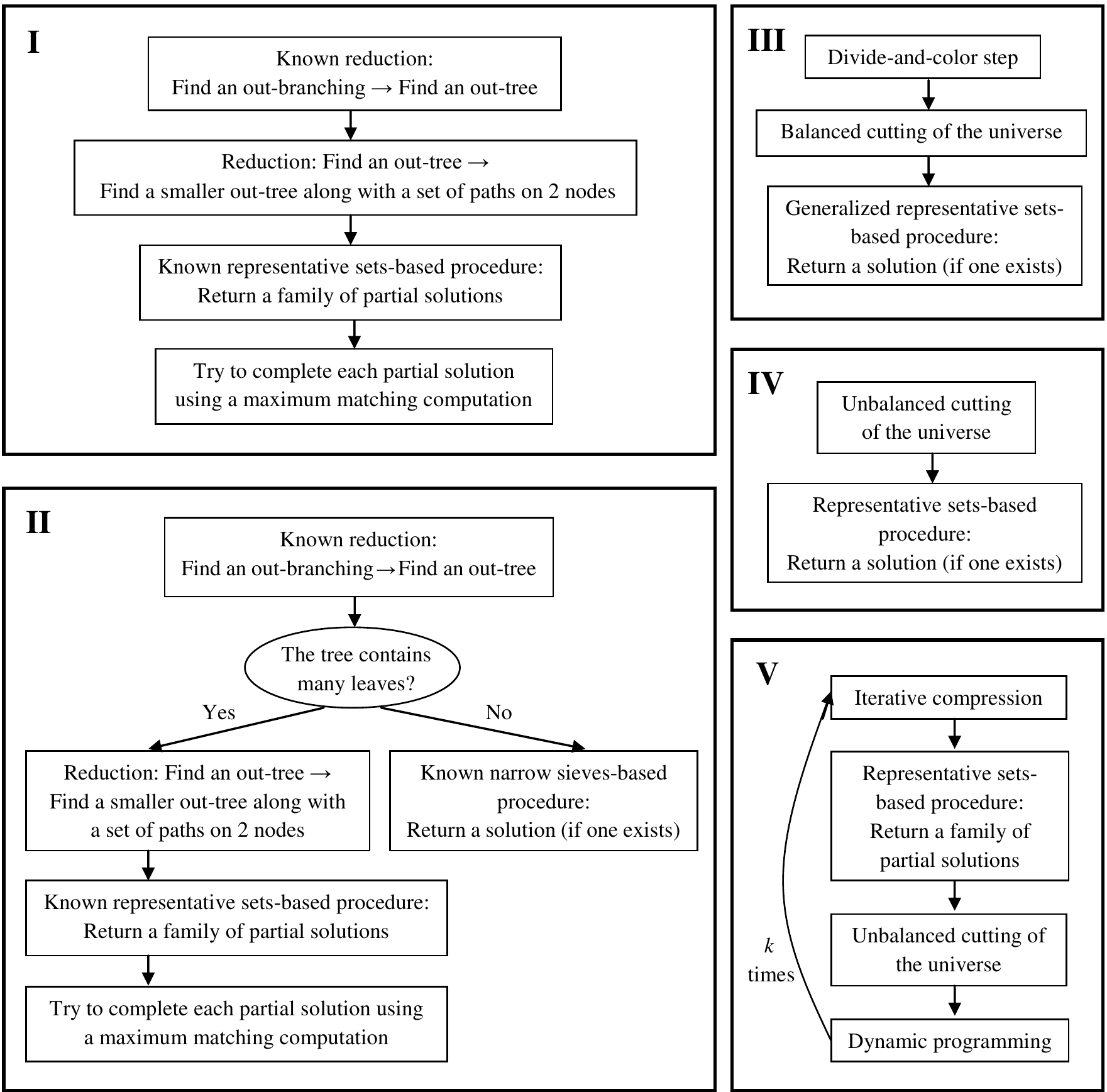}
\caption{Strategies for mixing color coding-related techniques, described in Section~\ref{section:strategies}. We use these strategies to develop a deterministic algorithm for {\sc $k$-IOB} (I), a randomized algorithm for {\sc $k$-IOB} (II), deterministic algorithms for {\sc $k$-Path}, {\sc $k$-Tree}, {\sc $(r,k)$-DM}, {\sc GM$_\mathrm{D}$} and {\sc PC}, including their weighted versions (III), a deterministic algorithm for {\sc $(3,k)$-WSP} (IV), and a deterministic algorithm for {\sc $P_2$-Packing}~(V).}
\label{fig:strategies}
\end{figure}

\bigskip
\myparagraph{Strategy II} Our second result, a randomized FPT algorithm for {\sc $k$-IOB}, builds upon our first algorithm and follows the strategy in Fig.~\ref{fig:strategies}(II). This strategy shows the usefulness of mixing narrow sieves and representative sets (Appendix \ref{section:kiobrand}), previously considered to be two independent tools for developing FPT algorithms. This strategy indicates that the representative sets technique is relevant to the design of fast randomized FPT algorithms, even for unweighted problems. We thus obtain the following theorem, which breaks, for the first time, the $O^*(4^k)$-time barrier for {\sc $k$-IOB}.\footnote{Previous algorithms obtained solutions in time $O^*(2^{k})$ to a problem that is {\em more general} than {\sc Directed $k$-Path}, and could thus solve {\sc $k$-IOB} in time $O^*(2^{2k})$ (relying solely on the first component in Strategies I and II). Thus, it seemed that an $O^*((4-\epsilon)^k)$ time algorithm for {\sc $k$-IOB} would imply an $O^*((2-\epsilon)^k)$ time algorithm for {\sc Directed $k$-Path}, which is a problem that is open for several decades.}

\begin{theorem}\label{theorem:kiobRand}
Relying on Strategy II, {\sc $k$-IOB} can be solved in randomized time $O^*(3.617^k)$.
\end{theorem}

{\noindent Generally, this strategy is relevant to the scenario of Strategy I (in the context of {\em randomized} algorithms), where we also rely on a condition described below. The difference is in the manner in which we handle this scenario---we ``guess'' what fraction of the solution solves the harder subproblem (i.e., the one that we need to solve using a color coding-related technique). If this fraction is small, it is advised to follow Strategy I (since representative sets are very efficient in finding a family of partial solutions that are significantly smaller than the size of the entire solution); otherwise, we need to rely on a condition stating that if this fraction is large, then the size of the entire solution is actually small---then, we can efficiently find it by using a randomized narrow sieves-based procedure.}

\bigskip
\myparagraph{Strategy III} Our third result, a deterministic FPT algorithm for {\sc Weighted $k$-Path}, follows the strategy in Fig.~\ref{fig:strategies}(III). This strategy can be used to speed-up algorithms for other problems based on a standard application of representative sets, such as the {\sc Weighted $k$-Tree}, {\sc $(r,k)$-DM}, {\sc GM$_\mathrm{D}$} and {\sc PC} algorithms of \cite{productFam,fsttcs13,mfcs14,esarepresentative}, where: (1) elements are never deleted from partial solutions (this is not the case in our {\sc $(3,k)$-WSP} algorithm); (2) the solution is one of the computed representative sets (this is not the case in our {\sc $k$-IOB} and {\sc $P_2$-Packing} algorithms). We rely on a generalization of Definition \ref{def:repfam} and~Theorem~\ref{theorem:esarep}:

\begin{definition}\label{def:disrep}
Let $E_1,E_2,\ldots,E_t$ be disjoint universes, $p_1,p_2,\ldots,p_t\in\mathbb{N}$, and ${\cal S}$ be a family of subsets of $(\bigcup_{i=1}^tE_i)$ such that $[\forall S\in{\cal S},i\in\{1,2,\ldots,t\}$: $|S\cap E_i|=p_i]$. Given a function $w:{\cal S}\rightarrow\mathbb{R}$ and parameters $k_1,k_2,\ldots,k_t\in\mathbb{N}$, we say that a subfamily $\widehat{\cal S}\subseteq{\cal S}$ {\em max (min) $(k_1-p_1,k_2-p_2,\ldots,k_t-p_t)$-represents} $\cal S$ if for any pair $X\in{\cal S}$ and $Y\subseteq (\bigcup_{i=1}^tE_i)\setminus X$ such that $[\forall i\in\{1,2,\ldots,t\}: |Y\cap E_i|\leq k_i-p_i]$, there is $\widehat{X}\in\widehat{\cal S}$ disjoint from $Y$ such that $w(\widehat{X})\geq w(X)$ $(w(\widehat{X})\leq w(X))$.
\end{definition}

\begin{theorem}\label{theorem:disrep}
Given fixed $c_1,c_2,\ldots,c_t\geq 1$, and $E_1,E_2,\ldots,E_t$, $p_1,p_2,\ldots,p_t$,~${\cal S}$,~$w$ and $k_1,k_2,\ldots,k_t$ as in Definition \ref{def:disrep}, a subfamily $\widehat{\cal S}\!\subseteq\!{\cal S}$ of size $\prod_{i=1}^t(\! \frac{(c_ik_i)^{k_i}}{{p_i}^{p_i}(ck_i-p_i)^{k_i-p_i}}$ $2^{o(k_i)}\log |E_i|)$ that max (min) $(k_1-p_1,k_2-p_2,\ldots,k_t-p_t)$-represents $\cal S$~can~be~found in time $O(|{\cal S}|\prod_{i=1}^t\left((c_ik_i/(c_ik_i-p_i))^{k_i-p_i}2^{o(k_i)}\log |E_i|\right) + |{\cal S}|\log|{\cal S}|)$.
\end{theorem}

{\noindent Section \ref{section:disrep} proves this theorem, and Appendices \ref{section:disrepkpath} and \ref{section:proofkcwp} use it to solve a subcase of {\sc Weighted $k$-Path}, that we call {\sc Cut Weighted $k$-Path ($k$-CWP)}. Then, Appendix~\ref{section:kpathdet} translates the problem to this subcase via divide-and-color preprocessing, mixed with a technique that we call {\em balanced~cutting}.}

Since Appendices \ref{section:disrepkpath}--\ref{section:kpathdet} are technically involved, Appendix \ref{sec:kPathIntuition} attempts to describe, in an informal manner, the intuition that guided us. Roughly speaking, this is the main idea---we first cut a small part of the universe that will play the role of $E_1$, and let the remaining part of the universe play the role of $E_2$ (this is the first, simple phase of balanced cutting). Then, using divide-and-color, we partition $E_1$ into two sets, $L$ and $R$, such that we can consider $L$ ($R$) only in the first (second) half of an execution of an algorithm that should solve {\sc $k$-CWP}. Now, we cut the entire universe into small pieces in a special, {\em implicit} manner (this is the second phase of balanced cutting); this overall translates {\sc Weighted $k$-Path} to {\sc $k$-CWP}, where we do not seek a path, but many small paths, and we are also told in which order we should seek them. The set $L$ ($R$) is distributed in a (at worst)``balanced manner'' between the pieces of the universe that should be considered in the first (second) half of the execution of the {\sc $k$-CWP} algorithm. We show that this balanced manner actually ``distorts'' the universe, so that when we need to compute a representative family, for example, in the first half of the execution, there is a significant gap between the fractions $\frac{p_1}{k_1}$ and $\frac{p_2}{k_2}$, where $p_1$ ($p_2$) is the number of elements from $L$ ($E_2$) that our partial solutions contain, and $k_1$ ($k_2$) is the total number of elements from $L$ ($E_2$) that are contained in a solution. This gap leads to an improved running time (gained by using generalized representative sets).

Thus, we obtain the following result.

\begin{theorem}\label{theorem:kPathRand}
Relying on Strategy III, {\sc Weighted $k$-Path} is solvable in deterministic time $O^*(2.59606^k)$.
\end{theorem}

{\noindent Generally, this strategy can be used to speed-up any standard representative sets-related algorithm (as mentioned above), and is therefore applicable to a wide variety of problems.}

\bigskip
\myparagraph{Strategies IV and V} Our fourth and fifth results, deterministic FPT algorithms for {\sc $(3,k)$-WSP} and {\sc $P_2$-Packing}, follow the strategies in Fig.~\ref{fig:strategies}(IV) and \ref{fig:strategies}(V). Here we also cut the universe into small parts (see Appendices \ref{section:3kwsp} and \ref{section:p2pack}), though in a different manner, which allows us to delete more elements from partial solutions than \cite{corrmatchpack}. We call this technique {\em unbalanced cutting}. Roughly speaking, unbalanced cutting {\em explicitly} partitions the universe into small pieces, using which it orders the entire universe (the order is only partially arbitrary), such that at certain points during the computation, we are ``given'' an element $e$ that implies that from now on, we should not try to add (to partial solutions) elements that are ``smaller'' than $e$---thus, we can delete (from our partial solutions) all the elements that are smaller than $e$ (since we will not encounter them again). Since we are handling smaller partial solutions, we get a better running time. The number of elements that we can delete at each point where we are ``given'' an element, is computed by defining a special recursive formula. Again, as in Strategy III, some nontrivial technical details are involved, and therefore we added Appendix \ref{sec:wspIntuition}, which attempts to describe the intuition that guided us.

To solve {\sc $P_2$-Packing}, we also use iterative compression \cite{iterativecompress}, relying on a result of \cite{p2packraible}. Informally, applying iterative compression means solving a variant of the problem where, given a (partial) solution of size $t\!-\!1$, we need to find a solution of size $t$. The general problem can be solved by running the specialized algorithm $k$ times, iteratively increasing the value of $t$ from 1 to $k$. Note that iterative compression has already been combined with representative sets in \cite{fsttcs13} (to solve~{\sc $(3,k)$-DM}).

Thus, we obtain the following results.

\begin{theorem}\label{theorem:WSP}
Relying on Strategy IV, {\sc $(3,k)$-WSP} is solvable in deterministic time $O^*(8.097^k)$.
\end{theorem}

\begin{theorem}\label{theorem:p2Pack}
{\sc $P_2$-Packing} can be solved in deterministic time $O^*(6.75^{k+o(k)})$. Relying on Strategy V, it can be solved in deterministic time $O^*(6.777^k)$.
\end{theorem}

{\noindent Generally, these strategies may be relevant to cases where we can isolate a certain layer of elements in a partial solution (in the case of {\sc $k$-WSP}, this layer consists of each smallest element in a set of 3 elements in the partial solution) such that as the computation progresses, we can remove (from partial solutions) elements from this layer. For a more precise explanation, in the context of {\sc $k$-WSP}, see Appendix \ref{sec:wspIntuition}. Applying Strategy IV attempts to allow us to delete elements not only from the layer that we isolated, but also from the other layers.} Strategy V can be viewed as a variant of Strategy IV.

\section{Computing Generalized Representative Sets}\label{section:disrep}

We now prove the correctness of Theorem \ref{theorem:disrep} (see Section \ref{section:strategies}). 

\begin{proof}
We first give the definition of a data structure necessary to our computation of representative sets. Let $E'$ be a universe of $n'$ elements, and suppose that $k',p'\in \mathbb{N}$ are parameters such that $p'\leq k'$. Rephrasing Definition \ref{dfn:universal_set}, we say that a family ${\cal F}\subseteq 2^{E'}$ is {\em $(E',k',p')$-good} if it satisfies the following condition: For every pair of sets $X\subseteq E'$ of size $p'$ and $Y\subseteq E'\setminus X$ of size at most $k'-p'$, there is a set $F\in{\cal F}$ such that $X\subseteq F$, and $Y\cap F=\emptyset$. An {\em $(E',k',p')$-separator} is a data structure containing such a family ${\cal F}$, which, given a set $S\subseteq E'$ of size $p'$, outputs the subfamily of sets in $\cal F$ that contain $S$, i.e., $\chi(S)=\{F\in{\cal F}: S\subseteq F\}$. The efficiency of such a data structure is measured by the parameters: $\zeta=\zeta(E',k',p')$, the number of sets in the family $\cal F$; $\tau_I=\tau_I(E',k',p')$, the time required to compute the family ${\cal F}$; and $\tau_Q=\tau_Q(E',k',p')$, an upper bound for the time required to output $\chi(X)$, for any $X\subseteq E'$ of size $p'$.

For any fixed $c'\geq 1$, the papers \cite{productFam,esarepresentative} show how to construct an $(E',k',p')$-separator,
that we call $D^{c'}_{(E',k',p')}$, for which $\displaystyle{\zeta\leq \frac{(c'k')^{k'}}{{p'}^{p'}(c'k'-p')^{k'-p'}}2^{o(k')}\log n'}$, $\displaystyle{\tau_I\leq \frac{(c'k')^{k'}}{{p'}^{p'}(c'k'-p')^{k'-p'}}}$ $\displaystyle{2^{o(k')}n'\log n'}$ and $\displaystyle{\tau_Q\leq (\frac{c'k'}{c'k'-p'})^{k'-p'}2^{o(k')}\log n'}$.

Let $E=\bigcup_{i=1}^tE_i$, $k=\sum_{i=1}^tk_i$ and $p=\sum_{i=1}^tp_i$. We assume that $|{\cal S}|$ is larger than the size of the desired representative family (else we can simply return $\cal S$). The pseudocode of our (deterministic) algorithm, \alg{GenRepAlg}, for computing representative sets is given below. The crux of its approach, which allows it, under certain ``separateness conditions'' (i.e., for certain parameters $p_1,p_2,\ldots,p_t$ and $k_1,k_2,\ldots,k_t$), to compute representative sets faster than \cite{esarepresentative} is the following: It avoids a direct construction of an $(E,k,p)$-separator, but obtains, more efficiently, all the {\em necessary} information provided by such a separator by using smaller separators of the form $D^{c_i}_{(E_i,k_i,p_i)}$. To this end, it first constructs the smaller separators (Step \ref{steprs:seperators}). It uses them to compute a family $\cal F$ (Step \ref{steprs:family}), and a corresponding function $\chi$ (Steps \ref{steprs:chi1} and \ref{steprs:chi2}). Note that $\cal F$ might {\em not} be $(E,k,p)$-good, but it will be sufficient for our purpose. Then, \alg{GenRepAlg} orders the sets in ${\cal S}$ according to their weights (Step \ref{steprs:order}). Finally, it returns all $S_i\in{\cal S}$ for which there is a set $F\in{\cal F}$ containing $S_i$ but no $S_j$, for $1\leq j<i$ (Steps \ref{steprs:Shatbegin}--\ref{steprs:Shatend}).


Observe that the size of $\cal F$ is exactly $\displaystyle{\prod_{i=1}^t|F_i|}$, which, by the properties of the separators $D^{c_i}_{(E_i,k_i,p_i)}$, is bounded by $\displaystyle{\prod_{i=1}^t \left( \frac{(c_ik_i)^{k_i}}{{p_i}^{p_i}(ck_i-p_i)^{k_i-p_i}}2^{o(k_i)}\log |E_i| \right)}$. We insert a set to $\widehat{\cal S}$ only if there exists an indicator of the form $z_F$ that is turned off, and afterwards, at least one such indicator is turned on (permanently). Therefore, the returned family $\widehat{\cal S}$ is of the desired size.

\begin{algorithm}[t!]
\caption{\alg{GenRepAlg}($c_1,\!c_2,\!...,\!c_t,\ E_1,\!E_2,\!...,\!E_t,\ p_1,\!p_2,\!...,\!p_t,\ {\cal S},\ w,\ k_1,\!k_2,\!...,\!k_t$)}
\begin{algorithmic}[1]

\FOR{$i=1,2,\ldots,t$}\label{steprs:firstloop}
	\STATE\label{steprs:seperators} construct $D^{c_i}_{(E_i,k_i,p_i)}$, and let ${\cal F}_i$ be the family it stores.
	\STATE\label{steprs:chi1} {\bf for all} $S\!\in\!{\cal S}$ {\bf do} use $D^{c_i}_{(E_i,k_i,p_i)}$to compute $\chi_i(S)\!=\!\{F_i\!\in\! {\cal F}_i\!: S\!\cap\! E_i\!\subseteq\! F_i\}$. {\bf end~for}
\ENDFOR

\STATE\label{steprs:family} compute ${\cal F}=\{F_1\cup F_2\cup\!...\cup F_t: F_1\in{\cal F}_1, F_2\in{\cal F}_2,\!..., F_t\in{\cal F}_t\}$.

\STATE\label{steprs:chi2} {\bf for all} $S\!\in\!{\cal S}$ {\bf do} compute $\chi(S)\!=\!\{F\!\in\! {\cal F}\!:(\forall i\!\in\!\{1,2,\!...,t\}\!: F\!\cap\! E_i\in\! \chi_i(S))\}$. {\bf end~for}

\STATE\label{steprs:order} order ${\cal S}=\{S_1,\!...,S_{|{\cal S}|}\}$ s.t.~$w(S_{i-1})\geq w(S_i)$ ($w(S_{i-1})\leq w(S_i)$), for all $2\leq i\leq |{\cal S}|$.

\STATE\label{steprs:Shatbegin} initialize $\widehat{\cal S}\Leftarrow\emptyset$.

\STATE {\bf for all} $F\in{\cal F}$ {\bf do} let $z_F\in\{0,1\}$ be an indicator variable for using $F$, initialized to 0. {\bf end for}

\FOR{$i=1,2,\ldots,|{\cal S}|$}
	\STATE compute $\chi_0(S_i)=\{F\in \chi(S_i): z_F=0\}$.
	\IF{$\chi_0(S_i)\neq\emptyset$}
		\STATE add $S_i$ to $\widehat{\cal S}$.
		\STATE {\bf for all} $F\in \chi_0(S_i)$ {\bf do} assign $z_F=1$. {\bf end for}
	\ENDIF
\ENDFOR

\STATE\label{steprs:Shatend} return $\widehat{\cal S}$.
\end{algorithmic}
\end{algorithm}

By the properties of the separators $D^{c_i}_{(E_i,k_i,p_i)}$, the time complexity of Step \ref{steprs:firstloop} is bounded by:\\ $\displaystyle{O(\sum_{i=1}^t\left(\!\frac{(c_ik_i)^{k_i}}{{p_i}^{p_i}(ck_i-p_i)^{k_i-p_i}}2^{o(k_i)}|E_i|\log |E_i| + |{\cal S}|(\frac{c_ik_i}{c_ik_i-p_i})^{k_i-p_i}2^{o(k_i)}\log |E_i|\!\right))}$, and of Steps \ref{steprs:family}--\ref{steprs:order} are bounded by $O(|{\cal F}|)$, $\displaystyle{O(|{\cal S}|\!\prod_{i=1}^t\!\!\left(\!\!(\frac{c_ik_i}{c_ik_i\!-\!p_i})^{k_i-p_i}2^{o(k_i)}\log |E_i|\!\!\right)\!)}$ and $O(|{\cal S}|\log|{\cal S}|)$, respectively. Moreover, the time complexity of the computation in Steps  \ref{steprs:Shatbegin}--\ref{steprs:Shatend} is bounded by $\displaystyle{O(|{\cal S}|\prod_{i=1}^t\!\!\left(\!\!(\frac{c_ik_i}{c_ik_i-p_i})^{k_i-p_i}2^{o(k_i)}\log |E_i|\!\right)\!)}$. We thus conclude that \alg{GenRepAlg} runs in the desired time.

It remains to show that $\widehat{\cal S}$ max (min) $(k_1-p_1,k_2-p_2,\ldots,k_t-p_t)$-represents $\cal S$. Consider any sets $X\in{\cal S}$ and $Y\subseteq E\setminus X$ such that $[\forall i\in\{1,2,\ldots,t\}: |Y\cap E_i|\leq k_i-p_i]$. We need to prove that there is a set $\widehat{X}\in\widehat{\cal S}$ disjoint from $Y$ such that $w(\widehat{X})\geq w(X)$ $(w(\widehat{X})\leq w(X))$. If $X\in\widehat{\cal S}$, then $X$ is the desired set, and thus we next assume that this is not the case. By the properties of the separators $D^{c_i}_{(E_i,k_i,p_i)}$, for all $i\in\{1,2,\ldots,t\}$, there is a set $F_i\in{\cal F}_i$ such that $X\cap E_i\subseteq F_i$ and $Y\cap F_i = (Y\cap E_i)\cap F_i = \emptyset$. Therefore, by our definition of ${\cal F}$ and $\chi(X)$, there is a set $F\in{\cal F}$ such that $X\cap E\subseteq F$ and $Y\cap F = \emptyset$. By the pseudocode, when we reach $X$ we do not insert it to ${\cal S}$ since we have already inserted at least one other set $X'$ that is ordered before $X$ and satisfies $F\in\chi(X')$. This set $X'$ is the desired set $\widehat{X}$.\qed
\end{proof}

\bibliographystyle{splncs03}
\bibliography{References}

\newpage

\appendix

\section{Solving the $k$-IOB Problem}\label{section:kiob}
In this section we solve {\sc $k$-IOB}, following the first two strategies in Section~\ref{section:strategies}. 

\subsection{From Out-Branchings to Small Out-Trees and Paths on 2 Nodes}\label{section:kiobred}

We first define the {\sc $(k_a,k_b,\ell_a,\ell_b)$-Internal Out-Tree ($(k_a,k_b,\ell_a,\ell_b)$-IOT)} problem, which requires finding an out-tree rather than an out-branching. Given a directed graph $G=(V,E)$, a node $r\in V$ that is the root of an out-branching of $G$, and parameters $k_a\leq k_b$ and $\ell_a\leq\ell_b$, this problem asks if $G$ has an out- tree $T$ rooted at $r$ that contains exactly $x$ internal nodes and $y$ leaves, for some $k_a\leq x\leq k_b$ and $\ell_a\leq y\leq\ell_b$. The following observations are implicit in \cite{kIOB49k}:

\begin{obs}\label{obs:toklIOT}
If {\sc $(k,k,1,k)$-IOT} can be solved in (deterministic) time $\tau(G,k)$, then {\sc $k$-IOB} can be solved in (det.) time $O(|V|(|E|+\tau(G,k)))$.
\end{obs}

\begin{obs}\label{obs:toklIOT2}
An input $(G,r,k,k,1,k)$ is a yes-instance of {\sc $(k,k,1,k)$-IOT} {\em iff} $(G,r,k,|V|,1,|V|)$ is a yes-instance of {\sc $(k,|V|,1,|V|)$-IOT}.
\end{obs}

{\noindent To show that we can focus on finding a {\em small} out-tree along with paths~on~2~nodes, we define the {\sc $(k,\ell,q)$-Tree\&Paths ($(k,\ell,q)$-T\&P)} problem: Given a directed graph $G=(V,E)$, a node $r\in V$ that is the root of an out-branching of $G$, and parameters $k$, $\ell\leq k$ and $q\geq \max\{0,2\ell-k\}$, it asks if $G$ has an out-tree $T$ rooted at $r$ that contains exactly $(k-q)$ internal nodes and $(\ell-q)$ leaves, along with $q$ (node-)disjoint paths, each on 2 nodes, that do not contain any node from $T$.}

Using Observations \ref{obs:toklIOT} and \ref{obs:toklIOT2}, we prove the following lemma.

\begin{lemma}\label{lemma:reductionproof}
If {\sc $(k,k,\ell,\ell)$-IOT} is solvable in (det.) time $\alpha(G,k,\ell)$ and {\sc $(k,\ell,q)$-T\&P} is solvable in (det.) time $\beta(G,k,\ell,q)$, then {\sc $k$-IOB} is solvable in (det.) time
$\displaystyle{O(|V|\left(|E|+\sum_{\ell=1}^k\min\left\{\alpha(G,k,\ell), \sum_{q=\max\{0,2\ell-k\}}^{\ell}\beta(G,k-q,\ell-q,q)\right\}\right))}$.
\end{lemma}

\begin{proof}
Let \alg{IOT1Alg} and \alg{AlgT\&P} be algorithms that solve {\sc $(k,k,\ell,\ell)$-IOT} in (rand./det.) time $\alpha(G,k,\ell)$ and {\sc $(k,\ell,q)$-T\&P} in (rand./det.) time $\beta(G,k,\ell,q)$, respectively. Then, we solve {\sc $(k,k,1,k)$-IOT} by using the algorithm \alg{IOT2Alg}, whose pseudocode is given below. \alg{IOT2Alg} considers every choice for $\ell$, the number of leaves in the tree it should attempt to find (Step 1). Then, for a given $\ell$, it decides (in Step 2) whether is it better (in terms of running time) to use \alg{IOT1Alg} (Step 3) or \alg{AlgT\&P} (Steps 5-7). In case it decides to use \alg{AlgT\&P}, it considers every choice for $q$, the number of disjoint paths that should be found (Step 5).

\begin{algorithm}[!ht]
\caption{\alg{IOT2Alg}($G=(V,E),r,k,k,1,k$)}
\begin{algorithmic}[1]
\FOR{$\ell=1,2,\ldots,k$}
	\IF{$\displaystyle{\alpha(G,k,\ell)\leq \sum_{q=\max\{0,2\ell-k\}}^{\ell}\beta(G,k,\ell,q)}$}
		\STATE\label{step:iot1} {\bf if} \alg{IOT1Alg}$(G,r,k,k,\ell,\ell)$ accepts {\bf then} accept. {\bf end if}
	\ELSE
		\FOR{$q=\max\{0,2\ell-k\},\max\{0,2\ell-k\}+1,\ldots,\ell$}
			\STATE\label{step:tp} {\bf if} \alg{T\&PAlg}$(G,r,k,\ell,q)$ accepts {\bf then} accept. {\bf end if}
		\ENDFOR
	\ENDIF
\ENDFOR
\STATE reject.
\end{algorithmic}
\end{algorithm}

{\noindent \alg{IOT2Alg} clearly runs in time $\displaystyle{O(\sum_{\ell=1}^k\min\!\left\{\alpha(G,k,\ell), \sum_{q=\max\{0,2\ell-k\}}^{\ell}\!\!\!\!\!\!\!\!\!\!\!\beta(G,k\!-\!q,\ell\!-\!q,q)\right\})}$. Thus, according to Observation \ref{obs:toklIOT}, it is enough to prove its correctness.}

\myparagraph{First Direction} First, suppose that ($G=(V,E),r,k,k,1,k$) is a yes-instance of {\sc $(k,k,1,k)$-IOT}. Let $T=(V_T,E_T)$ be an out-tree of {\em minimal number of leaves} among those that are solutions for this instance, and let $\ell$ denote the number of leaves in $T$. By our choice of $T$, it does not contain a leaf whose father has another child. Clearly, $T$ is a solution for the instance $(G,r,k,k,\ell,\ell)$ of {\sc $(k,k,\ell,\ell)$-IOT}. Therefore, if $\displaystyle{\alpha(G,k,\ell)\leq \sum_{q=\max\{0,2\ell-k\}}^{\ell}\beta(G,k-q,\ell-q,q)}$, \alg{IOT2Alg} accepts at Step \ref{step:iot1}. Next suppose that this condition is not fulfilled.

Now, as long as there is an internal node $v$ in $T$ that has at least two children, and at least one of them, $u$, has exactly one child and this child is a leaf, remove the edge $(v,u)$ from $T$. Denote the resulting out-tree by $T'$ and the resulting paths on 2 nodes by $P_1,P_2,\ldots,P_q$. An illustration is given in Fig.~\ref{fig:kIOBRed1}. Clearly, $q\leq\ell$ (since every resulting path contains a node that was a leaf in $T$, along with its father). By our construction of $T'$, it is an out-tree rooted at $r$ that contains exactly $(k-q)$ internal nodes and $(\ell-q)$ leaves. Furthermore, our choice of $T$ implies that both the father and grandfather of any leaf in $T'$ do not have more than one child. Indeed, if there is a leaf in $T'$ whose father has another child, such a situation should have existed in $T$ (which is a contradiction) since every leaf in $T'$ was also a leaf in $T$, and if there is a leaf in $T'$ whose grandfather has at least two children, the above process should not have stopped. Therefore, the number of leaves in $T'$ is at most half its number of internal nodes. We get that $(\ell-q)\leq (k-q)/2$, and thus $2\ell-k\leq q$. Therefore, \alg{IOT2Alg} accepts at Step \ref{step:tp}.

\begin{figure}[ht!]
\centering
\includegraphics[scale=0.6]{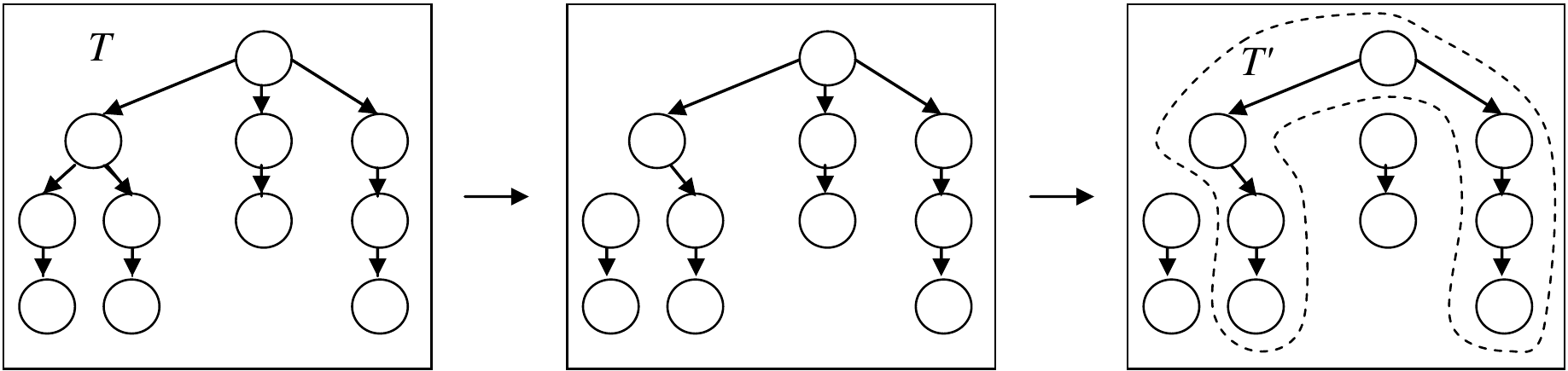}
\caption{The first direction of the proof of Lemma \ref{lemma:reductionproof}.}
\label{fig:kIOBRed1}
\end{figure}

\myparagraph{Second Direction} Now, suppose that \alg{IOT2Alg} accepts. If it accepts in Step \ref{step:iot1}, then this is clearly correct. Thus, we next suppose that it accepts in Step \ref{step:tp}, and denote by $\ell$ and $q$ the corresponding parameters. Let $T'$ and $P_1,P_2,\ldots,P_q$ be the disjoint out-tree and paths on 2 nodes, respectively, that form a solution for the instance $(G,r,k,\ell,q)$ of {\sc $(k,\ell,q)$-T\&P}. Since $T'$ is rooted at $r$ and there is an out-branching of $G$ that is rooted at $r$, we have that there is an out-branching of $G$ rooted at $r$ that {\em extends} $T'$ (i.e., the edges of $T'$ are also edges in this out-branching). Thus, we can denote by ${\cal T}$ the set of out-branchings of $G$ rooted at $r$ that extend $T'$ and have a maximum number of internal nodes, and let $T$ be an out-branching in $\cal T$ that, among the out-branchings in $\cal T$, contains a maximum number of paths from $\{P_1,P_2,\ldots,P_q\}$. Suppose, by way of contradiction, that $T$ contains less than $k$ internal nodes. Since $T'$ contains $(k-q)$ internal nodes, there is a path $(v\rightarrow u)\in\{P_1,P_2,\ldots,P_q\}$ such that both of $v$ and $u$ are leaves in $T$. By removing the edge incident to $u$ from $T$, and then inserting the edge $(v,u)$ to the result (see Fig.~\ref{fig:kIOBRed2}), we obtain an out-branching that has at least as many internal nodes as $T$, and more paths from $\{P_1,P_2,\ldots,P_q\}$ than $T$, which contradicts our choice of $T$. Therefore, $T$ is a solution to the instance $(G,r,k,|V|,1,|V|)$ of {\sc $(k,|V|,1,|V|)$-IOT}. By Observation \ref{obs:toklIOT2}, we conclude that $(G,r,k,k,1,k)$ is a yes-instance of {\sc $(k,k,1,k)$-IOT}.\qed
\end{proof}

\begin{figure}[ht!]
\centering
\includegraphics[scale=0.6]{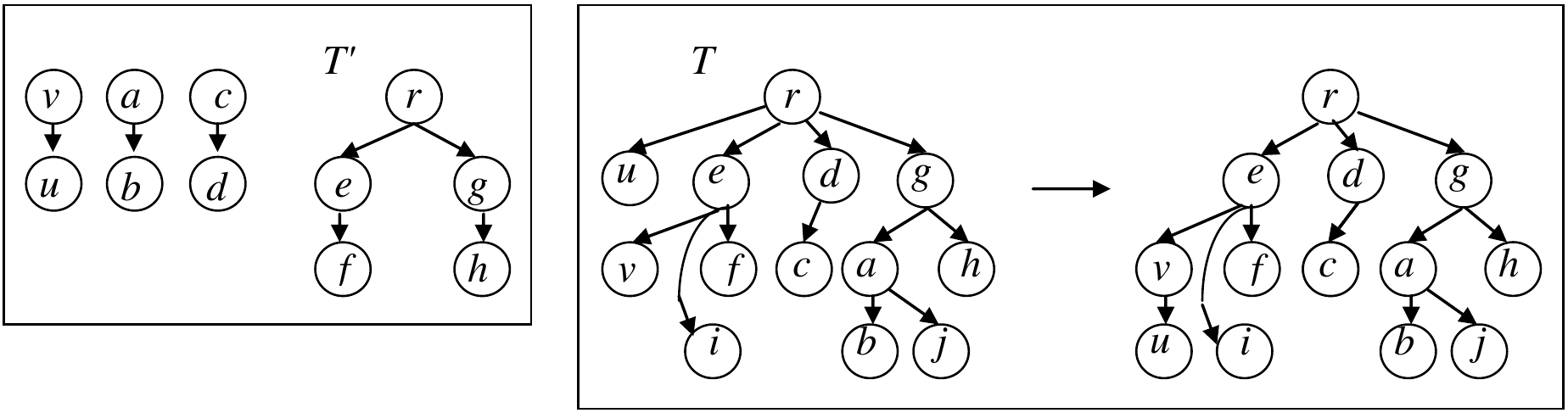}
\caption{The second direction of the proof of Lemma \ref{lemma:reductionproof}, where $k=6$, $\ell=5$ and $q=3$. It is possible that $d$ is the father of $c$ (in $T$), since $G$ may contain both $(c,d)$ and $(d,c)$. In fact, there might be no out-branching where $c$ is the father of $d$ (and thus we cannot reattach $d$ as a child of $c$). In the case of $v$ and $u$, the reattachment is possible (i.e., we obtain a {\em different out-branching}) because {\em both} $v$ and $u$ are leaves.}
\label{fig:kIOBRed2}
\end{figure}

{\noindent To obtain our deterministic algorithm for {\sc $k$-IOB}, we only need the following corollary:}

\begin{cor}\label{cor:reductionproof}
If {\sc $(k,\ell,q)$-T\&P} is solvable in (det.) time $\beta(G,k,\ell,q)$, then {\sc $k$-IOB} is solvable in (det.) time
$\displaystyle{O(|V|(|E|+\!\sum_{\ell=1}^k\sum_{q=\max\{0,2\ell-k\}}^{\ell}\!\!\!\!\!\!\!\!\beta(G,k\!-\!q,\ell\!-\!q,q)))}$.
\end{cor}

\subsection{A Deterministic Algorithm for $k$-IOB}\label{section:kiobdet}

Let $(G,r,k,\ell,q)$ be an instance of {\sc $(k,\ell,q)$-T\&P}. Moreover, let ${\cal T}_{x,y}$ denote the family that includes every set of nodes that is the node-set of an out-tree in $G$ that is rooted at $r$ and contains exactly $x$ internal nodes and $y$ leaves. The paper \cite{esarepresentative} gives a (deterministic) respresentative sets-based procedure, \alg{TreeAlg1}$(G,r,k,\ell,q,c)$, which satisfies the following:

\begin{lemma}
For any fixed $c\geq 1$ and $\epsilon>0$, \alg{TreeAlg1} computes a family
that $z$-represents ${\cal T}_{x,y}$ in time
$\displaystyle{O^*(\sum_{i=0}^{(x+y)}\left(\frac{(c(x+y+z))^{(2x+2y+2z-i)}}{i^i(c(x+y+z)-i)^{(2x+2y+2z-2i)}}2^{\epsilon(x+y+z)}\right))}$.
\end{lemma}

{\noindent Denoting $x=(k-q), y=(\ell-q), z=2q$ and ${\cal T}={\cal T}_{x,y}$, we have the following:}

\begin{cor}\label{cor:treealg}
For any fixed $c\geq 1$ and $\epsilon>0$, \alg{TreeAlg1} computes a family
that $(2q)$-represents ${\cal T}$ in time
$\displaystyle{O^*(\sum_{i=0}^{(k+\ell-2q)}\left(\frac{(c(k+\ell))^{(2k+2\ell-i)}}{i^i(c(k+\ell)-i)^{2k+2\ell-2i}}2^{\epsilon(k+\ell)}\right))}$.
\end{cor}

{\noindent The algorithm \alg{T\&PAlg}, whose pseudocode is given below, uses \alg{TreeAlg1} to compte a family $\widehat{\cal T}$ that represents ${\cal T}$ (Step 1). Then, it attempts to find a set of $q$ paths on 2 nodes that are disjoint from some tree in $\widehat{\cal T}$ via computations of maximum matchings (Steps 2 and 3); if it succeeds, it accepts (Step 4). For \alg{T\&PAlg}, we show:}

\begin{lemma}\label{lemma:tapc}
\alg{T\&PAlg} solves {\sc $(k,\!\ell,\!q)$-T\&P} in time $\displaystyle{O^*\!(\!\!\!\sum_{i=0}^{k\!+\!\ell\!-\!2q}\!\!\left(\!\!\frac{(c(k\!+\!\ell))^{2k\!+\!2\ell\!-\!i} \!\cdot\!2^{\frac{(k\!+\!\ell)}{10^{10}}}}{i^i(c(k\!+\!\ell)\!-\!i)^{2k\!+\!2\ell\!-\!2i}}\!\!\right)\!)}$.
\end{lemma}

\begin{proof}
By the pseudocode and Corollary \ref{cor:treealg}, \alg{T\&PAlg} runs in the desired time. Moreover, if it accepts, the input is clearly a yes-instance. Now, suppose that the input is a yes-instance, and let $T$ and $P_1,P_2,\ldots,P_q$ be a solution for it. By Corollary \ref{cor:treealg}, $\widehat{\cal T}$ $(2q)$-represents ${\cal T}$. Therefore, there exists $X\in\widehat{\cal T}$ that does not contain any node that belongs to $P_1,P_2,\ldots,P_q$. Thus, the underlying undirected graph of $G$ from which we remove the nodes in $X$ and the adjacent edges contains a maximum matching of size at least $q$. By the pseudocode, this implies that \alg{T\&PAlg} accepts.\qed
\end{proof}

\begin{algorithm}[h]
\caption{\alg{T\&PAlg}($G=(V,E),r,k,\ell,q,c$)}
\begin{algorithmic}[1]
\STATE $\widehat{{\cal T}}\Leftarrow$ \alg{TreeAlg1}$(G,r,k,\ell,q,c)$.
\FORALL{$X\in \widehat{{\cal T}}$}
	\STATE compute in polynomial time a maximum matching $M$ in the underlying undirected graph of the subgraph of $G$ induced by $V\setminus X$ (see \cite{edmondsmatch}). 
	\STATE{\bf if} $|M|\geq q$ {\bf then} accept. {\bf end if}
\ENDFOR
\STATE reject.
\end{algorithmic}
\end{algorithm}

{\noindent Upper bounds for $\displaystyle{O^*(\sum_{\ell=1}^k\sum_{q=\max\{0,2\ell-k\}}^{\ell}\sum_{i=0}^{(k+\ell-2q)}\left(\frac{(c(k+\ell))^{(2k+2\ell-i)} \cdot 2^{\frac{(k+\ell)}{10^{10}}}}{i^i(c(k+\ell)-i)^{2k+2\ell-2i}}\right))}$, given different parameters $1\leq c$, appear separately in Appendix \ref{section:kiobdettime}. By choosing $c=1.497$, we get the bound $O^*(5.139^k)$. Thus, by Corollary \ref{cor:reductionproof}, we have that Theorem \ref{theorem:kiobDet} (in Section \ref{section:strategies}) is correct.}

\subsection{A Randomized Algorithm for $k$-IOB}\label{section:kiobrand}

The papers \cite{thesis11,ipec13} give a (randomized) narrow sieves-based procedure, \alg{TreeAlg2}, which satisfies the following lemma:

\begin{lemma}\label{lemma:randkiobknown}
\alg{TreeAlg2} solves {\sc $(k,k,\ell,\ell)$-IOT} in time $O^*(2^{k+\ell})$.
\end{lemma}

{\noindent Note that $O^*(\sum_{\ell=1}^{0.8545k}2^{k+\ell}) =O^*(3.617^k)$. Also, we show (in Appendix \ref{section:kiobdettime}) that $\displaystyle{O^*(\!\!\sum_{\ell=0.8545k}^k\sum_{q=\max\{0,2\ell-k\}}^{\ell}\!\sum_{i=0}^{(k+\ell-2q)}\!\!\left(\!\frac{(1.765(k\!+\!\ell))^{(2k+2\ell-i)}\cdot2^{\frac{(k+\ell)}{10^{10}}}}{i^i(1.765(k\!+\!\ell)-i)^{2k+2\ell-2i}}\!\right))} = O^*(3.617^k)$. Thus, by Lemmas \ref{lemma:reductionproof}, \ref{lemma:tapc} (setting $c=1.765$) and \ref{lemma:randkiobknown}, we have that Theorem \ref{theorem:kiobRand} (in Section \ref{section:strategies}) is correct.}

\subsection{The Running Time of \alg{T\&PAlg}}\label{section:kiobdettime}

First, note that $(*)=\displaystyle{O^*(\sum_{\ell=\ell^*}^k\sum_{q=\max\{0,2\ell-k\}}^{\ell}\!\sum_{i=0}^{(k+\ell-2q)}\!\!\left(\!\frac{(c(k\!+\!\ell))^{(2k+2\ell-i)}}{i^i(c(k\!+\!\ell)\!-\!i)^{(2k+2\ell-2i)}}2^{\frac{(k+\ell)}{10^{10}}}\!\right))}$ is bounded by:\\ $\displaystyle{O^*(4^{\frac{k}{10^{10}}}\sum_{\ell=\ell^*}^k\sum_{i=0}^{\min\{k+\ell,3(k-\ell)\}}\left(\frac{(c(k+\ell))^{(2k+2\ell-i)}}{i^i(c(k+\ell)-i)^{(2k+2\ell-2i)}}\right))}$, which we can further bound by:\\$\displaystyle{O^*(4^{\frac{k}{10^{10}}}\max_{\frac{\ell^*}{k}\leq\beta\leq 1}\left(\max_{0\leq i\leq \min\{(1+\beta)k,3(1-\beta)k\}}\left(\frac{(c(1+\beta)k)^{(2(1+\beta)k-i)}}{i^i(c(1+\beta)k-i)^{(2(1+\beta)k-2i)}}\right)\right))}$.

Let $\alpha_{c}$ be the value $\alpha$ that maximizes $\displaystyle{\max_{0\leq\alpha\leq1}\{\frac{c^{2-\alpha}}{\alpha^{\alpha}(c-\alpha)^{2-2\alpha}}\}}$. If $\alpha_c(1\!+\!\beta)k\!\leq\! 3(1\!-\!\beta)k$, then the maximum of $\displaystyle{\max_{0\leq i\leq \min\{(1+\beta)k,3(1-\beta)k\}}\!\left(\!\frac{(c(1\!+\!\beta)k)^{(2(1+\beta)k-i)}}{i^i(c(1\!+\!\beta)k\!-\!i)^{(2(1+\beta)k-2i)}}\!\right)}$ is obtained at $i\!=\!\alpha_c(1\!+\!\beta)k$, and else it is obtained at $i\!=\!3(1\!-\!\beta)k$. Thus, we can further bound (*) by $O^*$ of $4^{\frac{k}{10^{10}}}$ times the following~expression:

\smallskip

\[\begin{array}{ll}
(\max\{&\displaystyle{ \max_{\frac{\ell^*}{k}\leq\beta\leq\frac{3-\alpha_c}{3+\alpha_c}}\!\!\left(\frac{c^{2-\alpha_{c}}}{\alpha_{c}^{\alpha_{c}}(c-\alpha_{c})^{2-2\alpha_{c}}}\right)^{1+\beta}\!\!\!\!\!\!,}\\

&\displaystyle{\max_{\max\{\frac{\ell^*}{k},\frac{3-\alpha_c}{3+\alpha_c}\}<\beta\leq 1} \frac{(c(1+\beta))^{5\beta-1}}{(3(1-\beta))^{3(1-\beta)}(c(1+\beta)-3(1-\beta))^{4(2\beta-1)}}\}})^k.
\end{array}\]

\medskip

{\noindent The first part of the expression is relevant only if $\frac{\ell^*}{k}\!\leq\!\frac{3-\alpha_c}{3+\alpha_c}$, else we regard~it~as~0.~Let $\beta_c$ be the value $\beta$ that maximizes $\displaystyle{\max_{\frac{3-\alpha_c}{3+\alpha_c}<\beta\leq 1} \frac{(c(1+\beta))^{5\beta-1}}{(3(1\!-\!\beta))^{3(1\!-\!\beta)}(c(1\!+\!\beta)-3(1\!-\!\beta))^{4(2\beta\!-\!1)}}}$.}

Overall, we get that (*) is bounded by $O^*$ of $4^{\frac{k}{10^{10}}}$ times the following:

\begin{itemize}
\item If $\displaystyle{\beta_c\leq \frac{\ell^*}{k}}$: $(\displaystyle{ \frac{(c(1+\frac{\ell^*}{k}))^{5\frac{\ell^*}{k}-1}}{(3(1-\frac{\ell^*}{k}))^{3(1-\frac{\ell^*}{k})}(c(1+\frac{\ell^*}{k})-3(1\!-\!\frac{\ell^*}{k}))^{4(2\frac{\ell^*}{k}-1)}}})^k$.

\item Else if $\displaystyle{\frac{3-\alpha_c}{3+\alpha_c}\leq\frac{\ell^*}{k}}\leq \beta_c$: $(\displaystyle{ \frac{(c(1+\beta_c))^{5\beta_c-1}}{(3(1\!-\!\beta_c))^{3(1\!-\!\beta_c)}(c(1\!+\!\beta_c)-3(1\!-\!\beta_c))^{4(2\beta_c\!-\!1)}}})^k$.

\item Else: $(\displaystyle{\max\{\left(\frac{c^{2-\alpha_{c}}}{\alpha_{c}^{\alpha_{c}}(c-\alpha_{c})^{2-2\alpha_{c}}}\right)^{\frac{6}{3+\alpha_c}}\!\!\!,\  \frac{(c(1+\beta_c))^{5\beta_c-1}}{(3(1\!-\!\beta_c))^{3(1\!-\!\beta_c)}(c(1\!+\!\beta_c)-3(1\!-\!\beta_c))^{4(2\beta_c\!-\!1)}}\}})^k$.
\end{itemize}

{\noindent Table \ref{tab:alphabeta} (below) presents approximate values of $\alpha_c$ and $\beta_c$, given different choices of $c$. Then, Table \ref{tab:kiobdettime} presents bounds for (*), given different choices of $c$, where $\ell^*=1/k$. In particular, by choosing $c=1.497$, we get the bound $O^*(5.139^k)$ (used in Section \ref{section:kiobdet}). Next, Table \ref{tab:forrand} presents other bounds for (*), corresponding to different choices of $c$ and $\gamma$, where $\ell^*=\gamma k$. In particular, by choosing $c=1.765$ and $\gamma=0.8545$, we get the bound $O^*(3.617^k)$ (used in Appendix \ref{section:kiobrand}).}

\begin{table}
\centering
\begin{tabular}{|c|c|c|c|c|c|}
	\hline
	$c$ & $\alpha_c$ & $\left(\!\frac{c^{2-\alpha_{c}}}{\alpha_{c}^{\alpha_{c}}(c-\alpha_{c})^{2-2\alpha_{c}}}\!\right)^{\frac{6}{3+\alpha_c}}\!\!\!\cdot 4^{\frac{1}{10^{10}}}$ & $\frac{3-\alpha_c}{3+\alpha_c}$ & $\beta_c$ &  $\frac{(c(1+\beta_c))^{5\beta_c-1}\cdot4^{\frac{1}{10^{10}}}}{(3(1\!-\!\beta_c))^{3(1\!-\!\beta_c)}(c(1+\beta_c)-3(1\!-\!\beta_c))^{4(2\beta_c\!-\!1)}}$ \\\hline
	1       					     & 0.55013    & $\leq5.873$  & 0.69008 & 0.71350 & $\leq5.9441$    \\\hline	
	1.4	                   & 0.54908    & $\leq5.094$  & 0.69058 & 0.71582 & $\leq5.1552$    \\\hline		
	1.45	                 & 0.55302    & $\leq5.080$ & 0.68870 & 0.71441 & $\leq5.1424$    \\\hline		
	1.495	                 & 0.55692    & $\leq5.075$  & 0.68685 & 0.71299 & $\leq5.13864$    \\\hline	
	1.496	                 & 0.55701    & $\leq5.075$  & 0.68681 & 0.71296 & $\leq5.13864$    \\\hline	
	{\bf1.497}	                 & {\bf0.55710}    & $\bf\leq5.075$  & {\bf0.68677} & {\bf0.71293} & $\bf\leq5.13863$    \\\hline	
	1.498	                 & 0.55719    & $\leq5.075$  & 0.68672 & 0.71289 & $\leq5.13863$    \\\hline	
	1.499	                 & 0.55729    & $\leq5.075$  & 0.68669 & 0.71286 & $\leq5.13864$    \\\hline		
	1.5	                   & 0.55737    & $\leq5.075$  & 0.68664 & 0.71283 & $\leq5.13865$    \\\hline		
\end{tabular}\medskip
\caption{Approximate values of $\alpha_c$ and $\beta_c$, corresponding to different choices of $c$.}
\label{tab:alphabeta}
\smallskip

\begin{tabular}{|c|c|c|c|c|c|c|c|c|c|}
	\hline
	$c$   & 1      & 1.4    & 1.45   & 1.495   & 1.496    & {\bf1.497}   & 1.498   & 1.499   & 1.5 \\\hline		
	Bound & 5.9441 & 5.1552 & 5.1424 & 5.13864 & 5.13864  & {\bf5.13863} & 5.13863 & 5.13864 & 5.13865 \\\hline	
\end{tabular}\medskip
\caption{Upper bounds for (*), corresponding to different choices of $c$, where $\ell^*=1/k$. An entry that stores a constant $a$ corresponds to the bound $O^*(a^k)$.}
\label{tab:kiobdettime}
\smallskip

\begin{tabular}{|c|c|c|c|c|c|c|c|c|c|}
	\hline
	$\gamma \setminus c$ & 1.763       & 1.764       & {\bf1.765}       & 1.766       \\\hline		
	0.8544               & 3.617665566 & 3.617665007 & 3.617665035      & 3.617665648 \\\hline
	{\bf0.8545}          & 3.615894763 & 3.615894103 & {\bf3.615894029} & 3.615894539 \\\hline		
\end{tabular}\medskip
\caption{Upper bounds for (*), corresponding to different choices of $c$ and $\ell^*$, where $\ell^*=\gamma k$. An entry that stores a constant $a$ corresponds to the bound $O^*(a^k)$.}
\label{tab:forrand}
\end{table}

\section{Solving Weighted $k$-Path}\label{section:standard}

In this section, we demonstrate a strategy for speeding-up standard representative sets-based algorithms, following the explanation in Section \ref{section:strategies}. The strategy is demonstrated using the {\sc Weighted $k$-Path} problem.

\subsection{Intuition}\label{sec:kPathIntuition}

In this section, we attempt to give the intuition that guided us through the development of our algorithm for {\sc Weighted $k$-Path}. Formal definitions and proofs are given in the four following sections. The explanations in this section are far from formal, but the intuition described here significantly helps to understand the four following sections.

First, let us consider some solution, which is a path $P=(V_P,E_P)$ (of minimal weight) on $k$ nodes. Now, suppose that we have a set $E_1\subseteq V$ (recall that $V$ is the node-set of $G$) such that when we look at $P$, starting from its first node (i.e., the node that does not have an ingoing neighbor in $P$), we count exactly $(i-1)$ nodes that belong to $E_2=V\setminus E_1$, then one node that belongs to $E_1$, then $(i-1)$ nodes that belong to $E_2$, and so on. That is, only and every $i^{th}$ node that we encounter belongs to $E_1$. For the sake of clarity, the manner in which we obtain $E_1$ is discussed at a later point (in this section). Moreover, for the sake of clarity,\footnote{The assumptions that we make for the sake of clarity of this section are not made in the following sections.} we suppose that $k/i\in\mathbb{N}$ and $i=4$ (as we will explain later, we actually need a larger $i$---that is, we need that $E_1\cap V_P$ will be significantly smaller than $E_2\cap V_P$ so that a certain partitioning step will be efficient). We shall call $E_1$ the blue part of the universe (which is $V$), $E_1\cap V_P$ the blue part of the path, $E_2$ the red part of the universe, and $E_2\cap V_P$ the red part of the path. An illustration of the solution is given in Fig.~\ref{fig:kPathInt1} below.

\begin{figure}[ht!]
\centering
\includegraphics[scale=0.65]{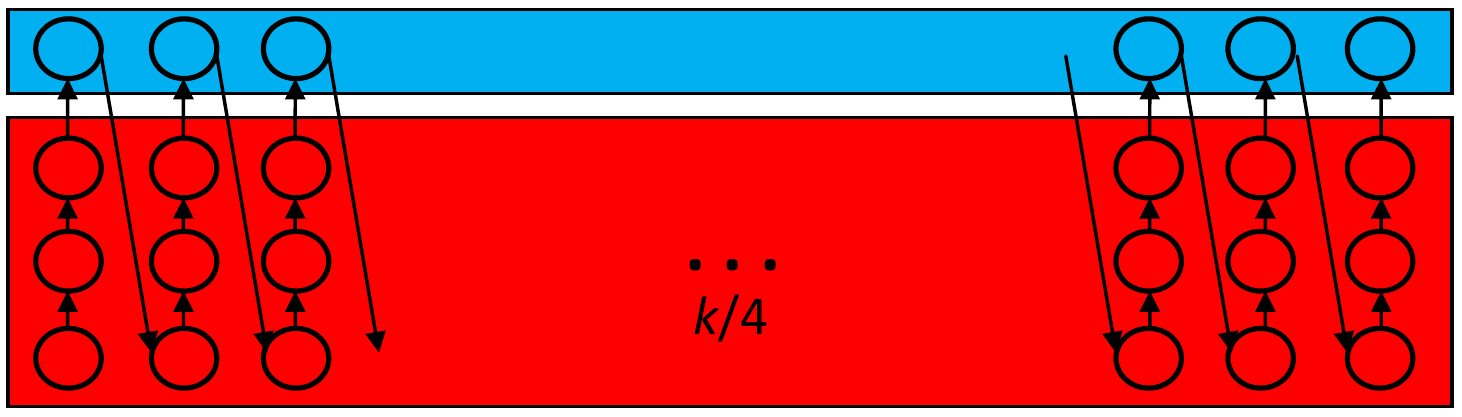}
\caption{A detailed illustration of a solution.}
\label{fig:kPathInt1}
\end{figure}

Since the nodes and edges in the illustration will be distracting, we omit them and illustrate the solution as shown in Fig.~\ref{fig:kPathInt2}.

\begin{figure}[ht!]
\centering
\includegraphics[scale=0.65]{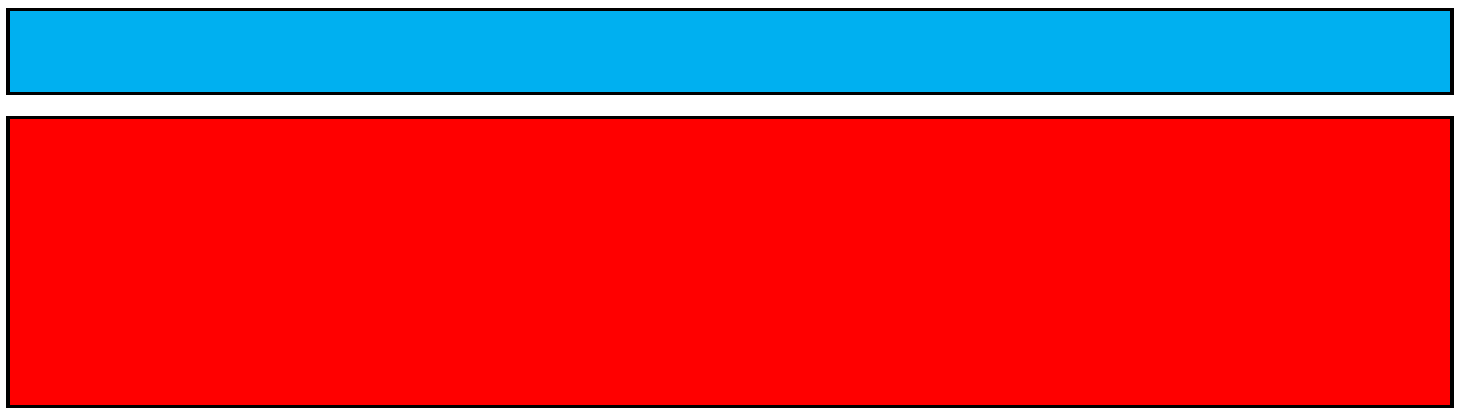}
\caption{A simplified illustration of a solution.}
\label{fig:kPathInt2}
\end{figure}

Suppose we are trying to find $P$ by executing a simple dynamic programming-based procedure, in which we embed (standard) computations of representative sets. That is, at each stage, for every node $v\in V$, we have a family $\cal S$ of partial solutions such that each of them is a path that ends at $v$, and all of these paths have the same size (between 1 and $k$, depending on the stage). We decrease the size of the family $\cal S$ by computing a family that represents it (recall, from Section \ref{section:techniques}, that this action does not prevent us from finding a solution---if there was a partial solution in $\cal S$ that could have been extended to a solution, there is also such a partial solution in the representative family). To advance to the next stage, we simply try to extend each of the remaining partial solutions, which is a path, by using every possible neighbor of its last node (that does not belong to the path, as the solution is a {\em simple} path). For the sake of clarity, suppose that we progress, at each stage, by extending the partial solutions by four nodes rather than one node. In the algorithm, this will not be possible, since $i$ (which here we assume to be four) will not be a constant\footnote{The number of options to elongate a path by using $i$ distinct nodes is potentially ${|V|\choose i}$, which does not result in a time of the form $O^*(f(k))$; therefore, we need to add the nodes one-by-one---thus, we have shorter stages, where each of them is immediately followed by computations of representative families.}---this just adds technical details that are distracting for the purpose of this section.

By the above procedure, we can illustrate a partial solution as shown in Fig.~\ref{fig:kPathInt3}. The colored part represents the partial solution, and the white part represents the location of nodes that we need to insert in order to (potentially) complete the partial solution to a solution (such as $P$).

\begin{figure}[ht!]
\centering
\includegraphics[scale=0.65]{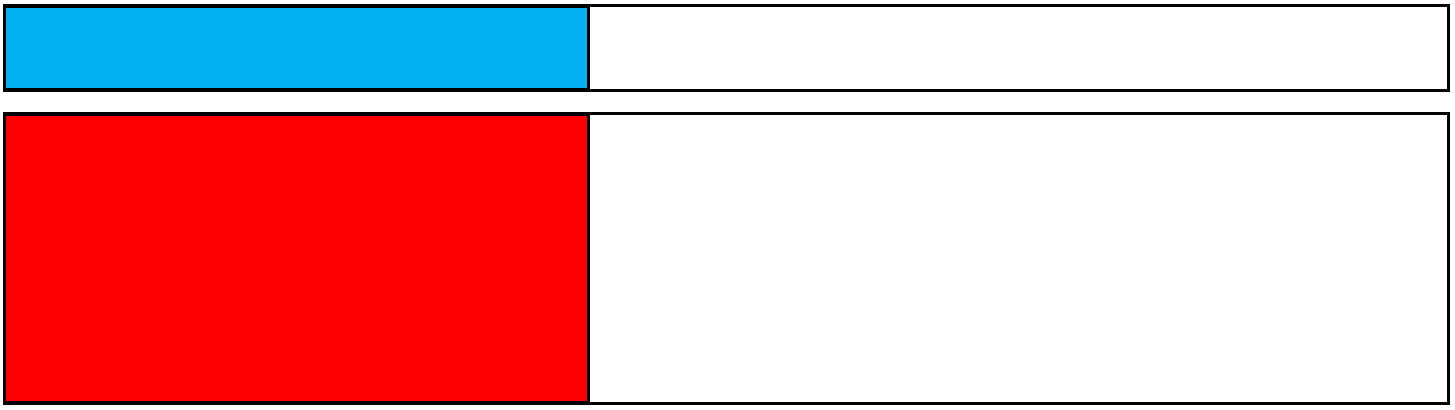}
\caption{An illustration of a partial solution.}
\label{fig:kPathInt3}
\end{figure}

Now, we consider a partition of the blue part of the universe into two sets, $L$ and $R$. For the sake of clarity, the manner in which we obtain this partition is discussed at a later point (in this section). We color $L$ in dark blue, and $R$ in light blue. We assume that $L$ ($R$) should be considered only in the first (second) half of the execution of our procedure. More precisely, we assume that if there is a solution, then there is also a solution, say $P$, such that $L$ captures the first $k/8$ elements of the blue part of $P$ and $R$ captures the remaining $k/8$ elements of the blue part of $P$ (see Fig.~\ref{fig:kPathInt4}). 

\begin{figure}[ht!]
\centering
\includegraphics[scale=0.65]{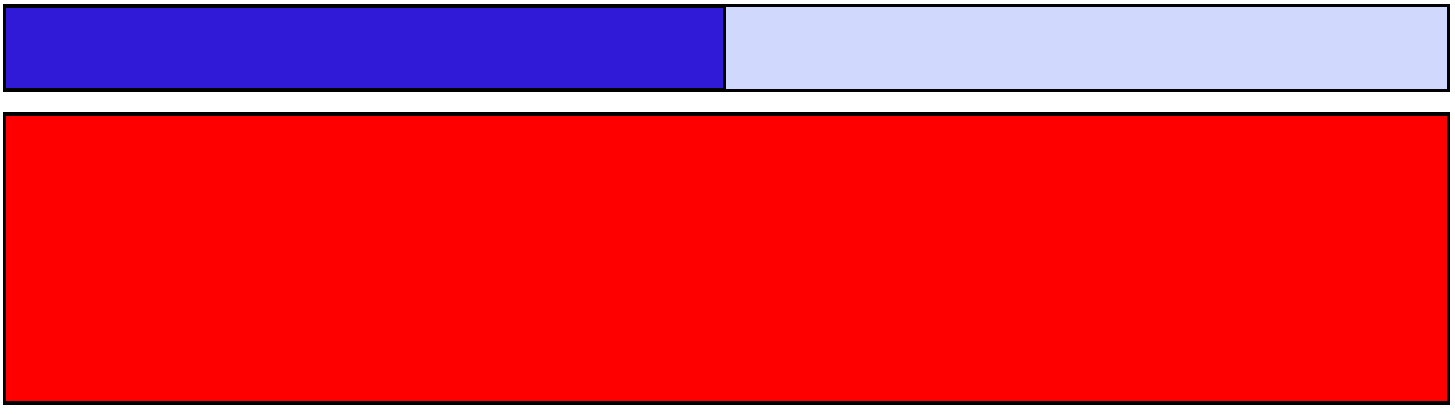}
\caption{An illustration of a solution, where the blue part is partitioned into dark and light blue parts.}
\label{fig:kPathInt4}
\end{figure}

Thus, when we are considering a partial solution in the first half of the execution, it should only contain {\em dark} blue and red elements. Moreover, we need to ensure that we have a partial solution that does not contain a certain set (using which it might be completed to a solution) of only {\em dark} blue and red elements---we do not need to ensure that it does not contain a certain set of light blue, dark blue and red elements, since we simply do not use any light blue element up to this point (see Fig.~\ref{fig:kPathInt5}). 

\begin{figure}[ht!]
\centering
\includegraphics[scale=0.65]{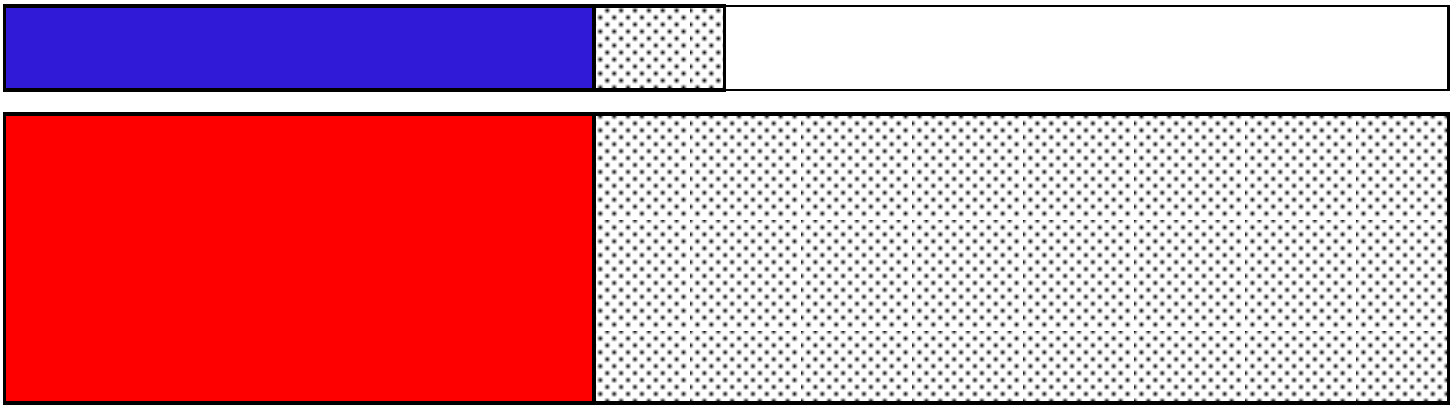}
\caption{An illustration of a partial solution; in the first half of the execution, we ensure that we have such a partial solution that does not contain the dotted parts.}
\label{fig:kPathInt5}
\end{figure}

When we reach the second half of the execution, we attempt to add to our partial solutions only {\em light} blue and red elements. Thus, we can ignore all of the dark blue elements in our partial solutions (we will not encounter them again), and we need to ensure that we have a partial solution that does not contain a certain set of only {\em light} blue and red elements (see Fig.~\ref{fig:kPathInt6}).

\begin{figure}[ht!]
\centering
\includegraphics[scale=0.65]{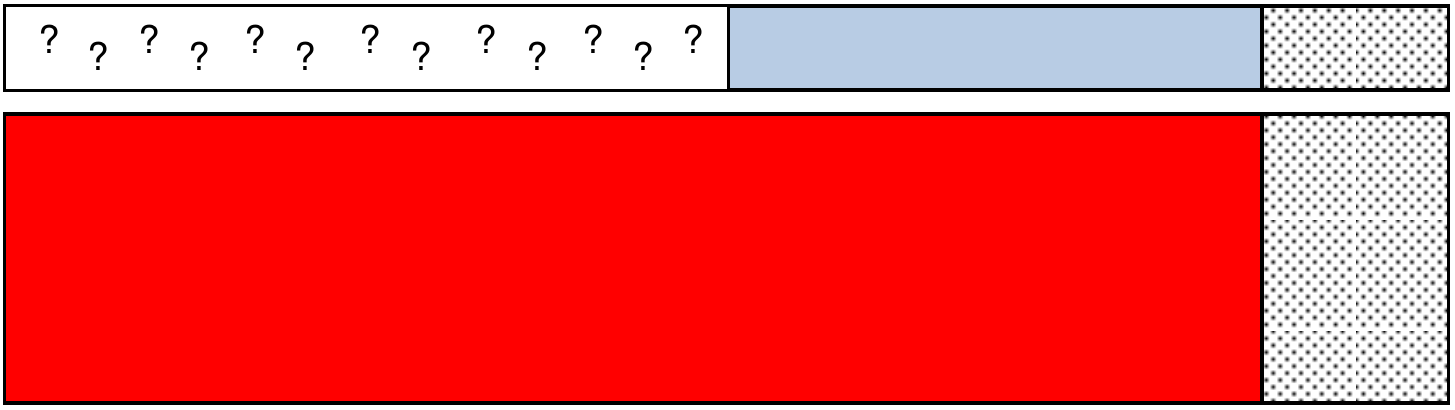}
\caption{An illustration of a partial solution; in the second half of the execution, we ignore dark blue elements, and therefore, we replaced the part containing them with question marks.}
\label{fig:kPathInt6}
\end{figure}

Therefore, we have shown that indeed $L$ ($R$) should be considered only in the first (second) half of the execution. Thus, in the first half of the execution, we can compute representative families faster, since there are less elements in the sets to which we need to ensure separateness, while in the second half of the execution, we can also compute representative families faster, since we have smaller partial solutions. However, this is not good enough---the time that we save in the computations of representative families does not justify the time that we need to obtain $L$ and $R$. Fortunately, we can save a lot more time by using {\em generalized} representative families, relying on the computation we developed in Section \ref{section:disrep}. The worst time necessary when computing a representative family for some family $\cal S$ (by \cite{productFam,esarepresentative}) occurs when each set in $\cal S$ contains slightly more than $k/2$ elements (recall that $k$ is the size of the complete solution, which in our case, is the number of nodes in the solution). We benefit from the usage of generalized representative sets, since the worst running times in the context of $L$ and $R$, actually occur at a point that is good with respect to the red part, and the worst running time in the context of the red part, actually occurs at a point that is good with respect to $L$ and $R$. More precisely, when we are looking at a partial solution that contains slightly more than $50\%$, say $55\%$, of the total number of dark blue elements in a solution, it only contains $(55/2)\%=27.5\%$ of the total number of red elements in a solution; also, when we are looking at a partial solution that contains $55\%$ of the total number of light blue elements in a solution, it already contains $50 + (55/2)\%=77.5\%$ of the total number of red elements in a solution; finally, when we are looking at a partial solution that contains $55\%$ of the total number of red elements in a solution, it already contains $100\%$ of the total number of dark blue elements in a solution, and only $2\cdot(55-50)\%=10\%$ of the total number of light blue elements in a solution. Observe that our computation of generalized representative sets is significantly faster under such distortion (see Theorem \ref{theorem:disrep}).

However, letting $L$ and $R$ have the same size still does not result in a significant enough gain from our usage of generalized representative sets. Recall that the worst time necessary when computing a representative family for some family $\cal S$ occurs when each set in $\cal S$ contains slightly {\em more} than $k/2$ elements. In our computation, the red part will be significantly larger than the blue part (otherwise computing $L$ and $R$ requires too much time); therefore, it seems reasonable that the separation between $L$ and $R$ should take place at the point that corresponds to the worst time of computing a representative family with respect to the red part---then, this worst time is not multiplied by extra time in the context of the blue part (since we will have partial solutions that contain all the necessary dark blue elements and thus we ignore these elements, and not even one light blue element). A corresponding illustration is given in Fig.~\ref{fig:kPathInt7}.

\begin{figure}[ht!]
\centering
\includegraphics[scale=0.65]{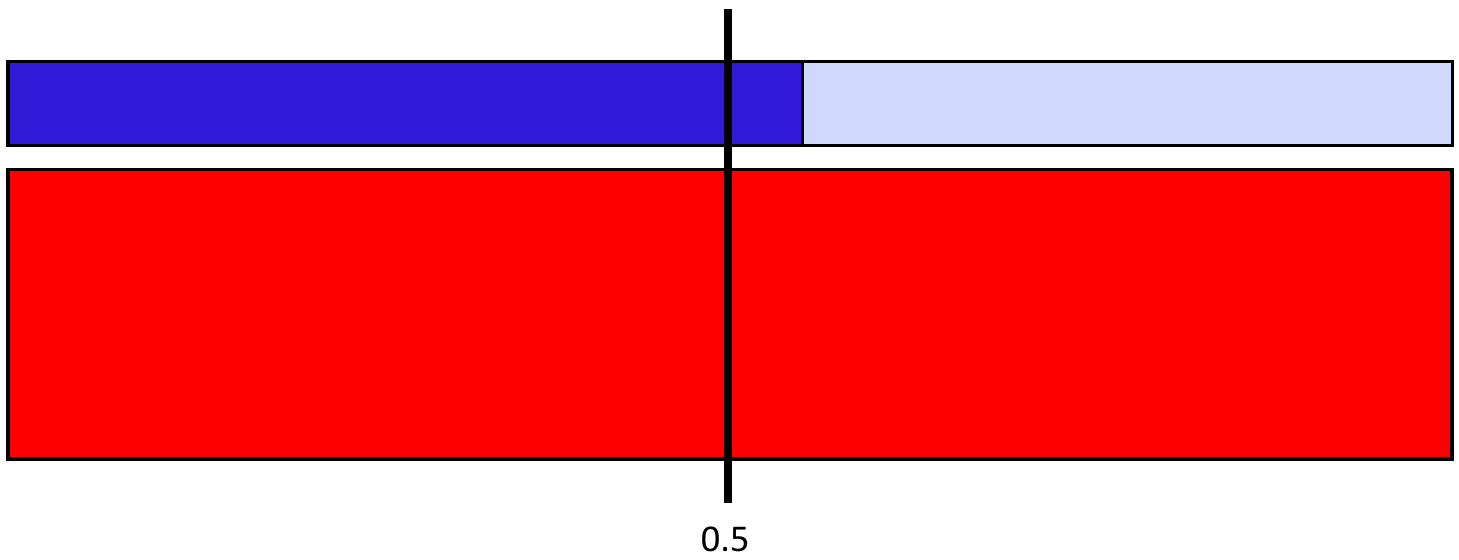}
\caption{An illustration of a solution; the number of dark blue elements that it contains is larger than the number of light blue elements that it contains.}
\label{fig:kPathInt7}
\end{figure}

In our algorithm, the choice is slightly more complicated technically---indeed, the best choice does not exactly correspond to the worst point mentioned above, although it is quite close to it; there are also some complications, in this context, that are caused by the fact that each part among $L$, $R$ and $E_2$ will have its own, different tradeoff parameter $c$ (see Theorem \ref{theorem:disrep}),\footnote{In fact, the red part will have two tradeoff parameters, one to be used in the first half of the execution, and one in the second; this causes another minor technical complication when transferring (at roughly the middle of the execution) between these parameters---to perform this efficiently, we actually consider a phase where we iterate over many intermediate parameters.} and that the shift also effects the time required to compute the sets $L$ and $R$ (though due to the fact that the shift is small, the latter effect is almost negligible). We will have a parameter, $\delta$, that controls the necessary shift. For the sake of clarity, we next ignore the shift (i.e., we suppose that $|L|=|R|$).

Now, we turn to explain how to compute the dark and light blue parts, assuming that we have the blue part. Only afterwards we shall explain how to obtain the blue part, since its computation introduces several complications, on which we will focus separately in the remaining part of this section. The dark and light blue parts can be easily computed using a standard step of divide-and-color (see Section \ref{section:techniques}). That is, we shall consider many options of a pair of a dark and light blue parts, and ensure that at least one of them will be good (i.e., captures the blue elements as described in Fig.~\ref{fig:kPathInt4}). More precisely, assuming that we have some arbitrary association between indices and elements, we compute a $(|V|,|L\cup R|,|R|)$-universal set (see Theorem \ref{theorem:splitter}); for each function, we let $L$ ($R$) contain the elements corresponding to indices to which the function assigns '0' ('1'). To ensure that the size of the universal set is not too large (which results in a slow running time), $|L\cup R|$ should be significantly smaller than $k$; we will have a parameter, $\gamma$, that controls the exact size.

The blue part that we compute can actually be very different from the nicely organized blue part that is described in Fig.~\ref{fig:kPathInt1} (since obtaining such a part by using, e.g., a universal set, is very inefficient). We simply require that the blue part should contain exactly $|L\cup R|$ elements from some solution (if one exists). This can be easily accomplished in polynomial time. Indeed, we just need to define some arbitrary order on $V$, say $v_1,v_2,\ldots,v_{|V|}$, and examine each node $v_i$, letting the blue part contain all the elements that are greater or equal to $v_i$. In this manner we will encounter at least one blue part that has the above mentioned desired property. Next, assume that we are currently considering such a blue part, $E_1$.

As mentioned above, the blue part, $E_1$, is not necessarily as nicely balanced along the solution as described in Fig.~\ref{fig:kPathInt1}, but it can be congested at certain points, and sparse in others, along the solution (since we simply defined a set that contains a desired number of elements from a solution, regardless of their location in the path that is a solution). This is very problematic in terms of the gain from using {\em generalized} representative sets that we described above. For example, looking at an extreme case (for the sake of clarity), the solution can be colored as illustrated in Fig.~\ref{fig:kPathInt8} below. Then, the set of all the most problematic points of the computation of representative sets in the context of the red part contains all the points between the two vertical red lines, including the points corresponding to the worst computation times with respect to the dark blue part (indicated by a dark blue line) and the light blue part (indicated by a light blue line). That is, we do not ensure, as before, that any point in the execution that is problematic with respect to the red part is actually good with respect to both the dark and light blue parts.

\begin{figure}[ht!]
\centering
\includegraphics[scale=0.65]{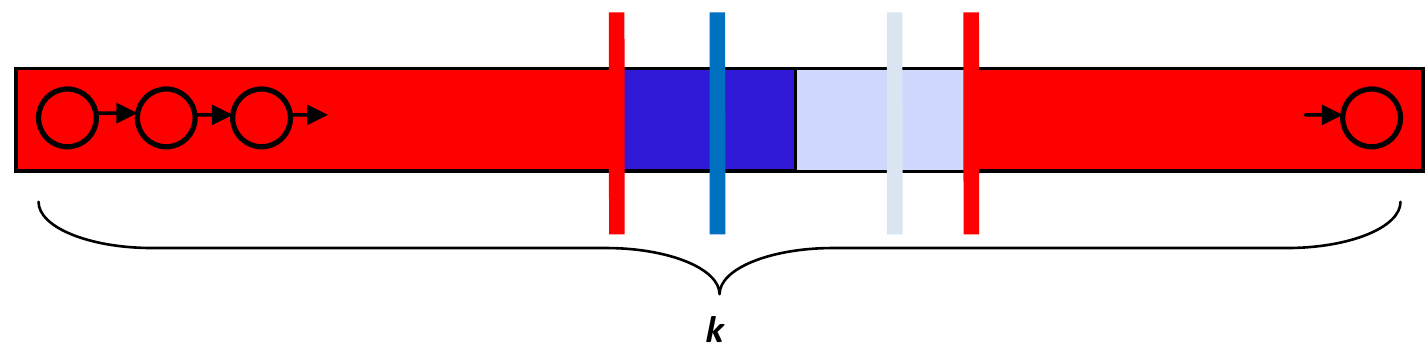}
\caption{An illustration of a solution whose blue part is congested roughly (slightly after) at its middle.}
\label{fig:kPathInt8}
\end{figure}

We do not need to ensure that the blue part will be exactly as nicely ordered as described in Fig.~\ref{fig:kPathInt1}, but we still need to ensure that congestions of the dark blue part will be as close as possible to the beginning of the solution, and congestions of the light blue part will be as close as possible to the end of the solution---thus, we ensure that we gain from our computation of {\em generalized} representative sets. Observe that the more the dark (light) blue part is congested at the beginning (end) of the solution, we will actually gain more from our computation of generalized representative sets, but we will only be able to ensure (using the method we describe next) that at worst, the blue part is {\em approximately} balanced along the solution. To present the method in which we reorder the solution, consider the arbitrary solution illustrated in Fig.~\ref{fig:kPathInt9}.

\begin{figure}[ht!]
\centering
\includegraphics[scale=0.65]{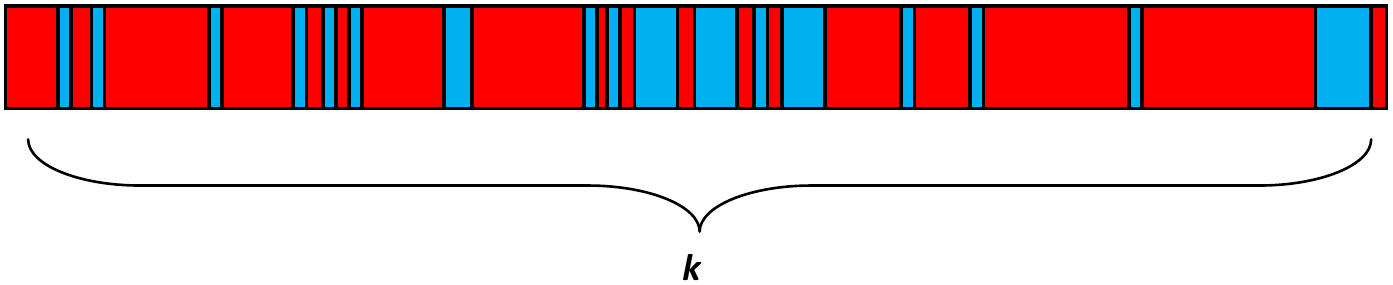}
\caption{An arbitrary, unorganized solution.}
\label{fig:kPathInt9}
\end{figure}

Now, we want to reorganize it, and therefore we cut it into a {\em constant} number of very small pieces of the same size (see Fig.~\ref{fig:kPathInt10}), apart from one piece, which may have a different size, but for the sake of clarity, we will assume in this section that such a piece does not exist. In the algorithm, we will have a parameter, $\epsilon$, that controls the size of these pieces.

\begin{figure}[ht!]
\centering
\includegraphics[scale=0.65]{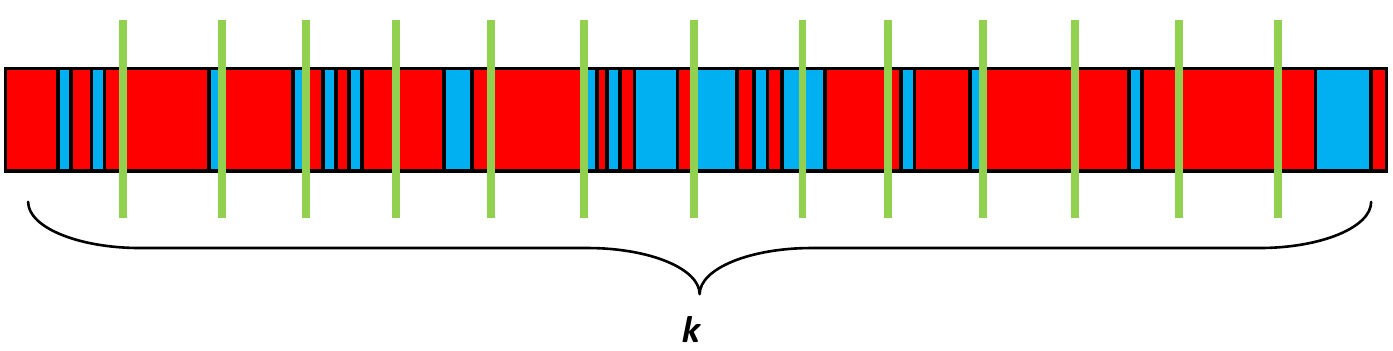}
\caption{A solution, cut into small pieces (by the green vertical lines).}
\label{fig:kPathInt10}
\end{figure}

We now reorder the pieces, so that the pieces that contain more blue elements will be located closer to the beginning and the end (see Fig.~\ref{fig:kPathInt11}). The solution is then ordered as we desired in order to benefit from the usage of generalized representative sets---of course it is no longer a directed path, but a collection of small directed paths (which is not a problem, since they can be reorganized, according to their original order, to obtain a directed path). Using the divide-and-color step described earlier, the dark blue part should appear only in ``early'' pieces, and the light blue part should appear only in the latter pieces; there will be only one piece that may contain both dark and light blue elements.

\begin{figure}[ht!]
\centering
\includegraphics[scale=0.9]{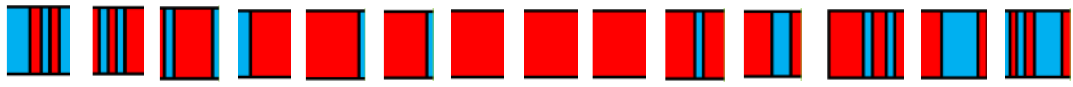}
\caption{A solution, cut into small pieces which were reordered.}
\label{fig:kPathInt11}
\end{figure}

Clearly, since we do not know the solution in advance, we cannot explicitly cut it into small pieces (and order them). Moreover, we cannot explicitly partition the universe $V$ into the set of elements that should be considered in the first piece, the set of elements that should be considered in the second piece, and so on, since this will be extremely inefficient. However, we can {\em implicitly} cut the universe as follows. We simply ``guess'' which are the nodes that are the first and the last in every piece, according to the desired order of the pieces (more precisely, we iterate over every option of choosing these nodes---there is a polynomial number of options since there is a constant number of pieces). We will have functions, $\ell_1,\ell_2,r_1$ and $r_2$, which will indicate (via their images) which are the nodes that we guessed; the functions $\ell_1$ and $\ell_2$ will belong to the first half of the execution, and for an index $j$, $\ell_1$ ($\ell_2$) assigns the first (last) node in the path that we are currently seeking; the functions $r_1$ and $r_2$ are similar, relating to the second half of the execution. Apart from these functions, we will have two nodes, $v^{\ell}$ and $v^r$, that are the first and last nodes of the middle piece, respectively (handling of the middle piece is slightly more technical than the other pieces). To avoid confusion in the following section, we note that each node (besides two nodes) should appear exactly twice, both as a node that should start a piece and as a node that should end a piece (to ensure that the pieces together create a path)---that is, the pieces should be disjoint excluding the nodes connecting them. Also, observe that the node that ends piece $j$ does not necessarily begins piece $j+1$ (e.g., to see this, compare Figures \ref{fig:kPathInt10} and \ref{fig:kPathInt11}---the pieces were reordered). The dynamic programming-based procedure that relies on these functions is technically involved, in particular, since we do not know exactly how many blue elements belong to every piece (though we can upper and lower bound their number\footnote{For example, we can lower bound the number of dark blue elements in the first piece (which should not contain any light blue element), since we can assume that this piece contains more dark blue elements than any other piece of the first ``almost half'' of the execution (the pieces of the other ``almost half'' should not contain dark blue elements).}---which is necessary). However, once the intuition in this section is clear, the technical details of the dynamic programming (in Appendix \ref{section:proofkcwp}) are maybe tedious, but straightforward (except perhaps noticing that one should gradually transfer between the tradeoff parameters of the red part when considering the middle piece).

\subsection{A Subcase of Weighted $k$-Path: Definition}\label{section:disrepkpath}

We next consider {\sc Cut Weighted $k$-Path ($k$-CWP)}, the subcase of {\sc Weighted $k$-Path} that we solve using a generalized representative sets-based procedure. Since the definition of {\sc $k$-CWP} is slightly technical, we first given a brief remider (from Appendix \ref{sec:kPathIntuition}) of relevant explanations of the main idea. Roughly speaking, we define this problem to be {\sc Weighted $k$-Path} where we are given two disjoint sets of nodes, $L$ and $R$, and functions that cut the solution, a ``light'' path on $k$ nodes (if one exists), in a manner that will allow us to consider $L$ only in the first half (plus a chosen small fraction) of the execution of the following procedure that solves it, and $R$ only in the second half (minus the fraction) of its execution. It is also important that these functions ``spread'' the nodes of $L$ (resp. $R$) in a certain manner, which is at worst approximately balanced, among the paths considered in the first (resp. second) half of the execution. Approximate balance actually distorts the balance in the computed partial solutions: partial solutions computed in the first (resp. second) half of the execution will contain a much smaller (resp. larger) fraction of the nodes from $V\setminus (L\cup R)$ belonging to the solution than of the nodes from $L$ (resp. $R$) belonging to the solution, especially when we are close to the middle of the execution. This distortion allows us to benefit (in terms of running time) from the fact that we compute {\em generalized} representative sets.

Fix $0<\epsilon,\delta,\gamma<0.1$ (the exact values are determined later), such that $\frac{1}{\epsilon},m,\widetilde{m}\in\mathbb{N}$, where $m=\frac{1}{2}(\frac{1}{\epsilon}-1)$ and $\widetilde{m}=\delta(\frac{1}{\epsilon}-1)$. Also denote $\widetilde{k}=k-1$. Informally, using the notions described in the previous section, $\frac{1}{\epsilon}$ will be the number of pieces, $\gamma$ will control the size of the blue part of the path (i.e., $|(L\cup R)\cap V_P|$, where $V_P$ is the node-set of a solution (see Appendix \ref{sec:kPathIntuition}), which will be roughly $\gamma k$), $\delta$ will control the shift (i.e., the difference between $|L\cap V_P|$ and $|R\cap V_P|$, which will be roughly $2\delta\gamma k$), $(m+\widetilde{m})$ will be the number of pieces to be considered at the beginning of the execution (in the context of $L$), and $(m-\widetilde{m})$ will be the number of pieces to be considered at the end of the execution (in the context of $R$).\footnote{Observe that $2m$ is the total number of pieces {\em excluding} the ``middle'' piece (i.e., the piece where both $L$ and $R$ are relevant).}

\myparagraph{Input} Formally, the input for {\sc $k$-CWP} consists of a directed graph $G=(V,E)$, a weight function $w: E\rightarrow \mathbb{R}$, a weight $W\in\mathbb{R}$ and a parameter $k\in\mathbb{N}$, along with two disjoint subsets $L,R\subseteq V$, four {\em injective} functions $\ell_1,\ell_2: \{1,2,\ldots,m+\widetilde{m}\}\rightarrow V\setminus (L\cup R)$ and $r_1,r_2: \{1,2,\ldots,m-\widetilde{m}\}\rightarrow V\setminus (L\cup R)$, and distinct $v^{\ell},v^r\in V\setminus(L\cup R)$. The functions should satisfy the following ``function~conditions'':
\begin{enumerate}
\item $(\image(\ell_1)\cap\image(r_1))=(\image(\ell_2)\cap\image(r_2))=\emptyset$, $v^\ell\notin\image(\ell_1)\cup\image(r_1)$ and $v^r\notin\image(\ell_2)\cup\image(r_2)$.
\item $|(\image(\ell_1)\cup\image(r_1)\cup\{v^\ell\})\setminus(\image(\ell_2)\cup\image(r_2)\cup\{v^r\})|=|(\image(\ell_2)\cup\image(r_2)\cup\{v^r\})\setminus(\image(\ell_1)\cup\image(r_1)\cup\{v^\ell\})|=1$.
\end{enumerate}

Recall, from the previous section, that the functions $\ell_1,\ell_2,r_1$ and $r_2$ indicate (via their images) which are the nodes that begin and end each piece (whose solution is a small directed path). The functions $\ell_1$ and $\ell_2$ concern the first half of the execution, and for an index $j$, $\ell_1$ ($\ell_2$) assigns the first (last) node in the small path that we are currently seeking; the functions $r_1$ and $r_2$ are similar, concerning the second half of the execution. The nodes $v^{\ell}$ and $v^r$ specify the first and last nodes of the ``middle'' piece, respectively. Informally, the first condition ensures that a node can neither begin nor end two pieces; the second condition ensures that each node that ends a piece should also start a piece---except one node, since overall we want to construct a simple path and not a simple cycle---and vice versa.

\myparagraph{Objective} We need to decide if $G$ has $\frac{1}{\epsilon}$ simple internally node-disjoint directed paths, $P_1,P_2,\ldots,P_{\frac{1}{\epsilon}}$, whose total weight is at most $W$, and whose number of nodes from $L$ and $R$ is almost the same---$\lfloor(\frac{1}{2}+\delta)\gamma k\rfloor$ from $L$ and $\lfloor(\frac{1}{2}-\delta)\gamma k\rfloor$ from $R$,\footnote{For intution why we take more nodes from $L$ than from $R$, note that the worst performance of the {\sc Weighted $k$-Path} algorithm in \cite{esarepresentative} is slightly {\em after} passing half of its execution (i.e., after computing partial solutions that are paths on $\lceil\frac{k}{2}\rceil$ nodes). This issue was also discussed in the previous section.} such that the following ``solution conditions'' are satisfied:\footnote{In these conditions, for the sake of clarity, we include $v^\ell$ and $v^r$ when referring to ``the images of the input functions''.}
\begin{enumerate}
\item $\forall i\in\{1,2,\ldots,m+\widetilde{m}\}:$ The first and last nodes of $P_i$ are $\ell_1(i)$ and $\ell_2(i)$, respectively, and it contains exactly $\lfloor\epsilon \widetilde{k}\rfloor-1$ internal nodes, where none of them belongs to $R$ or the images of the input functions. Moreover, the total number of nodes from $L$ that are contained in the paths $P_1,\ldots,P_i$ is at least $\frac{i}{m+\widetilde{m}}(\lfloor(\frac{1}{2}+\delta)\gamma k\rfloor-(k-2m\lfloor\epsilon \widetilde{k}\rfloor-2))$.

\item The first and last nodes of $P_{m+\widetilde{m}+1}$ are $v^\ell$ and $v^r$, respectively, and it contains exactly $k-2m\lfloor\epsilon \widetilde{k}\rfloor-2$ internal nodes, where none of them belongs to the images of the input functions.

\item $\forall i\in\{1,2,\ldots,m-\widetilde{m}\}:$ The first and last nodes of $P_{m+\widetilde{m}+1+i}$ are $r_1(i)$ and $r_2(i)$, respectively, and it contains exactly $\lfloor\epsilon \widetilde{k}\rfloor-1$ internal nodes, where none of them belongs to $L$ or the images of the input functions. Moreover, the total number of nodes from $R$ that are contained in the paths $P_{m+\widetilde{m}+1},\ldots,P_i$ is at most $\frac{i}{m-\widetilde{m}}(\lfloor(\frac{1}{2}-\delta)\gamma k\rfloor-(k-2m\lfloor\epsilon \widetilde{k}\rfloor-2))+(k-2m\lfloor\epsilon \widetilde{k}\rfloor-2)$.
\end{enumerate}

Informally, these conditions ensure that the functions, as well as $v^\ell$ and $v^r$, indeed indicate which are the nodes that begin and end each piece (as specified in the explanation of the definition of the input). They also ensure that the number of nodes that are internal nodes in the small paths that we seek, together with the union of $\{v^\ell,v^r\}$ and the images of the functions, is exactly $k$ (so we will overall obtain a path on $k$ nodes). Moreover, they ensure that nodes in $L$ ($R$) tend to be congested at paths of smaller (larger) indices, which we will seek at the beginning (end) of the execution.

\myparagraph{Lemma} Following the explanations given in the previous section, Appendix \ref{section:proofkcwp} shows that {\sc $k$-CWP} can indeed be efficiently solved using a dynamic programming-based procedure embedded with computations of generalized representative sets. More precisely, we prove the following lemma:

\begin{lemma}\label{lemma:kcwpprocor}
For any fixed $c_1,c_2,c_\ell,c_r\!\geq\! 1$  and $0\!<\!\lambda\!<\!1$ s.t.~$c_1\!\geq\! c_2$, there~is~a~small $\epsilon\!=\!\epsilon(\lambda,\delta,\gamma,c_1,c_2)$ s.t.~{\sc $k$-CWP} is solvable in deterministic time $O^*(X\cdot 2^{\lambda k})$~for~$X\!=$
$\max\left\{ \displaystyle{\max_{0\leq i\leq (\frac{1}{2}+\delta)\gamma k}\ \max_{0\leq j\leq \frac{i}{\gamma k}(k-\gamma k)}\!\!Q_\ell Q_1},\ 
\displaystyle{\max_{0\leq i\leq (\frac{1}{2}-\delta)\gamma k}\ \max_{\frac{i+(\frac{1}{2}+\delta)\gamma k}{\gamma k}(k-\gamma k)\leq j\leq k-\gamma k}\!\!\!\!\!\!Q_r Q_2}\right\}$,

\smallskip
{\noindent where $\displaystyle{Q_\ell = \frac{(c_\ell(\frac{1}{2}+\delta)\gamma k)^{2(\frac{1}{2}+\delta)\gamma k - i}}{i^i(c_\ell(\frac{1}{2}+\delta)\gamma k - i)^{2(\frac{1}{2}+\delta)\gamma k - 2i}}}$,
$Q_1 = \displaystyle{\frac{(c_1(k-\gamma k))^{2(k-\gamma k)-j}}{j^j(c_1(k-\gamma k)-j)^{2(k-\gamma k)-2j}}}$,}

\smallskip
{\noindent $Q_r = \displaystyle{\frac{(c_r(\frac{1}{2}-\delta)\gamma k)^{2(\frac{1}{2}-\delta)\gamma k - i}}{i^i(c_r(\frac{1}{2}-\delta)\gamma k - i)^{2(\frac{1}{2}-\delta)\gamma k - 2i}}}$
and $Q_2 = \displaystyle{\frac{(c_2(k-\gamma k))^{2(k-\gamma k)-j}}{j^j(c_2(k-\gamma k)-j)^{2(k-\gamma k)-2j}}}$.}
\end{lemma}

\subsection{A Subcase of Weighted $k$-Path: Computation}\label{section:proofkcwp}

In this section we prove the following lemma, which implies the correctness of Lemma \ref{lemma:kcwpprocor}:

\begin{lemma}\label{lemma:kcwp}
Let $x=(\lfloor(\frac{1}{2}+\delta)\gamma k\rfloor+\lfloor(\frac{1}{2}-\delta)\gamma k\rfloor)$. Then, for any fixed $c_1,c_2,c_\ell,c_r\geq 1$ such that $c_1\geq c_2$, {\sc $k$-CWP} can be solved in deterministic time $O^*(X\cdot 2^{o(k)})$, where $X$ is bounded by the maximum among the following four expressions:\footnote{We give each expression a short name that roughly describes its relevance.}
\begin{enumerate}
\item \underline{``First half ($+\delta$) of the procedure''}:\\
\smallskip
$\displaystyle{\max_{i=1}^{m+\widetilde{m}}
\ \ \max_{j=\frac{i-1}{m+\widetilde{m}}(\lfloor(\frac{1}{2}+\delta)\gamma k\rfloor-(k-2m\lfloor\epsilon \widetilde{k}\rfloor-2))}^{\lfloor(\frac{1}{2}+\delta)\gamma k\rfloor}
\ \ \max_{s=1+(i-1)(\lfloor\epsilon \widetilde{k}\rfloor-1)-j}^{i(\lfloor\epsilon \widetilde{k}\rfloor-1)-j}}$

\medskip
$\displaystyle{\left( \frac{(c_\ell\lfloor(\frac{1}{2}+\delta)\gamma k\rfloor)^{2\lfloor(\frac{1}{2}+\delta)\gamma k\rfloor-j}}{j^j(c_\ell\lfloor(\frac{1}{2}+\delta)\gamma k\rfloor-j)^{2\lfloor(\frac{1}{2}+\delta)\gamma k\rfloor-2j}}
\cdot
\frac{(c_1(\widetilde{k}-\frac{1}{\epsilon}-x))^{2(\widetilde{k}-\frac{1}{\epsilon}-x)-s}}{s^s(c_1(\widetilde{k}-\frac{1}{\epsilon}-x)-s)^{2(\widetilde{k}-\frac{1}{\epsilon}-x)-2s}} \right)}$

\bigskip
\medskip
\item \underline{``Transition (phase 1) from $L$ to $R$''}:\\
\smallskip
$\displaystyle{\max_{i=\lfloor(\frac{1}{2}+\delta)\gamma k\rfloor-(k-2m\lfloor\epsilon \widetilde{k}\rfloor-2)}^{\lfloor(\frac{1}{2}+\delta)\gamma k\rfloor}
\ \max_{j=0}^{k-2m\lfloor\epsilon \widetilde{k}\rfloor-2}
\ \ \ \max_{s=1+(m+\widetilde{m})(\lfloor\epsilon \widetilde{k}\rfloor-1)-i-j}^{(\widetilde{k}-\frac{1}{\epsilon}\lfloor\epsilon \widetilde{k}\rfloor)+(m+\widetilde{m}+1)(\lfloor\epsilon \widetilde{k}\rfloor-1)-\lfloor(\frac{1}{2}+\delta)\gamma k\rfloor-j}}$

\medskip
$\displaystyle{\left(\!\!
\frac{(c_\ell\lfloor\!(\frac{1}{2}\!+\!\delta)\gamma k\!\rfloor)^{2\lfloor\!(\frac{1}{2}\!+\!\delta)\gamma k\!\rfloor\!-\!i}}{i^i(c_\ell\lfloor\!(\frac{1}{2}\!+\!\delta)\gamma k\!\rfloor\!-\!i)^{2\lfloor\!(\frac{1}{2}\!+\!\delta)\gamma k\!\rfloor\!-\!2i}}
\!\cdot\!
\frac{(c_r\lfloor\!(\frac{1}{2}\!-\!\delta)\gamma k\!\rfloor)^{2\lfloor\!(\frac{1}{2}\!-\!\delta)\gamma k\!\rfloor\!-\!j}}{j^j(c_r\lfloor\!(\frac{1}{2}\!-\!\delta)\gamma k\!\rfloor\!-\!j)^{2\lfloor\!(\frac{1}{2}\!-\!\delta)\gamma k\!\rfloor\!-\!2j}}
\!\cdot\!
\frac{(c_1(\widetilde{k}\!\!-\!\!\frac{1}{\epsilon}\!\!-\!\!x))^{2(\widetilde{k}\!-\!\frac{1}{\epsilon}\!-\!x)\!-\!s}}{s^s(c_1(\widetilde{k}\!\!-\!\!\frac{1}{\epsilon}\!\!-\!\!x)\!\!-\!\!s)^{2(\widetilde{k}\!-\!\frac{1}{\epsilon}\!-\!x)\!-\!2s}} \!\!\right)}$

\medskip
\bigskip
\item \underline{``Transition (phase 2) from $c_1$ to $c_2$''}:\footnote{Assume WLOG that $\frac{c_1-c_2}{\epsilon}\in\mathbb{N}$.}\\\smallskip
$\displaystyle{\max_{c\in\{c_1-\epsilon,c_1-2\epsilon,\ldots,c_2\}}
\ \ \ \max_{j=0}^{\widetilde{k}-2m\lfloor\epsilon \widetilde{k}\rfloor}
\max_{s=1+(\widetilde{k}-\frac{1}{\epsilon}\lfloor\epsilon \widetilde{k}\rfloor)+(m+\widetilde{m}+1)(\lfloor\epsilon\widetilde{k}\rfloor-1)-\lfloor(\frac{1}{2}+\delta)\gamma k\rfloor-j}}$

\medskip
$\displaystyle{\left(\! \frac{(c_r\lfloor\!(\frac{1}{2}\!-\!\delta)\gamma k\!\rfloor)^{2\lfloor\!(\frac{1}{2}\!-\!\delta)\gamma k\!\rfloor\!-\!j}}{j^j(c_r\lfloor\!(\frac{1}{2}\!-\!\delta)\gamma k\!\rfloor\!-\!j)^{2\lfloor\!(\frac{1}{2}\!-\!\delta)\gamma k\!\rfloor\!-\!2j}}
\cdot
\frac{((c\!+\!\epsilon)(\widetilde{k}\!\!-\!\!\frac{1}{\epsilon}\!\!-\!\!x))^{\widetilde{k}\!-\!\frac{1}{\epsilon}\!-\!x}}{s^s((c\!+\!\epsilon)(\widetilde{k}\!\!-\!\!\frac{1}{\epsilon}\!\!-\!\!x)\!-\!s)^{\widetilde{k}\!-\!\frac{1}{\epsilon}\!-\!x\!-\!s}}
\cdot
\frac{(c(\widetilde{k}\!\!-\!\!\frac{1}{\epsilon}\!\!-\!\!x))^{\widetilde{k}\!-\!\frac{1}{\epsilon}\!-\!x\!-\!s}}{(c(\widetilde{k}\!\!-\!\!\frac{1}{\epsilon}\!\!-\!\!x)\!\!-\!\!s)^{\widetilde{k}\!-\!\frac{1}{\epsilon}\!-\!x\!-\!s}} \!\!\right)}$

\medskip
\bigskip
\item \underline{``Second half ($-\delta$) of the procedure''}:\\
\smallskip
$\displaystyle{\max_{i=1}^{m\!-\!\widetilde{m}}
\ \max_{j=0}^{\frac{i}{m\!-\!\widetilde{m}}(\lfloor(\frac{1}{2}\!-\!\delta)\gamma k\rfloor\!-\!(k\!-\!2m\lfloor\epsilon \widetilde{k}\rfloor\!-\!2))\!+\!(k\!-\!2m\lfloor\epsilon \widetilde{k}\rfloor\!-\!2)}
\ \max_{s=1\!+\!(\widetilde{k}\!-\!\frac{1}{\epsilon}\lfloor\epsilon \widetilde{k}\rfloor)\!+\!(m\!+\!\widetilde{m}\!+\!i)(\lfloor\epsilon \widetilde{k}\rfloor)\!-\!\lfloor(\frac{1}{2}\!+\!\delta)\gamma k\rfloor\!-\!j}^{(\widetilde{k}\!-\!\frac{1}{\epsilon}\lfloor\epsilon \widetilde{k}\rfloor)\!+\!(m\!+\!\widetilde{m}\!+\!i\!+\!1)(\lfloor\epsilon \widetilde{k}\rfloor\!-\!1)\!-\!\lfloor(\frac{1}{2}\!+\!\delta)\gamma k\rfloor\!-\!j}}$

\medskip
$\displaystyle{\left(\! \frac{(c_r\lfloor(\frac{1}{2}-\delta)\gamma k\rfloor)^{2\lfloor(\frac{1}{2}-\delta)\gamma k\rfloor-j}}{j^j(c_r\lfloor(\frac{1}{2}-\delta)\gamma k\rfloor-j)^{2\lfloor(\frac{1}{2}-\delta)\gamma k\rfloor-2j}}
\cdot
\frac{(c_2(\widetilde{k}-\frac{1}{\epsilon}-x))^{2(\widetilde{k}-\frac{1}{\epsilon}-x)-s}}{s^s(c_2(\widetilde{k}-\frac{1}{\epsilon}-x)-s)^{2(\widetilde{k}-\frac{1}{\epsilon}-x)-2s}} \!\right)}$
\end{enumerate}
\end{lemma}

\medskip
\begin{proof}
Let $IMG$ denote the union of the images of $\ell_1,\ell_2,r_1$ and $r_2$. When we next refer to (generalized) representative families (see Definition \ref{def:disrep}), suppose that $E_1=L$, $E_2=R$, $E_3=V\setminus(L\cup R\cup IMG)$, $k_1=\lfloor(\frac{1}{2}+\delta)\gamma k\rfloor$, $k_2=\lfloor(\frac{1}{2}-\delta)\gamma k\rfloor$ and $k_3=\widetilde{k}-\frac{1}{\epsilon}-x$.

We now present a standard dynamic programming-based procedure to prove the lemma, in which we embed representative sets computations (after each computation of a family of partial solutions, we compute a family that represents it). To this end, we use the following three matrices (a formal definition of the partial solutions that they store is given below):
\begin{enumerate}
\item M has an entry $[i,j,s,v]$ for all
$i\in\{1,\ldots,m\!+\!\widetilde{m}\}$,
$j\in\{\frac{i-1}{m+\widetilde{m}}(\lfloor(\frac{1}{2}\!+\!\delta)\gamma k\rfloor\!-\!(k\!-\!2m\lfloor\epsilon \widetilde{k}\rfloor\!-\!2)),\ldots,\lfloor(\frac{1}{2}\!+\!\delta)\gamma k\rfloor\}$,
$s\in\{1\!+\!(i\!-\!1)(\lfloor\epsilon \widetilde{k}\rfloor\!-\!1)\!-\!j,\ldots,i(\lfloor\epsilon \widetilde{k}\rfloor\!-\!1)\!-\!j\}$,
and $v\in V\setminus (R\cup IMG)$
such that [$(j\!+\!s=(i\!-\!1)(\lfloor\epsilon\widetilde{k}\rfloor\!-\!1)\!+\!1) \rightarrow (\{\ell_1(i),v\}\in E)$] and [$(j\!+\!s=i(\lfloor\epsilon\widetilde{k}\rfloor\!-\!1)) \rightarrow (\{v,\ell_2(i)\}\in E)$].

\bigskip
Informally: M concerns the first (approximately) half of the execution. The parameter $i$ specifies the number of the piece that we are currently handling; the parameters $j$ and $s$ concern the number of nodes that we have used (so far) from $L$ and $V\setminus(L\cup R\cup IMG)$, respectively; the parameter $v$ is the last node that we added to the currently constructed path.

\bigskip
\item N has an entry $[i,j,s,v]$ for all
$i\in\{\lfloor(\frac{1}{2}+\delta)\gamma k\rfloor\!-\!(k\!-\!2m\lfloor\epsilon \widetilde{k}\rfloor\!-\!2),\ldots,\lfloor(\frac{1}{2}\!+\!\delta)\gamma k\rfloor\}$,
$j\in\{0,\ldots,k\!-\!2m\lfloor\epsilon \widetilde{k}\rfloor\!-\!2\}$,
$s\in\{1\!+\!(m\!+\!\widetilde{m})(\lfloor\epsilon\widetilde{k}\rfloor\!-\!1)\!-\!i\!-\!j,\ldots,(\widetilde{k}\!-\!\frac{1}{\epsilon}\lfloor\epsilon \widetilde{k}\rfloor)\!+\!(m\!+\!\widetilde{m}\!+\!1)(\lfloor\epsilon \widetilde{k}\rfloor\!-\!1)\!-\!\lfloor(\frac{1}{2}\!+\!\delta)\gamma k\rfloor\!-\!j\}$,
and $v\in V\setminus IMG$
such that [$(i\!+\!j\!+\!s=(m\!+\!\widetilde{m})(\lfloor\epsilon\widetilde{k}\rfloor\!-\!1)\!+\!1) \rightarrow (\{v^\ell,v\}\in E)$] and [$(i\!+\!j+\!\!s=(m\!+\!\widetilde{m}\!+\!1)(\lfloor\epsilon\widetilde{k}\rfloor\!-\!1)) \rightarrow (\{v,v^r\}\in E)$].

\bigskip
Informally: N concerns the (approximately) middle piece. The parameters $i$, $j$ and $s$ concern the number of nodes that we have used (so far) from $L$,$R$ and $V\setminus(L\cup R\cup IMG)$, respectively; the parameter $v$ is the last node that we added to the currently constructed path.

\bigskip
\item K has an entry $[i,j,s,v]$ for all
$i\in\{1,\ldots,m\!-\!\widetilde{m}\}$,
$j\in\{0,\ldots,\frac{i}{m\!-\!\widetilde{m}}(\lfloor(\frac{1}{2}\!-\!\delta)\gamma k\rfloor\!-\!(k\!-\!2m\lfloor\epsilon \widetilde{k}\rfloor\!-\!2))\!+\!(k\!-\!2m\lfloor\epsilon \widetilde{k}\rfloor\!-\!2)\}$,
$s\in\{1\!+\!(\widetilde{k}\!-\!\frac{1}{\epsilon}\lfloor\epsilon \widetilde{k}\rfloor)\!+\!(m\!+\!\widetilde{m}\!+\!i)(\lfloor\epsilon \widetilde{k}\rfloor\!-\!1)\!-\!\lfloor(\frac{1}{2}+\delta)\gamma k\rfloor\!-\!j,\ldots,(\widetilde{k}\!-\!\frac{1}{\epsilon}\lfloor\epsilon \widetilde{k}\rfloor)\!+\!(m\!+\!\widetilde{m}\!+\!i\!+\!1)(\lfloor\epsilon \widetilde{k}\rfloor\!-\!1)\!-\!\lfloor(\frac{1}{2}+\delta)\gamma k\rfloor\!-\!j\}$,
and $v\in V\setminus (L\cup IMG)$
such that [$(\lfloor(\frac{1}{2}\!+\!\delta)\gamma k\rfloor\!+\!j\!+\!s=(i\!-\!1)(\lfloor\epsilon\widetilde{k}\rfloor\!-\!1)\!+\!1) \rightarrow (\{r_1(i),v\}\in E)$] and [$(\lfloor(\frac{1}{2}\!+\!\delta)\gamma k\rfloor\!+\!j\!+\!s=i(\lfloor\epsilon\widetilde{k}\rfloor\!-\!1)) \rightarrow (\{v,r_2(i)\}\in E)$].

\bigskip
Informally: K concerns the second (approximately) half of the execution. The parameter $i$ specifies the number of the piece that we are currently handling; the parameters $j$ and $s$ concern the number of nodes that we have used (so far) from $R$ and $V\setminus(L\cup R\cup IMG)$, respectively; the parameter $v$ is the last node that we added to the currently constructed path.
\end{enumerate}

\smallskip
{\noindent We next assume that a reference to an undefined entry returns $\emptyset$. The other entries will store the following families of partial solutions, where we assume that we track the weights of the sets of paths corresponding to the partial solutions:\footnote{When computing an entry and obtaining the same partial solution from several different sets of paths, we store the minimal weight.}}
\begin{enumerate}
\item M$[i,j,s,v]$: A family that min $(k_1-j,k_2,k_3-s)$-represents the family that contains any union of node-sets of $i$ simple internally node-disjoint directed paths $P_1,P_2,\ldots,P_i$, excluding the nodes in $IMG$, which satisfy the following conditions:
	\begin{itemize}
	\item $P_1,P_2,\ldots P_{i-1}$ satisfy ``solution condition 1''.
	\item The first and last nodes of $P_i$ are $\ell_1(i)$ and $v$, respectively, and it does not contain internal nodes from $(R\cup IMG)$. Moreover, the total number of nodes from $L$ and $V\setminus(L\cup IMG)$ that are contained in the paths $P_1,\ldots,P_i$ are exactly $j$ and $s$, respectively.
	\end{itemize}

\bigskip
\item N$[i,j,s,v]$: A family that min $(k_1-i,k_2-j,k_3-s)$-represents the family that contains any union of node-sets of $\widetilde{i}=(m+\widetilde{m}+1)$ simple internally node-disjoint directed paths $P_1,P_2,\ldots,P_{\widetilde{i}}$, excluding the nodes in $IMG$, which satisfy the following conditions:
	\begin{itemize}
	\item $P_1,P_2,\ldots P_{\widetilde{i}-1}$ satisfy ``solution condition 1''.
	\item The first and last nodes of $P_{\widetilde{i}}$ are $v^\ell$ and $v$, respectively, and it does not contain internal nodes from $IMG$. Moreover, the total number of nodes from $L$, $R$ and $V\setminus(L\cup R\cup IMG)$ that are contained in the paths $P_1,\ldots,P_{\widetilde{i}}$ are exactly $i$, $j$ and $s$, respectively.
	\end{itemize}

\bigskip
\item K$[i,j,s,v]$: A family that min $(0,k_2-j,k_3-s)$-represents the family that contains any union of node-sets of $\widetilde{i}=m+\widetilde{m}+1+i$ simple internally node-disjoint directed paths $P_1,P_2,\ldots,P_{\widetilde{i}}$, excluding the nodes in $IMG$, which satisfy the following conditions:
	\begin{itemize}
	\item $P_1,P_2,\ldots P_{\widetilde{i}-1}$ satisfy ``solution conditions 1--3''.
	\item The first and last nodes of $P_i$ are $r_1(i)$ and $v$, respectively, and it does not contain internal nodes from $(L\cup IMG)$. Moreover, the total number of nodes from $R$ and $V\setminus(L\cup R\cup IMG)$ that are contained in the paths $P_1,\ldots,P_{\widetilde{i}}$ are exactly $j$ and $s$, respectively.
	\end{itemize}
\end{enumerate}

\bigskip
{\noindent The entries are computed in the following order:}
\begin{enumerate}
\item For $i=1,\ldots,m\!+\!\widetilde{m}$:
	\begin{itemize}
	\item For $j=\frac{i-1}{m+\widetilde{m}}(\lfloor(\frac{1}{2}\!+\!\delta)\gamma k\rfloor\!-\!(k\!-\!2m\lfloor\epsilon \widetilde{k}\rfloor\!-\!2)),\ldots,\lfloor(\frac{1}{2}\!+\!\delta)\gamma k\rfloor$:
		\begin{itemize}
			\item For $s=1\!+\!(i\!-\!1)(\lfloor\epsilon \widetilde{k}\rfloor\!-\!1)\!-\!j,\ldots,i(\lfloor\epsilon \widetilde{k}\rfloor\!-\!1)\!-\!j$:
				\begin{itemize}
					\item Compute all entries of the form M$[i,j,s,v]$.
				\end{itemize}
		\end{itemize}
	\end{itemize}

\smallskip
\item For $i=\lfloor(\frac{1}{2}\!+\!\delta)\gamma k\rfloor\!-\!(k\!-\!2m\lfloor\epsilon \widetilde{k}\rfloor\!-\!2),\ldots,\lfloor(\frac{1}{2}\!+\!\delta)\gamma k\rfloor$:
	\begin{itemize}
	\item For $j=0,\ldots,k\!-\!2m\lfloor\epsilon \widetilde{k}\rfloor\!-\!2$:
		\begin{itemize}
			\item For $s=1\!+\!(m\!+\!\widetilde{m})(\lfloor\epsilon\widetilde{k}\rfloor\!-\!1)\!-\!i\!-\!j,\ldots,(\widetilde{k}\!-\!\frac{1}{\epsilon}\lfloor\epsilon \widetilde{k}\rfloor)\!+\!(m\!+\!\widetilde{m}\!+\!1)(\lfloor\epsilon \widetilde{k}\rfloor\!-\!1)\!-\!\lfloor(\frac{1}{2}\!+\!\delta)\gamma k\rfloor\!-\!j$:
				\begin{itemize}
					\item Compute all entries of the form K$[i,j,s,v]$.
				\end{itemize}
		\end{itemize}
	\end{itemize}

\smallskip	
\item For $i=1,\ldots,m\!-\!\widetilde{m}$:
	\begin{itemize}
	\item For $j=0,\ldots,\frac{i}{m\!-\!\widetilde{m}}(\lfloor(\frac{1}{2}\!-\!\delta)\gamma k\rfloor\!-\!(k\!-\!2m\lfloor\epsilon \widetilde{k}\rfloor\!-\!2))\!+\!(k\!-\!2m\lfloor\epsilon \widetilde{k}\rfloor\!-\!2)$:
		\begin{itemize}
			\item For $s=1\!+\!(\widetilde{k}\!-\!\frac{1}{\epsilon}\lfloor\epsilon \widetilde{k}\rfloor)\!+\!(m\!+\!\widetilde{m}\!+\!i)(\lfloor\epsilon \widetilde{k}\rfloor\!-\!1)\!-\!\lfloor(\frac{1}{2}+\delta)\gamma k\rfloor\!-\!j,\ldots,(\widetilde{k}\!-\!\frac{1}{\epsilon}\lfloor\epsilon k\rfloor)\!+\!(m\!+\!\widetilde{m}\!+\!i\!+\!1)(\lfloor\epsilon \widetilde{k}\rfloor\!-\!1)\!-\!\lfloor(\frac{1}{2}+\delta)\gamma k\rfloor\!-\!j$:
				\begin{itemize}
					\item Compute all entries of the form N$[i,j,s,v]$.
				\end{itemize}
		\end{itemize}
	\end{itemize}
\end{enumerate}

\bigskip
{\noindent We now give the recursive formulas using which the entries are computed.\footnote{Entries whose computation is not specified hold $\emptyset$.} }

\bigskip
\underline{The matrix M}:
	\begin{enumerate}
	\item If $j+s=1$:
		\begin{enumerate}
		\item If $j=1$ and $v\in L$: M$[i,j,s,v] = \{\{v\}\}$.
		\item Else if $s=1$ and $v\notin L$: M$[i,j,s,v] = \{\{v\}\}$.
		\end{enumerate}

	\item Else:
		\begin{enumerate}
		\item If $v\in L$: M$[i,j,s,v] =$\\$\displaystyle{\{\{v\}\cup A: A\in\!\bigcup_{u\in V}\!\mathrm{M}[i\!-\!1,j\!-\!1,s,u]\} \cup \{\{v\}\cup A: A\in\!\!\!\!\!\bigcup_{(u,v)\in E}\!\!\!\!\!\mathrm{M}[i,j\!-\!1,s,u]\}}$.
				
		\item Else: M$[i,j,s,v] =$\\$\displaystyle{\{\{v\}\cup A: A\in\!\bigcup_{u\in V}\!\mathrm{M}[i\!-\!1,j,s\!-\!1,u]\} \cup \{\{v\}\cup A: A\in\!\!\!\!\!\bigcup_{(u,v)\in E}\!\!\!\!\!\mathrm{M}[i,j,s\!-\!1,u]\}}$.
	\end{enumerate}
	\end{enumerate}

\vspace{-1.5em}	
	\begin{itemize}
	\item After 2: Replace the result by a family that min $(k_1-j,k_2,k_3-s)$-represents it, computed using $c_\ell$, 1 and $c_1$, corresponding to $E_1$, $E_2$ and $E_3$, respectively.
	\end{itemize}

\smallskip
\underline{The matrix N}:
	\begin{enumerate}
	\item If $i+j+s=1+(m+\widetilde{m})(\lfloor\epsilon\widetilde{k}\rfloor-1)$:
		\begin{enumerate}
		\item If $j=0$ and $v\in L$: N$\displaystyle{[i,j,s,v] = \{\{v\}\cup A: A\in\!\bigcup_{u\in V}\!\mathrm{M}[(m+\widetilde{m}),j\!-\!1,s,u]\}}$.
		\item Else if $j=1$ and $v\in R$: N$\displaystyle{[i,j,s,v] = \{\{v\}\cup A: A\in\!\bigcup_{u\in V}\!\mathrm{M}[(m+\widetilde{m}),j,s,u]\}}$.
		\item Else if $j=0$ and $v\in V\setminus(L\cup R\cup IMG)$:\\N$\displaystyle{[i,j,s,v] = \{\{v\}\cup A: A\in\!\bigcup_{u\in V}\!\mathrm{M}[(m+\widetilde{m}),j,s\!-\!1,u]\}}$.
		\end{enumerate}
	
	\item Else:
		\begin{enumerate}
		\item If $v\in L$: N$\displaystyle{[i,j,s,v] = \{\{v\}\cup A: A\in\!\bigcup_{(u,v)\in V}\!\mathrm{N}[i\!-\!1,j,s,u]\}}$.
		\item Else if $v\in R$: N$\displaystyle{[i,j,s,v] = \{\{v\}\cup A: A\in\!\bigcup_{(u,v)\in V}\!\mathrm{N}[i,j\!-\!1,s,u]\}}$.
		\item Else: N$\displaystyle{[i,j,s,v] = \{\{v\}\cup A: A\in\!\bigcup_{(u,v)\in V}\!\mathrm{N}[i,j,s\!-\!1,u]\}}$.
		\end{enumerate}
	\end{enumerate}

\vspace{-1.5em}
	\begin{itemize}
	\item After 1 and 2: Replace the result by a family that min $(k_1-i,k_2-j,k_3-s)$-represents it, computed using $c_\ell$, $c_r$ and $c_1$, corresponding to $E_1$, $E_2$ and $E_3$, respectively.
	\end{itemize}

\smallskip	
\underline{The matrix K}:
	\begin{enumerate}
	\item If $j+s=1+(\widetilde{k}-\frac{1}{\epsilon}\lfloor\epsilon \widetilde{k}\rfloor)+(m+\widetilde{m}+1)(\lfloor\epsilon\widetilde{k}\rfloor-1)-(\lfloor(\frac{1}{2}+\delta)\gamma k\rfloor$:
		\begin{enumerate}
		\item If $v\in R$: K$\displaystyle{[i,j,s,v] = \{\{v\}\cup A: A\in\!\bigcup_{u\in V}\!\mathrm{N}[(\lfloor(\frac{1}{2}\!+\!\delta)\gamma k\rfloor,j\!-\!1,s,u]\}}$.
		\item Else: K$\displaystyle{[i,j,s,v] = \{\{v\}\cup A: A\in\!\bigcup_{u\in V}\!\mathrm{N}[(\lfloor(\frac{1}{2}\!+\!\delta)\gamma k\rfloor,j,s\!-\!1,u]\}}$.
		\end{enumerate}
	
	\item Else:
		\begin{enumerate}
		\item If $v\in R$: K$[i,j,s,v] =$\\$\displaystyle{\{\{v\}\cup A: A\in\!\bigcup_{u\in V}\!\mathrm{K}[i\!-\!1,j\!-\!1,s,u]\} \cup \{\{v\}\cup A: A\in\!\!\!\!\!\bigcup_{(u,v)\in E}\!\!\!\!\!\mathrm{K}[i,j\!-\!1,s,u]\}}$.
		\item Else: K$[i,j,s,v] =$\\$\displaystyle{\{\{v\}\cup A: A\in\!\bigcup_{u\in V}\!\mathrm{K}[i\!-\!1,j,s\!-\!1,u]\} \cup \{\{v\}\cup A: A\in\!\!\!\!\!\bigcup_{(u,v)\in E}\!\!\!\!\!\mathrm{K}[i,j,s\!-\!1,u]\}}$.
		\end{enumerate}
	\end{enumerate}	

\vspace{-1.5em}
	\begin{itemize}
	\item After 1: Perform the following computation.
		\begin{enumerate}
		\item Initialize $A$ to be the result.
		\item For $c=c_1-\epsilon,c_1-2\epsilon,\ldots,c_2$:
			\begin{itemize}
			\item Replace $A$ by a family that $(k_1,k_2-j,k_3-s)$-represents it, computed using 1, $c_r$ and $c$, corresponding to $E_1$, $E_2$ and $E_3$, respectively.
			\end{itemize}
		\item Replace the original result by $A$.
		\end{enumerate}
	\item After 2: Replace the result by a family that min $(k_1,k_2-j,k_3-s)$-represents it, computed using 1, $c_r$ and $c_2$, corresponding to $E_1$, $E_2$ and $E_3$, respectively.
	\end{itemize}

{\noindent Finally, we return yes iff at least one entry of the form K$[\frac{1}{\epsilon},k_2,k_3,v]$ contains a set of weight at most $(W-w((v,r_2(\frac{1}{\epsilon}))))$.}

By Theorem \ref{theorem:disrep}, up to a factor of $2^{o(k)}$, the running time required to compute M is bounded by $O^*$ of the first expression in the lemma, the running time required to compute N is bounded by $O^*$ of the second expression, the running time required to compute the entries of K in which $i=1$ and $j+s=1+(\widetilde{k}-\frac{1}{\epsilon}\lfloor\epsilon \widetilde{k}\rfloor)+(m+\widetilde{m}+1)(\lfloor\epsilon\widetilde{k}\rfloor-1)-(\lfloor(\frac{1}{2}+\delta)\gamma k\rfloor$ is bounded by $O^*$ of the third expression, and the running time required to compute the other entries of K is bounded by $O^*$ of the fourth expression.\qed
\end{proof}

\subsection{A Deterministic Algorithm for Weighted $k$-Path}\label{section:kpathdet}

{\noindent Fix $0\!<\!\delta,\gamma\!<\!0.1$ and $c_1,c_2,c_\ell,c_r\!\geq\! 1$ such that~$c_1\!\geq\! c_2$, and let $\epsilon\!=\!\epsilon(\frac{1}{10^{11}},\delta,\gamma,c_1,c_2)$.

\begin{algorithm}
\caption{\alg{PathAlg}($G=(V,E),w,W,k$)}
\begin{algorithmic}[1]
\STATE let $v_1,v_2,\ldots,v_n$ be an arbitrary ordering of the nodes in $V$.
\STATE get an $(n,
\lfloor(\frac{1}{2}+\delta)\gamma k\rfloor+\lfloor(\frac{1}{2}-\delta)\gamma k\rfloor,
\lfloor(\frac{1}{2}-\delta)\gamma k\rfloor)$-universal set ${\cal F}$ by using Theorem~\ref{theorem:splitter}.

\FORALL{$f\in {\cal F}$}
	\FORALL{$U\subseteq V$ such that $|U|=\frac{1}{\epsilon}+1$}\label{step:firstcut0}
	\FOR{$i=1,2,\ldots,n$}\label{step:firstcut}
		\STATE\label{step:divandcol} let $L=\{v_j\!\in\! V\setminus U: j\!\geq\! i, f(v_j)\!=\!0\}$ and $R=\{v_j\!\in\! V\setminus U: j\!\geq\! i, f(v_j)\!=\!1\}$.
		\FORALL{$\ell_1,\ell_2\!: \{1,2,\ldots,m\!+\!\widetilde{m}\}\rightarrow U$, $\ r_1,r_2\!: \{1,2,\ldots,m\!-\!\widetilde{m}\}\rightarrow U$,  $\ v^{\ell},v^r\!\!\in\!U$}\label{step:functions}
			\IF{$(G,w,W,k,L,R,\ell_1,\ell_2,r_1,r_2,v^\ell,v^r)$ is an input for {\sc $k$-CWP}}\label{step:legalcheck}
				\STATE\label{step:runkcwpro} {\bf if} the procedure of Lemma \ref{lemma:kcwpprocor} accepts it {\bf then} accept. {\bf end if}
			\ENDIF
		\ENDFOR
	\ENDFOR
	\ENDFOR
\ENDFOR

\STATE reject.
\end{algorithmic}
\end{algorithm}

We now give \alg{PathAlg}, a deterministic algorithm for {\sc Weighted $k$-Path} (follow the pseudocode below). \alg{PathAlg} first orders the nodes in $V$ (Step 1), and computes a universal set, on whose functions it iterates, that will be soon relevant to the computation of $L$ and $R$ (Steps 2--3). It performs an exhaustive search (in Step 4) as a part of the balanced cutting (of the universe) technique, attempting to capture (in a ``small'' set $U$) the nodes that should connect that small directed paths that are part of a solution to (a soon to be constructed instance of) {\sc $k$-CWP}. Then, it performs (in Step \ref{step:firstcut}) another exhaustive search (that is also a part of the balanced cutting technique), iterating over every node $v_i\!\in\! V$ (Step \ref{step:firstcut}), to capture exactly $(\lfloor(\frac{1}{2}\!+\!\delta)\gamma k\rfloor\!+\!\lfloor(\frac{1}{2}\!-\!\delta)\gamma k\rfloor)$ of the nodes on a path that is a solution (if one exists) in the set $\{v_i,v_{i+1},\!...,v_n\}\setminus U$ (informally, we compute ``blue part of the universe''). It can then (in Step \ref{step:divandcol}) obtain the sets $L$ and $R$ by using a divide-and-color step---informally, it inserts all the ``blue elements'' to whose indices the currently examined function assigns '0' ('1') to $L$ ($R$). Afterwards it further cuts the universe (which is $V$) by exhaustively iterating over every option to choose four functions $\ell_1,\ell_2,r_1,r_2$ that are legal according to Step \ref{step:legalcheck}, finishing the phase of balanced cutting in Strategy~III (Section \ref{section:strategies}). At this point, one should observe that \alg{PathAlg} does not explicitly cut the universe, obtaining several disjoint universes, which will be extremely inefficient. In Steps \ref{step:firstcut0} and \ref{step:functions} it only considers every option to choose which nodes should cut the universe in a balanced manner (those are the nodes in the images of $\ell_1,\ell_2,r_1$ and $r_2$), which is all the information necessary as input for {\sc $k$-CWP}. There are $|V|^{O(\frac{1}{\epsilon})}\!=\!O^*(1)$ such options. Overall, \alg{PathAlg} obtains a set of inputs for {\sc $k$-CWP}, accepting (in Step \ref{step:runkcwpro}) {\em iff} at least one of them is a yes-instance.}

In the following result, $X$ is bounded as in Lemma \ref{lemma:kcwpprocor}.

\begin{lemma}\label{lemma:pathalgcor}
\alg{PathAlg} solves {\sc Weighted $k$-Path} in time $O^*(X\cdot {\gamma k\choose (\frac{1}{2}-\delta)\gamma k} 2^{\frac{1}{10^{10}}k})$.
\end{lemma}

\begin{proof} 
Recall that $m=\frac{1}{2}(\frac{1}{\epsilon}-1),\widetilde{m}=\delta(\frac{1}{\epsilon}-1)$ and $\widetilde{k}=k-1$.

\myparagraph{First Direction} For the easier direction, suppose that \alg{PathAlg} accepts, and let $L$, $R$, $\ell_1,\ell_2,r_1$ and $r_2$ be the parameters corresponding to the iteration in which it accepts. We get that $(G,w,W,k,L,R,\ell_1,\ell_2,r_1,r_2)$ is a yes-instance of {\sc $k$-CWP}, and can thus denote by $P_1,P_2,\ldots,P_{\frac{1}{\epsilon}}$ the paths that form a solution to this instance. Since ``solutions conditions'' 1--3 are satisfied, the paths are internally node-disjoint, and they contain (together) exactly $k-\frac{1}{\epsilon}-1$ internal nodes.
Now, by ``function condition'' 2, for each path $P_i$ (except one) that ends at some node, there is a path $P_j$ that starts at this node, and vice versa. Therefore, we can connect the paths, and obtain a directed graph whose underlying undirected graph is a tree.
Moreover, by ``function condition'' 1, by the choice (in a legal input for {\sc $k$-CWP}) of the functions (in particular, the functions are injective), as well as $v^\ell$ and $v^r$, and by the ``solution conditions'' (which, in particular, imply that the paths do not contain internal nodes that belong to the images of the functions), we get that the above mentioned tree must be a directed simple path. Since the number of nodes that connect (and belong to) the paths is $\frac{1}{\epsilon}-1$, and two paths have (each) a distinct node that is not an internal node, we have thus constructed a directed simple path on $k$ nodes. It weight is at most $W$, as it contains the same edges as $P_1,P_2,\ldots,P_{\frac{1}{\epsilon}}$, whose total weight is at most~$W$.


\myparagraph{Second Direction} Now, suppose that the input is a yes-instance of {\sc Weighted $k$-Path}. Thus, we can denote by $P$ a path that is a solution to this instance. Let $U$ be the set of the $[(j-1)\lfloor\epsilon \widetilde{k}\rfloor+1]^{st}$, $k^{st}$  and $[(j-1)\lfloor\epsilon \widetilde{k}\rfloor+k-2m\lfloor\epsilon \widetilde{k}\rfloor]^{st}$ nodes on $P$, for all $j\in\{1,2,\ldots,\frac{1}{\epsilon}\}$. Let $v_i$ denote a node in $V$ such that $P$ contains exactly $(\lfloor(\frac{1}{2}+\delta)\gamma k\rfloor+\lfloor(\frac{1}{2}-\delta)\gamma k\rfloor)$ nodes from $V^*_i$, which denotes the set of nodes in $V$ that are not ordered before $v_i$ and do not belong to $U$. Let ${\cal P}_{m+\widetilde{m}+1}$ denote the set of each subpath of $P$ that starts at its $[1+(j-1)\lfloor\epsilon \widetilde{k}\rfloor]^{st}$ node, for any $j\in\{1,2,\ldots, \frac{1}{\epsilon}\}$, and contains exactly $k-2m\lfloor\epsilon \widetilde{k}\rfloor$ nodes. Let $P_{m+\widetilde{m}+1}$ be a subpath in ${\cal P}_{m+\widetilde{m}+1}$ that contains the maximum number of nodes from $V^*_i$ among the subpaths in ${\cal P}_{m+\widetilde{m}+1}$. Denote the node-set of $P_{m+\widetilde{m}+1}$ by $V_{m+\widetilde{m}+1}$. Let $v^\ell$ and $v^r$ denote the first and last nodes of $P_{m+\widetilde{m}+1}$, respectively. Remove the internal nodes of $P_{m+\widetilde{m}+1}$ from $P$. Let $P'$ and $P''$ be the two resulting subpaths, and denote their number of nodes by $k'$ and $k''$, respectively. If $k'>1$, let $P'_j$, for all $j\in\{1,2,\ldots,\frac{k'-1}{\lfloor\epsilon \widetilde{k}\rfloor}\}$, denote the subpath of $P'$ on $\lfloor\epsilon \widetilde{k}\rfloor+1$ nodes that starts at the $[(j-1)\lfloor\epsilon \widetilde{k}\rfloor+1]^{st}$ node of $P'$. Moreover, if $k''>1$, let $P''_j$, for all $j\in\{1,2,\ldots,\frac{k''-1}{\lfloor\epsilon \widetilde{k}\rfloor}\}$, denote the subpath of $P''$ on $\lfloor\epsilon \widetilde{k}\rfloor+1$ nodes that starts at the $[(j-1)\lfloor\epsilon \widetilde{k}\rfloor+1]^{st}$ node of $P''$. Now, denote the paths of the forms $P'_j$ and $P''_j$ arbitrarily as $\widetilde{P}_1,\widetilde{P}_2,\ldots,\widetilde{P}_{2m}$. Moreover, denote their node-sets as $\widetilde{V}_1,\widetilde{V}_2,\ldots,\widetilde{V}_{2m}$, respectively. By our choice of $U$, the first and last nodes of each of these paths belong to $U$.

Let $PRT$ be the set of all partitions of $\{\widetilde{P}_1,\widetilde{P}_2,\ldots,\widetilde{P}_{2m}\}$ into two sets, $\widetilde{\cal P}_L$ of size $(m+\widetilde{m})$ and $\widetilde{\cal P}_R$ of size $(m-\widetilde{m})$.
By our definition of $V^*_i$, we have that $\displaystyle{\sum_{j=1}^{2m}|\widetilde{V}_j\cap V^*_i|+|V_{m+\widetilde{m}+1}\cap V^*_i|=(\lfloor(\frac{1}{2}+\delta)\gamma k\rfloor+\lfloor(\frac{1}{2}-\delta)\gamma k\rfloor)}$. Furthermore, by our choice of $P_{m+\widetilde{m}+1}$, we have that $|\widetilde{V}_j\cap V^*_i|\leq |V_{m+\widetilde{m}+1}\cap V^*_i|$, for all $j\in\{1,2,\ldots,2m\}$.
Therefore, there is a partition $({\cal P}_L,{\cal P}_R)$ in $PRT$ such that $\displaystyle{\sum_{\widetilde{P}_j\in{\cal P}_L}|\widetilde{V}_j\cap V^*_i|\leq \lfloor(\frac{1}{2}+\delta)\gamma k\rfloor}$ and $\displaystyle{\sum_{\widetilde{P}_j\in{\cal P}_R}|\widetilde{V}_j\cap V^*_i|\leq \lfloor(\frac{1}{2}-\delta)\gamma k\rfloor}$. Next consider such a partition.

Denote by $P_1,P_2,\ldots,P_{m+\widetilde{m}}$ an order of the paths in ${\cal P}_L$ such that $|V_{j-1}\cap V^*_i|\geq |V_j\cap V^*_i|$, for all $j\in\{2,3,\ldots,m+\widetilde{m}\}$, where $V_s$ is the node-set of $P_s$, for all $s\in\{1,2,\ldots,m+\widetilde{m}\}$. Furthermore, denote by $P_{m+\widetilde{m}+2},P_{m+\widetilde{m}+3},\ldots,P_{2m+1}$ an order of the paths in ${\cal P}_R$ such that $|V_{j-1}\cap V^*_i|\leq |V_j\cap V^*_i|$, for all $j\in\{m+\widetilde{m}+3,m+\widetilde{m}+4,\ldots,2m+1\}$, where $V_s$ is the node-set of $P_s$, for all $s\in\{m+\widetilde{m}+2,m+\widetilde{m}+3,\ldots,2m+1\}$.
We define $\ell_1,\ell_2: \{1,2,\ldots,m+\widetilde{m}\}\rightarrow U$ by letting $\ell_1(j)$ and $\ell_2(j)$ be the first and last nodes of $P_j$, respectively, for all $j\in\{1,2,\ldots,m+\widetilde{m}\}$. Similarly, we define $r_1,r_2: \{1,2,\ldots,m-\widetilde{m}\}\rightarrow U$ by letting $r_1(j)$ and $r_2(j)$ be the first and last nodes of $P_{m+\widetilde{m}+1+j}$, respectively, for all $j\in\{1,2,\ldots,m-\widetilde{m}\}$.

By our choice of $({\cal P}_L,{\cal P}_R)$, we can select a set of nodes $L_{m+\widetilde{m}+1}$ from $V_{m+\widetilde{m}+1}\cap V^*_i$ such that for the sets $R_{m+\widetilde{m}+1}=(V_{m+\widetilde{m}+1}\cap V^*_i)\setminus L_{m+\widetilde{m}+1}$, $\ L^*=L_{m+\widetilde{m}+1}\cup$ $\displaystyle{ (\bigcup_{j\in\{1,2,\ldots,m+\widetilde{m}\}}(V_j\cap V^*_i))}$ and $\displaystyle{R^*=R_{m+\widetilde{m}+1}\cup (\!\!\!\!\bigcup_{j\in\{m+\widetilde{m}+2,m+\widetilde{m}+3,\ldots,2m+1\}}\!\!\!\!\!\!\!(V_j\cap V^*_i))}$, we get that $|L^*|=\lfloor(\frac{1}{2}+\delta)\gamma k\rfloor$ and $|R^*|=\lfloor(\frac{1}{2}-\delta)\gamma k\rfloor$.
Let ${\cal F}$ be the $(n,
\lfloor(\frac{1}{2}+\delta)\gamma k\rfloor+\lfloor(\frac{1}{2}-\delta)\gamma k\rfloor,
\lfloor(\frac{1}{2}-\delta)\gamma k\rfloor)$-universal set computed by using the algorithm in Theorem \ref{theorem:splitter}. By its definition, there exists $f\in{\cal F}$ such that $L^*\subseteq L=\{v_j\in V\setminus U: j\geq i, f(v_j)=0\}$ and $R^*\subseteq R=\{v_j\in V\setminus U: j\geq i, f(v_j)=1\}$.

Note that there exists an execution of Step \ref{step:runkcwpro} where \alg{PathAlg} calls the procedure of Lemma \ref{lemma:kcwpprocor} with the input $(G,w,W,k,L,R,\ell_1,\ell_2,r_1,r_2,v^\ell,v^r)$. Since this is a yes-instance of {\sc $k$-CWP} (the paths $P_1,P_2,\ldots,P_{2m+1}$ form a solution for this instance), we get that \alg{PathAlg} accepts.\qed
\end{proof}

{\noindent Upper bounds for the running time in Lemma \ref{lemma:pathalgcor}, considering different parameters $\gamma,c,c_\ell$ and $c_r$, are given in Appendix \ref{section:kpathdetapp}. By choosing $\delta = 0.046,\gamma=0.084,c_1=1.504,c_2=1.398,c_\ell=1.092$ and $c_r=1.876$, we conclude that Theorem \ref{theorem:kPathRand} (in Section \ref{section:strategies}) is correct.}

\subsection{The Running Time of \alg{PathAlg}}\label{section:kpathdetapp}

Consider the expression $X$ in the running time given in Lemma \ref{lemma:pathalgcor}. First, note that it is bounded by the maximum of the following two expressions:

\[1.\displaystyle{\max_{0\leq\alpha\leq1}\ \max_{0\leq j\leq (\frac{1}{2}+\delta)\alpha(k-\gamma k)}
\left[
\left(\frac{c_\ell^{2 - \alpha}}{\alpha^\alpha(c_\ell - \alpha)^{2 - 2\alpha}}\right)^{(\frac{1}{2}+\delta)\gamma k}
\!\!\!\cdot
\!\frac{(c_1(k-\gamma k))^{2(k-\gamma k)-j}}{j^j(c_1(k-\gamma k)-j)^{2(k-\gamma k)-2j}}
\right]\!.}\]

\[2.\displaystyle{\max_{0\leq\alpha\leq1}\ \max_{(\frac{1}{2}\!+\!\delta\!+\!\alpha(\frac{1}{2}\!-\!\delta))(k\!-\!\gamma k)\leq j\leq k\!-\!\gamma k}
\!\left[\!
\left(\!\frac{c_r^{2 - \alpha}}{\alpha^\alpha(c_r \!-\! \alpha)^{2 \!-\! 2\alpha}}\!\right)^{(\frac{1}{2}\!-\!\delta)\gamma k}
\!\!\!\!\!\!\cdot
\!\frac{(c_2(k-\gamma k))^{2(k-\gamma k)-j}}{j^j(c_2(k\!-\!\gamma k)\!-\!j)^{2(k\!-\!\gamma k)\!-\!2j}}
\!\right]\!.}\]

{\noindent We can further bound this expression by $Y^k$, where $Y$ is bounded by the maximum of the following two expressions:}

\begin{enumerate}
\item $\displaystyle{\max_{0\leq\alpha\leq1}\ \max_{0\leq\beta\leq\alpha(\frac{1}{2}+\delta)}
\left[
\left(\frac{c_\ell^{2 - \alpha}}{\alpha^\alpha(c_\ell - \alpha)^{2 - 2\alpha}}\right)^{(\frac{1}{2}+\delta)\gamma}
\cdot
\left(\frac{c_1^{2-\beta}}{\beta^\beta(c_1-\beta)^{2-2\beta}}\right)^{(1-\gamma)}
\right]\!.}$

\item $\displaystyle{\max_{0\leq\alpha\leq1}\ \max_{(\frac{1}{2}+\delta+\alpha(\frac{1}{2}-\delta))\leq\beta\leq1}
\left[
\left(\frac{c_r^{2 - \alpha}}{\alpha^\alpha(c_r - \alpha)^{2 - 2\alpha}}\right)^{(\frac{1}{2}-\delta)\gamma}
\!\!\!\cdot
\!\left(\frac{c_2^{2-\beta}}{\beta^\beta(c_2-\beta)^{2-2\beta}}\right)^{(1-\gamma)}
\right]\!.}$
\end{enumerate}

{\noindent Let $\alpha^\ell_c$ be the $\alpha$ that maximizes $\displaystyle{\left(\frac{c_\ell^{2 - \alpha}}{\alpha^\alpha(c_\ell - \alpha)^{2 - 2\alpha}}\right)}$, $\alpha^r_c$ be the $\alpha$ that maximizes $\displaystyle{\left(\frac{c_r^{2 - \alpha}}{\alpha^\alpha(c_r - \alpha)^{2 - 2\alpha}}\right)}$, $\beta^1_c$ be the $\beta$ that maximizes $\displaystyle{\left(\frac{c_1^{2-\beta}}{\beta^\beta(c_1-\beta)^{2-2\beta}}\right)}$, and $\beta^2_c$ be the $\beta$ that maximizes $\displaystyle{\left(\frac{c_2^{2-\beta}}{\beta^\beta(c_2-\beta)^{2-2\beta}}\right)}$.
For any choice of values for $c_1,c_2,c_\ell$ and $c_r$ we should try below, we have that $(\frac{1}{2}+\delta)<\beta^1_c$ and $\beta^2_c<(\frac{1}{2}+\delta+\alpha^r_c(\frac{1}{2}-\delta))$. Denote $\alpha'=\max\{0,\frac{\beta^2_c-1/2-\delta}{1/2-\delta}\}$. Therefore, we can further bound $Y$ by the maximum of the following two expressions:}

\medskip
{\noindent $\mathrm{1.} \displaystyle{\max_{\alpha^\ell_c\leq\alpha\leq1}}$ of}
\vspace{-0.4em}
\[\displaystyle{\!\!
\left(\!\frac{c_\ell^{2 - \alpha}}{\alpha^\alpha(c_\ell - \alpha)^{2 - 2\alpha}}\!\right)^{\!(\frac{1}{2}+\delta)\gamma}
\!\!\cdot\!
\left(\!\frac{c_1^{2-\alpha(\frac{1}{2}+\delta)}}{(\alpha(\frac{1}{2}+\delta))^{\alpha(\frac{1}{2}+\delta)}(c_1-\alpha(\frac{1}{2}+\delta))^{2-2\alpha(\frac{1}{2}+\delta)}}\!\right)^{\!(1-\gamma)}}\]

{\noindent $\mathrm{2.} \displaystyle{\max_{\alpha'\leq\alpha\leq\alpha^r_c}}$ of}
\vspace{-0.4em}
\[\displaystyle{\!\!\!
\left(\!\frac{c_r^{2 - \alpha}}{\alpha^\alpha(c_r \!-\! \alpha)^{2 \!-\! 2\alpha}}\!\right)^{\!(\frac{1}{2}\!-\!\delta)\gamma}
\!\!\!\!\cdot\!
\left(\!\frac{c_2^{2-(\frac{1}{2}+\delta+\alpha(\frac{1}{2}-\delta))}}{(\frac{1}{2}\!+\!\delta\!+\!\alpha(\frac{1}{2}\!-\!\delta))^{(\frac{1}{2}\!+\!\delta\!+\!\alpha(\frac{1}{2}\!-\!\delta))}(c_2\!-\!(\frac{1}{2}\!+\!\delta+\alpha(\frac{1}{2}\!-\!\delta)))^{2\!-\!2(\frac{1}{2}\!+\!\delta\!+\!\alpha(\frac{1}{2}\!-\!\delta))}}\!\right)^{\!(1\!-\!\gamma)}}\]

\smallskip
{\noindent Denote the bounds in the first and second items by $Y_1$ and $Y_2$, respectively. Let, $Z_1=Y_1^k\cdot {\gamma k\choose (\frac{1}{2}-\delta)\gamma k} 2^{\frac{1}{10^{10}}k}$, $Z_2=Y_2^k\cdot {\gamma k\choose (\frac{1}{2}-\delta)\gamma k} 2^{\frac{1}{10^{10}}k}$ and $Z=\max\{Z_1,Z_2\}$. Note that running time of \alg{PathAlg} is bounded by $O^*(Z)$. Table \ref{tab:kpath}, given below, presents bounds for $Z$, corresponding to different choices of $\delta,\gamma,c_1,c_2,c_\ell$ and $c_r$. In particular, by choosing $\delta = 0.046,\gamma=0.084,c_1=1.504,c_2=1.398,c_\ell=1.092$ and $c_r=1.876$, we get the bound $2.59606$. In this case the maximum of $Z_1$ is obtained at $\alpha\cong0.908105$, where it is almost equal to 2.59606, and the maximum of $Z_2$ is obtained at $\alpha\cong0.123734$, where it is also almost equal to 2.59606.}

\begin{table}
[!ht]
\centering
\begin{tabular}{|c|c|c|c|}
	\hline
	$(\delta,\gamma,c_1,c_2,c_\ell,c_r)$    & $Z$             & $Z_1$     & $Z_2$     \\\hline\hline		
	$\bf (0.046,0.084,1.504,1.398,1.092,1.876)$ & {\bf 2.5960542} & {\bf 2.5960542} & {\bf 2.5960425} \\\hline
	$(0.045,0.084,1.504,1.398,1.092,1.876)$ &       2.5965734 & 2.5953152 & 2.5965734 \\\hline
	$(0.047,0.084,1.504,1.398,1.092,1.876)$ &       2.5967889 & 2.5967889 & 2.5955049 \\\hline
	$(0.046,0.083,1.504,1.398,1.092,1.876)$ &       2.5960903 & 2.5960421 & 2.5960903 \\\hline
	$(0.046,0.085,1.504,1.398,1.092,1.876)$ &       2.5960711 & 2.5960711 & 2.5959989 \\\hline
	$(0.046,0.084,1.503,1.398,1.092,1.876)$ &       2.5960547 & 2.5960547 & 2.5960425 \\\hline
	$(0.046,0.084,1.505,1.398,1.092,1.876)$ &       2.5960545 & 2.5960545 & 2.5960425 \\\hline
	$(0.046,0.084,1.504,1.397,1.092,1.876)$ &       2.5960542 & 2.5960542 & 2.5960430 \\\hline
	$(0.046,0.084,1.504,1.399,1.092,1.876)$ &       2.5960542 & 2.5960542 & 2.5960434 \\\hline
	$(0.046,0.084,1.504,1.398,1.091,1.876)$ &       2.5960544 & 2.5960544 & 2.5960425 \\\hline
	$(0.046,0.084,1.504,1.398,1.093,1.876)$ &       2.5960545 & 2.5960545 & 2.5960425 \\\hline
	$(0.046,0.084,1.504,1.398,1.092,1.875)$ &       2.5960542 & 2.5960542 & 2.5960425 \\\hline
	$(0.046,0.084,1.504,1.398,1.092,1.877)$ &       2.5960542 & 2.5960542 & 2.5960425 \\\hline	
\end{tabular}\medskip
\caption{Upper bounds for $Z$, $Z_1$ and $Z_2$, corresponding to different choices of $\delta,\gamma,c_1,c_2,c_\ell$ and $c_r$.}
\label{tab:kpath}
\end{table}

\section{Solving the $(3,k)$-WSP Problem}\label{section:3kwsp}

In this section we develop a deterministic FPT algorithm for {\sc $(3,k)$-WSP}, following the fourth strategy in Section~\ref{section:strategies}. We first solve a subcase of {\sc $(3,k)$-WSP}, and then show how to translate {\sc $(3,k)$-WSP} to this subcase by using unbalanced cutting of the universe. 

\subsection{Intuition}\label{sec:wspIntuition}

In this section, we attempt to describe the intuition that guided us through the development of the algorithm.

\myparagraph{Previous Algorithm}
We first give a brief overview of the algorithm for {\sc $(3,k)$-WSP} of \cite{corrmatchpack} since our algorithm builds upon it, and it will allow us to explain the intuition behind our strategy. To this end, we need the following observation of \cite{predivandcol}, assuming an arbitrary order on the elements in $U$:

\begin{obs}\label{obs:wspobs}
Let ${\cal S}'\subseteq{\cal S}$, and denote $S_\mathrm{min} = \{u\in U: \exists S\in{\cal S}'$ in which $u$ is the smallest element$\}$. Then, any $S\in {\cal S}$ whose smallest element is larger than $\mathrm{max}(S_\mathrm{min})$ does not contain any element from $S_\mathrm{min}$.
\end{obs}

The algorithm of \cite{corrmatchpack} iterates over $U$ in an ascending order, such that when it reaches an element $u\in U$, it has already computed representative families of families of partial solutions (each partial solution corresponds to a family of disjoint 3-sets from ${\cal S}$) that include only sets from ${\cal S}$ whose smallest elements are smaller than $u$. Then, it tries to extend the partial solutions by adding sets whose smallest element is $u$; this is followed by computations of representative families to reduce the size of the resulting families. By Observation \ref{obs:wspobs}, the elements in $U$ that are the smallest elements of 3-sets in the partial solutions do not appear in any 3-set whose smallest element is at least $u$. This allows the algorithm to delete the smallest elements of 3-sets after adding them to partial solutions (since it will not encounter them again), which results in faster computations of representative families (since we are handling partial solutions of smaller sets) that overall improve the running time of the algorithm.

\medskip
\myparagraph{Our Improvement}
We introduce a technique, unbalanced cutting of the universe,\footnote{The name unbalanced cutting intends to imply that in this technique, we repeatedly attempt to cut a small piece from the universe that contains as many elements as possible from a certain set, which concerns the solution and is therefore unknown in advance. This will allow us to delete elements from partial solutions in a manner that distorts the balance between the number of small and large elements  they~contain.} that allows us to ensure that we delete a set of elements that is larger than the set that contains only each element that is the smallest in a 3-set that belongs to a partial solution. Clearly, in the previous algorithm, we can delete from every 3-set not only its smallest element (which is smaller than $u$), but also any other element that that is smaller than $u$. Therefore, it is apparent that we need to find some special method of ordering the elements (in a preprocessing phase) such that when we reach an element $u$ (in the iteration) we can only focus on partial solutions $\cal P$ such that the union of the 3-sets in $\cal P$ contains significantly more than $|{\cal P}|/3$ elements that are smaller than $u$. To this end, we {\em explicitly} cut the universe into a {\em constant} number of small parts, order them, and using this order, define an order on the entire universe. Then, we will show that we have a constant number of stages in the algorithm, where at each stage we are ``given'' an element such that we can delete all the elements that are smaller than it. A lower bound on the number of elements that we should delete will be computed using a certain recursive formula---if a partial solution does not contain enough elements that we can delete (according to the formula), we will remove it (we will be able to show that this removal does not affect the correctness of the algorithm).

Now, we give an intuitive, informal explanation of the manner in which unbalanced cutting works, and why it is correct. Let $u_1,u_2,\ldots,u_n$ be an arbitrary order of the elements in $U$. We will attempt to cut $\frac{1}{\epsilon}$ ($\epsilon$ will be a parameter to be determined in our algorithm) pieces from the universe in a certain order, and thus obtain $\frac{1}{\epsilon}+1$ small universes (the last universe will simply contain all the elements that we did not insert to one of the universes that we cut from $U$). Now, the elements will be reordered as follows: the smallest will be those in the first piece that we cut, afterwards will be those in the second piece that we cut, and so on. The order between the elements in a certain piece will be defined by the arbitrary order that we chose (i.e., $u_1,u_2,\ldots,u_n$).

We will actually consider all the options to cut the universe (by an exhaustive search), and thus we need to ensure that there are not too many such options (we will get a polynomial number of options), and that we will consider at least one ``good'' option (we will discuss the meaning of ``good'' later). Therefore, we suggest that an option should correspond to a result of the following process. First, one guesses two elements, $\ell_1<r_1$, and let the first piece, $U^1$, contain them and all the elements between them. Now, we remove the first piece from $U$. Considering the part that is left from $U$, again, one guesses two elements, $\ell_2<r_2$, and let the second piece, $U^2$, contain them and all the elements between them (that were not already inserted to $U^1$). Then, we also remove the second piece from $U$. The process continues in this manner, until we have cut $\frac{1}{\epsilon}$ pieces. Clearly, the number of options to cut the universe is smaller than $|U|^{\frac{2}{\epsilon}}=O^*(1)$.

The stages in the algorithm correspond to the small universes---we will have $\frac{1}{\epsilon}+1$ stages, such that at stage $2\leq i\leq\frac{1}{\epsilon}+1$, the element that we are ``given'' is the largest element in the piece $U^{i-1}$.\footnote{At the first stage, we can assume that the element that we are given is simply the smallest element in the universe, since we cannot compute a useful element for the first stage; this is not important, since the pieces are {\em very} small, and the individual effect of a piece in this context is negligible.} We still iterate over $U$ in an ascending order (but now this is performed according to the new order---corresponding to the currently examined option that cut the universe), such that at stage $i$, when inserting a 3-set to partial solutions, we can remove its smallest element. However, at stage $i$, we can also remove from our partial solutions all the elements that are smaller than the given element; there should be enough such elements at each partial solution (according to a recursive formula that we discuss later), since if there are not, we simply discard the partial solution. Moreover, overall at stage $i$ (that is not the last stage), we are supposed to insert exactly $\lfloor\epsilon k\rfloor$ 3-sets to each of our partial solutions (the 3-sets are of course inserted one-by-one, in order to embed a computation of a representative family ``as soon as possible''). We are also not supposed to insert 3-sets that contain elements that are smaller than the element that we were given at the previous stage (otherwise we could not have deleted elements that were smaller than it).

Now, we need to argue that we do consider a ``good'' option---this means that there is an option such that that by following the dynamic programming-based approach described above, we can construct a solution. Here we will not discuss this in detail, but only give some informal explanations that allow to intuitively understand the general approach that our proof will employ.\footnote{Exact, formal definitions and proofs are given in the following sections} To this end, consider some solution ${\cal S}'$.
Let $U^1$ be a set of consecutive elements in $U$ that captures elements from exactly $\lfloor\epsilon k\rfloor$ 3-sets in ${\cal S}'$ (i.e., there are exactly $\lfloor\epsilon k\rfloor$ 3-sets in ${\cal S}'$ that have, each, at least one element contained in $U^1$). Let ${\cal X}$ be the family of $\lfloor\epsilon k\rfloor$ 3-sets from ${\cal S}'$ that we used (i.e., that have elements that belong to $U^1$). Moreover, let $Y$ be the set of elements in the 3-sets in $\cal X$ that are not smaller than the largest element in $U^1$; also, $Y$ should not contain any element that is the smallest in a 3-set in $\cal X$. We remove the elements of $U^1$ from $U$.\footnote{This is done since the pieces should be disjoint; also, note that in the remaining universe, the element ordered exactly before the smallest element of $U^1$ and the element ordered exactly after its largest element are considered to be consecutive.} We need to define $U^2$ so that again, it captures $\lfloor\epsilon k\rfloor$ new 3-sets in the manner described above (more precisely, it should capture elements from exactly $\lfloor\epsilon k\rfloor$ 3-sets in ${\cal S}'\setminus{\cal X}$). However, we also want to ensure that $U^2$ captures as many elements as possible from $Y$; thus, when we will be given its largest element (when the second stage, corresponding to $U^2$, {\em ends}), we will be able to delete many elements from partial solutions (those are the elements it captures from $Y$). At worst, in the second stage, we could not remove any extra element from the sets in ${\cal S}'$ by using the element that was given to us (i.e., the largest element of $U^1$), and therefore $|Y|=2\lfloor\epsilon k\rfloor$. There are $3(k-\lfloor\epsilon k\rfloor)$ elements in the 3-sets in ${\cal S}'\setminus{\cal X}$, and therefore we have at most $y=3(k-\lfloor\epsilon k\rfloor)/\lfloor\epsilon k\rfloor$ candidates for $U^2$. Thus, by the pigeonhole principle, we have an option that captures at least $|Y|/y$ elements from $Y$. This will be the lower bound for the number of extra elements that we expect to delete up to the point when we move from the second stage to the third stage. The above process (of defining $U^2$) continues in a similar manner towards defining $U^3$, and so on.\footnote{In the algorithm, the elements that we are ``given'' (i.e., the largest elements in the pieces that we cut), will be specified using a function $f$.}

Finally, we would like to mention that one gets a recursive formula that indicates the number of extra elements that should be deleted up to a certain stage, since we need to consider the maximum number of elements (to be deleted) that we can capture in a piece $U^i$ and which were not already captured in one of the {\em previous} pieces. That is, when considering the minimum number of extra elements that should be deleted up to a certain stage $i$, we need to consider a lower bound on the number of extra elements that should have been already deleted up to the previous stage, $i-1$.

\subsection{A Subcase of $(3,k)$-WSP}\label{section:subcasewsp}

We now consider {\sc Cut $(3,k)$-WSP ($(3,k)$-CWSP)}, the subcase of {\sc $(3,k)$-WSP} that we solve using a representative sets-based procedure. Recall from the previous section that informally, in defining this subcase, the main idea is to introduce a function that cuts a solution in a manner that allows us, while executing a procedure that builds upon the algorithm of \cite{corrmatchpack}, to delete more elements than just those that are the smallest in the inserted sets. More precisely, in the algorithm, we will have a fixed, though large, number of locations where we are ``given'' an element that indicates that from now on, we can delete all elements that are smaller than it.

\subsubsection{Definition}\hspace{1em}$~$\\

{\noindent Fix $0<\epsilon<0.1$, whose value is determined later, such that $\frac{1}{\epsilon}\in\mathbb{N}$. Recall that informally, $\frac{1}{\epsilon}$ is the number of pieces that we cut (and thus we obtain $\frac{1}{\epsilon}+1$ universes).}

\myparagraph{Input} Formally, the input for {\sc $(3,k)$-CWSP} consists of an {\em ordered} universe $U=\{u_1,u_2,\ldots,u_n\}$, where $u_{i-1}<u_i$ for all $i\in\{2,3,\ldots,n\}$. We are also given a family $\cal S$ of subsets of size 3 of $U$, a weight function $w: {\cal S}\rightarrow \mathbb{R}$, a weight $W\in\mathbb{R}$ and a parameter $k\in\mathbb{N}$, along with a non-decreasing function $f: \{1,2,\ldots,\frac{1}{\epsilon}\}\rightarrow U$. 

\myparagraph{Objective} We need to decide if there is an {\em ordered} subfamily ${\cal S}'\subseteq {\cal S}$ of $k$ disjoint sets, denoted accordingly as ${\cal S}'=\{S_1,S_2,\ldots,S_k\}$, whose total weight is at least $W$, such that the following ``solution conditions'' are satisfied:
\begin{enumerate}
\item $\forall i\in\{2,3,\ldots,k\}$: $\min(S_{i-1})<\min(S_i)$ (i.e., the smallest element in $S_{i-1}$ is smaller than the smallest one in $S_i$).

\item $\forall i\in\{1,2,\ldots,\frac{1}{\epsilon}\}$: There are at least
$R(i)$
elements in $\displaystyle{\bigcup_{j=1}^{i\lfloor\epsilon k\rfloor}(S_j\setminus\{\min(S_j)\})}$ that are smaller or equal to $f(i)$, where $R$ is defined according to the following recursion.
	\bigskip
	\begin{itemize}
	\item $R(0)=R(1)=0$.\footnote{We define $R(0)$ since it will simplify the proof of our algorithm (note that we refer to $R(0)$ in the running time in Lemma \ref{lemma:cwsprun}).}
	\item For $j=2,\ldots,\frac{1}{\epsilon}$: $\displaystyle{R(j) = R(j-1) +
	\left\lceil\frac{2(j-1)\lfloor\epsilon k\rfloor
	-R(j-1)}
	{\lceil 3(k-(j-1)\lfloor\epsilon k\rfloor)/\lfloor \epsilon k\rfloor\rceil}\right\rceil}$.
	\end{itemize}

\item $\forall i\in\{1,2,\ldots,\frac{1}{\epsilon}\}$: All the elements in $\displaystyle{\bigcup_{j=1+i\lfloor\epsilon k\rfloor}^{k}S_j}$ are larger than $f(i)$.
\end{enumerate}

\myparagraph{More Intuition} Roughly speaking, in the algorithm, we shall attempt to construct ${\cal S}'$ by finding its sets according to their order (i.e., we first find $S_1$, then $S_2$, and so on). Now, following the intuition described in the previous section, the first condition is meant to ensure that once we insert a set $S\in {\cal S}$ to a partial solution, we can remove its smallest element (since next we will only attempt to insert sets that contain only larger elements); the third condition is meant to ensure that once we finish stage $i$ (and move to stage $i+1$), we can remove from our partial solutions all of the elements that are at most $f(i)$; the second condition is meant to ensure that there are enough extra elements that we can remove (by relying on $f(i)$).

\subsubsection{Solving the Subcase}\hspace{1em}$~$\\

{\noindent We next show that {\sc $(3,k)$-CWSP} can indeed be efficiently solved using representative sets, proving the following lemma:}

\begin{lemma}\label{lemma:cwsprun}
For any fixed $c\geq 1$ and $0<\epsilon<1$, {\sc $(3,k)$-CWSP} can be solved in deterministic time $Y=\displaystyle{O^*(2^{o(k)}\cdot
\max_{i=1}^{\frac{1}{\epsilon}}\max_{j=1+(i-1)\lfloor\epsilon k\rfloor}^{i\lfloor\epsilon k\rfloor}X)}$, where

\begin{itemize}
\item $X=\displaystyle{\frac{(c(3k\!-\!(j\!-\!1)-R(i\!-\!i')))^{3k-(j-1)-R(i-i')}}
{(2(j\!-\!1)\!-\!R(i\!-\!i'))
^{2(j-1)-R(i-i')}
\cdot\widetilde{X}}
\ \cdot}$
\[\displaystyle{(\frac{c(3k-j-R(i-1))}{c(3k-j-R(i-1))-(2j-R(i-1))})^{(3k-j-R(i-1))-(2j-R(i-1))},}\]

\item $\widetilde{X}=(c(3k-(j-1)-R(i-i'))-(2(j-1)-R(i-i')))
^{(3k-(j-1)-R(i-i'))-(2(j-1)-R(i-i'))}$.

\bigskip
\item $i'\!\in\!\{1,2\}$ such that $[i'\!=\!2\leftrightarrow j\!=\!1\!+\!(i\!-\!1)\lfloor\epsilon k\rfloor]$.
\end{itemize}
\end{lemma}

\begin{proof}
When we next refer to representative families, suppose that $E=U$. We now present a standard dynamic programming-based procedure to prove the lemma, in which we embed computations of representative sets. To this end, we use a matrix M that has an entry $[i,j,s_1,s_2,\ldots,s_{\frac{1}{\epsilon}},m]$ for all the parameters that satisfy the following conditions:
\begin{enumerate}
\item $i\in\{1,2,\ldots,\frac{1}{\epsilon}+1\}$.
\item $j\in \{1+(i-1)\lfloor\epsilon k\rfloor,\ldots,J\}$ where $(i\leq\frac{1}{\epsilon}\rightarrow J=i\lfloor\epsilon k\rfloor)\wedge(i=\frac{1}{\epsilon}+1\rightarrow J\!=\!k)$.
\item $\forall\ell\in\{1,\ldots,\frac{1}{\epsilon}\}: s_\ell\in\{0,1,\ldots,2j\}$. 
\item $\forall\ell\in\{1,\ldots,i-1\}: s_\ell\geq R(\ell)$.
\item $m\in\{j,\ldots,n\}$ such that ($i>1\rightarrow u_m>f(i-1)$).
\end{enumerate}

Informally, following the notions described in Appendix \ref{sec:wspIntuition}, $i$ specifies the current stage in the algorithm, $j$ specifies the number of 3-sets that are related to each partial solution in the family stored in the entry, each $s_\ell$ specifies the number of extra elements that should be deleted once we finish stage $\ell$ by using the element that we are ``given'' at that point (i.e., by using $f(\ell)$),\footnote{In the fourth condition, we do not include $\ell=i$, since the stage $i$ is not yet finished---only when it is finished, we need to ensure that we can delete at least $R(i)$ elements.} and $m$ specifies the smallest element in the 3-set that was the last one inserted to the partial solutions in the family stored in the entry.

We next assume that a reference to an undefined entry returns $\emptyset$. The other entries will store the following families of partial solutions, where we assume that we track the weights of the partial solutions:\footnote{When computing an entry and obtaining the same partial solution from several different families of sets, we store the maximal weight.}
\begin{itemize}
\item M$[i,j,s_1,\ldots,s_{\frac{1}{\epsilon}},m]$: A family that max $3(k-j)$-represents the family $F$ defined as follows. Let $\widetilde{F}$ be the family of every ordered disjoint sets $S_1,S_2,\ldots,S_j$ such that:
	\begin{enumerate}
	\medskip
	\item $\forall \ell\in\{2,3,\ldots,j\}$: $\min(S_{\ell-1})<\min(S_\ell)$.
	\medskip
	\item $\min(S_j)=u_m$.
	\item $\forall \ell\in\{1,2,\ldots,i-1\}$: There are at least $R(\ell)$
elements in $\displaystyle{\bigcup_{p=1}^{\ell\lfloor\epsilon k\rfloor}(S_p\!\setminus\!\{\min(S_p)\})}$ that are at most $f(\ell)$.\footnote{Observe that this condition is not covered by the next one, since $\ell\lfloor\epsilon k\rfloor\neq j$ (those are the maximum values for $p$ in these conditions).}
	\item $\forall \ell\in\{1,2,\ldots,\frac{1}{\epsilon}\}$: There are exactly $s_\ell$ elements in $\displaystyle{\bigcup_{p=1}^{j}(S_p\setminus\{\min(S_p)\})}$ that are at most $f(\ell)$.
	\item $\forall \ell\in\{1,2,\ldots,i-1\}$: All the elements in $\displaystyle{\bigcup_{p=1+\ell\lfloor\epsilon k\rfloor}^{j}\!\!\!S_p}$ are larger than~$f(\ell)$.
	\end{enumerate}

\smallskip	
Define $f(0)$ as a value smaller than $u_1$. Then, for each such sets $S_1,S_2,\ldots,S_j$ in $\widetilde{F}$, $F$ contains $\displaystyle{(\bigcup_{\ell=1}^j(S_\ell\setminus\{\min(S_\ell)\}))\setminus\{u_\ell\in U: u_\ell\leq f(i\!-\!1)\}}$.
\end{itemize}

\medskip
{\noindent The entries are computed in the following order:}
\begin{itemize}
\item For $i=1,\ldots,\frac{1}{\epsilon}+1$:
	\begin{itemize}
	\item For $j=1+(i-1)\lfloor\epsilon k\rfloor,\ldots,J$:
		\begin{itemize}
				\item Compute all entries of the form M$[i,j,s_1,\ldots,s_{\frac{1}{\epsilon}},m]$.
		\end{itemize}
	\end{itemize}
\end{itemize}

\smallskip
\medskip
{\noindent We now give the recursive formulas using which the entries are computed.}
\begin{enumerate}
\item If $j\!=\!1$: M$[1,1,s_1,\ldots,s_{\frac{1}{\epsilon}},m] = \{S\!\in\!{\cal S}: \min(S)\!=\!u_m, [\forall\ell\in\{1,\ldots,\frac{1}{\epsilon}\}: |\{u_p\!\in\! S\!\setminus\!\{u_m\}: u_p\!<\!f(\ell)\}|\!=\!s_\ell]\}$.

\bigskip
\item Else: M$[i,j,s_1,\ldots,s_{\frac{1}{\epsilon}},m]=$\\
$\displaystyle{\{S\!\cup\! A: S\!\cap\! A\!=\!\emptyset, S\!\in\!{\cal S}, \min(S)\!=\!u_m, A\!\in\!\!\!\bigcup_{m'=1}^{m-1}\bigcup_{i'=0}^1\mathrm{M}[i\!-\!i',j\!-\!1,s_1',\!...,s_{\frac{1}{\epsilon}}',m']\}}$,\\
where [$\forall\ell\in\{1,\ldots,\frac{1}{\epsilon}\}: s_\ell' = s_\ell-|\{u_p\in S\setminus\{u_m\}: u_p\leq f(\ell)\}|$].
\end{enumerate}

\begin{itemize}
\item After 2: 
	\begin{enumerate}
	\item If $i>1$: Remove from each set in M$[i,j,s_1,\ldots,s_{\frac{1}{\epsilon}},m]$ the elements that are at most $f(i-1)$.
	\item Replace the result by a family that max $3(k-j)$-represents it.
	\end{enumerate}
\end{itemize}

{\noindent Finally, we return yes {\em iff} at least one entry of the form M$[i,k,s_1,\ldots,s_{\frac{1}{\epsilon}},m]$ contains a solution of weight at least $W$. Observe that Theorem \ref{theorem:esarep} ensures that the procedure can be performed in the desired time. Indeed, there is a polynomial number of entries (since $\frac{1}{\epsilon}=O(1)$), and it straightforward to verify that the time required to compute entry M$[i,j,s_1,\ldots,s_{\frac{1}{\epsilon}},m]$ is bounded by $O^*(X\cdot 2^{o(k)})$.}\qed
\end{proof}

\subsubsection{A More Careful Running Time Analysis}\hspace{1em}$~$\\

{\noindent We proceed to analyze the bound, denoted $Y$, for the running time in Lemma \ref{lemma:cwsprun}. First, define $T$ according to the following recursion:}
\begin{itemize}
\item $T(0)=T(1)=0$.
\item For $i=2,\ldots,\frac{1}{\epsilon}$: $\displaystyle{T(i) = T(i-1) +
	\epsilon \left(\frac{2(j-1)\epsilon
	-T(i-1)}
	{3(1-(i-1)\epsilon)}\right)}$.
\end{itemize}

\smallskip
{\noindent For the sake of readability of the formulas below, denote $\widehat{T}(i)=T(i-1)$. Thus, we get that $Y=O^*(2^{o(k)}\cdot \max_{i=1}^{\frac{1}{\epsilon}} X')$, where $X'$ is equal to}
\[\displaystyle{
\max_{(i-1)\epsilon k\leq j\leq i\epsilon k}\!\!
\left(\!\!
\frac{(c(3k-j-\widehat{T}(i)k))^{2(3k-j-\widehat{T}(i)k)-(2j-\widehat{T}(i)k)}}{(2j\!-\!\widehat{T}(i)k)^{2j-\widehat{T}(i)k}(c(3k\!-\!j\!-\!\widehat{T}(i))k)\!-\!(2j\!-\!\widehat{T}(i)k))^{2(3k-j-\widehat{T}(i)k)-2(2j-\widehat{T}(i)k)}}\!\!\right)
}\!.\]

\smallskip
{\noindent Now, we can further bound $X'$ by the following expression:}

\[\displaystyle{\max_{(i-1)\leq\alpha\leq i}
\left(\!
\frac{(c(3-\alpha\epsilon-\widehat{T}(i)))^{2(3-\alpha\epsilon-\widehat{T}(i))-(2\alpha\epsilon-\widehat{T}(i))}}{(2\alpha\epsilon-\widehat{T}(i))^{2\alpha\epsilon-\widehat{T}(i)}(c(3-\alpha\epsilon-\widehat{T}(i))-(2\alpha\epsilon-\widehat{T}(i)))^{2(3-\alpha\epsilon-\widehat{T}(i))-2(2\alpha\epsilon-\widehat{T}(i))}}
\!\right)^k
}\!\!\!.\]

{\noindent Therefore, $Y$ is bounded by $O^*$ of:}

\[\displaystyle{2^{o(k)}\cdot
\max_{i=1}^{\frac{1}{\epsilon}}\max_{(i-1)\leq\alpha\leq i}
\left(
\frac{(c(3-\alpha\epsilon-\widehat{T}(i)))^{6-4\alpha\epsilon-\widehat{T}(i)}}{(2\alpha\epsilon-\widehat{T}(i))^{2\alpha\epsilon-\widehat{T}(i)}(c(3-\alpha\epsilon-\widehat{T}(i))-2\alpha\epsilon+\widehat{T}(i))^{6-6\alpha\epsilon}}
\right)^k
}\!\!\!.\]

{\noindent To allow us to compute a bound efficiently, we choose $\epsilon=10^{-5}$. Choosing a smaller $\epsilon$ results in a better bound, but the improvement is negligible. Table \ref{tab:cwsp}, given below, presents bounds for $Y$, corresponding to different choices of $c$. In particular, by choosing $c=1.591$, we get the bound $O^*(8.097^k)$.}

\begin{table}
[!ht]
\centering
\begin{tabular}{|c|c|c|c|}
	\hline
	$c$         & $Y$                   & $i$         & $T(i-1)$       \\\hline\hline	
	$1.59$     &     $O^*(8.096400^k)$ & 54511       & 0.1476545       \\\hline		
	$\bf 1.591$ & $\bf O^*(8.096396^k)$ & {\bf 54515} & {\bf 0.1476821} \\\hline
	$1.592$     &     $O^*(8.096397^k)$ & 54518       & 0.1477028       \\\hline		
\end{tabular}\medskip
\caption{Upper bounds for $Y$, corresponding to different choices of $c$, where $\epsilon=10^{-5}$. The second and third entries specify approximate values for $i$ and $T(i-1)$ that correspond to the maximum.}
\label{tab:cwsp}
\end{table}

{\noindent We thus obtain the following corollary:}

\begin{cor}\label{cor:cwsp}
For $c=1.591$ and $\epsilon=10^{-5}$, {\sc $(3,k)$-CWSP} can be solved in deterministic time $O^*(8.097^k)$.
\end{cor}

\subsection{A Deterministic Algorithm for $(3,k)$-WSP}\label{section:wspdet}

Fix $c=1.591$ and $\epsilon=10^{-5}$. We now present the pseudocode of \alg{WSPAlg}, a deterministic algorithm for {\sc $(3,k)$-WSP}, followed by informal explanations.

\begin{algorithm}[!ht]
\caption{\alg{WSPAlg}($U,{\cal S},w,W,k$)}
\begin{algorithmic}[1]
\STATE\label{step:wspalgorder} order the universe $U$ arbitrarily as $U=\{u_1,u_2,\ldots,u_n\}$.

\FORALL{distinct $\ell_1,\ell_2,\ldots,\ell_{\frac{1}{\epsilon}},r_1,r_2,\ldots,r_{\frac{1}{\epsilon}}\in U$ s.t. $\ell_i<r_i$ for all $i\in\{1,2,\ldots,\frac{1}{\epsilon}\}$}\label{step:orderlr}
	\FOR{$i=1,2,\ldots,\frac{1}{\epsilon}$}
		\STATE\label{step:orderUi} define $U^i=\{{\ell_i},{\ell_i+1},\ldots,{r_i}\}\setminus \bigcup_{j=1}^{i-1}U^j$.
	\ENDFOR\label{step:orderUi2}
	\STATE\label{step:orderUlast2} $U^{\frac{1}{\epsilon}+1} = U\setminus \bigcup_{j=1}^{\frac{1}{\epsilon}}U^j$.
	\STATE\label{step:orderU} let $U'=\{u'_1,u'_2,\ldots,u'_n\}$ be an ordered copy of $U$ such that the elements in $U^1$ appear first, then those in $U^2$, and so on, where the order between the elements in each $U^i$ is preserved according to their order in $U$.
	\STATE\label{step:orderf} define $f:\{1,2,\ldots,\frac{1}{\epsilon}\}\rightarrow U'$ such that $f(i)$ is the largest element in $U'$ that belongs to $U^i$.
	\STATE\label{step:wspaccept} {\bf if} the procedure of Corollary \ref{cor:cwsp} accepts $(U',{\cal S},w,W,k,f)$ {\bf then} accept. {\bf end if}
\ENDFOR


\STATE reject.
\end{algorithmic}
\end{algorithm}

Algorithm \alg{WSPAlg} first orders the elements in $U$ (Step 1). Then, it exhaustively examines every option to {\em explicitly} cut $\frac{1}{\epsilon}$ pieces from the universe, such that a piece consists of elements that are consecutively ordered in $U$ when considering only those left after removing the elements that belong to pieces we have already defined (see Steps \ref{step:orderlr}--\ref{step:orderUi2}). More precisely, a piece $U^i$ is defined to contain all the elements between $\ell_i$ and $r_i$ that are not contained in any piece $U^j$ for $j<i$ (Step \ref{step:orderUi}). By the outer loop (Step \ref{step:orderlr}), we consider all possible choices for such values $\ell_i$ and $r_i$. Note that there are only $O(|U|^{O(\frac{1}{\epsilon})})=O^*(1)$ such options.
Another piece, $U^{\frac{1}{\epsilon}+1}$, simply consists of the elements in $U$ that were not chosen for the previous pieces (Step \ref{step:orderUlast2}).
The order between the pieces is defined according to the order in which they were defined, which, in turn, defines a reordering of the entire universe $U$ (in Step \ref{step:orderU}). The function $f$ is also defined according to the order between the pieces (see Step \ref{step:orderf}), assigning, in an ascending order, the last element of each piece. Overall, \alg{WSPAlg} uses unbalanced cutting of the universe to obtain a set of inputs for {\sc $(3,k)$-CWSP}, and accepts (in Step \ref{step:wspaccept}) {\em iff} at least one of them is a yes-instance.

We next consider the correctness and running time of \alg{WSPAlg}. We thus prove the correctness of Theorem \ref{theorem:WSP} (in Section \ref{section:strategies}).

\begin{theorem}\label{theorem:wspalgcor}
\alg{WSPAlg} solves {\sc $(3,k)$-WSP} in time $O^*(8.097^k)$.
\end{theorem}

\begin{proof}
The running time of the algorithm follows immediately from the pseudocode and Corollary \ref{cor:cwsp}.

For the easier direction, note that for any ordering of the universe $U$ as $U'$, and for any function $f:\{1,2,\ldots,\frac{1}{\epsilon}\}$, a solution to the instance $(U',{\cal S},w,W,k,f)$ of {\sc $(3,k)$-CWSP} is also a solution to the instance $(U,{\cal S},w,W,k)$ of {\sc $(3,k)$-WSP}. Therefore, if the algorithm accepts, the input is a yes-instance of {\sc $(3,k)$-WSP}.

Now, suppose that the input is a yes-instance of {\sc $(3,k)$-WSP}, and let ${\cal S}'$ be a corresponding solution. Let $U=\{u_1,u_2,\ldots,u_n\}$ be the order chosen by \alg{WSPAlg} in Step \ref{step:wspalgorder}. We define $U^1,U^2,\ldots,U^{\frac{1}{\epsilon}}$, $f$ and an order $S_1,S_2,\ldots,S_k$ of the sets in ${\cal S}'$ as follows.

\begin{enumerate}
\item Initialize:
	\begin{enumerate}
	\smallskip
	\item Let $S_1,S_2,\ldots,S_{\lfloor\epsilon k\rfloor}$ be the sets in ${\cal S}'$ satisfying [$\min(S_1)\!<\!\min(S_2)\!<\!\ldots\!<\!\min(S_{\lfloor\epsilon k\rfloor})$] and [$\min(S_{\lfloor\epsilon k\rfloor})\!<\!\min(S)$] for all $S\!\in\!{\cal S}'\setminus\{S_1,S_2,\ldots,S_{\lfloor\epsilon k\rfloor}\}$.
	
	\smallskip
	\item Let $U^1=\{u_1, u_2, \ldots, \min(S_{\lfloor\epsilon k\rfloor})\}$ .
	
	\smallskip
	\item Let $f(1)=\min(S_{\lfloor\epsilon k\rfloor})$. Note that all the elements in $(\bigcup{\cal S}')\setminus(\bigcup_{j=1}^{\lfloor\epsilon k\rfloor}S_j)$ are larger than $f(1)$.
	
	\smallskip
	\item Let $U^1_{ord}=\{u'_1,u'_2,\ldots,u'_n\}$ be an ordered copy of $U$ such that the elements in $U^1$ appear first, where the internal order between the elements in $U^1$, as well as the internal order between the elements in $U\setminus U^1$, are preserved according to their order in $U$.
	
	\smallskip
	\item Let $P_1=\emptyset$.
	\end{enumerate}

\medskip	
\item For $i=2,3,\ldots,\frac{1}{\epsilon}$:
	\begin{enumerate}
	\smallskip
	\item Let $A' = \bigcup_{j=1}^{(i-1)\lfloor\epsilon k\rfloor}S_j$, and $A'_{min} = \bigcup_{j=1}^{(i-1)\lfloor\epsilon k\rfloor}\{\min(S_j)\}$ where $\min$ is computed according to the ordered universe $U^{i-1}_{ord}$. 
		
	\smallskip
	\item Let $B'=(\bigcup{\cal S}')\setminus(\bigcup_{j=1}^{(i-1)\lfloor\epsilon k\rfloor}S_j)$.
	
	\smallskip
	\item Let $B^1,B^2,\ldots,B^x$, where $x=\lceil\frac{|B'|}{\lfloor\epsilon k\rfloor}\rceil$, be a partition of $U^{i-1}_{ord}\setminus(\bigcup_{j=1}^{i-1}U^j)$ (into $x$ disjoint universes) such that each $B^i$ is a set of elements that are consecutively ordered in $U^{i-1}_{ord}$ and contains exactly $\lfloor\epsilon k\rfloor$ elements from $B'$, except $B^x$, if $(|B'|\!\ \mathrm{mod}\!\ \lfloor\epsilon k\rfloor)\neq 0$, which contains exactly $(|B'|\!\ \mathrm{mod}\!\ \lfloor\epsilon k\rfloor)$ elements from $B'$.
	
	\smallskip
	\item Let $U^i$ be a universe that contains the maximum number of elements from $A'\setminus (A'_{min}\cup P_{i-1})$ among the universes $B^1,B^2,\ldots,B^x$. If $U^i$ contains elements from less than $\lfloor\epsilon k\rfloor$ sets in $({\cal S}'\setminus\{S_1,S_2,\ldots,S_{(i-1)\lfloor\epsilon k\rfloor}\}$, add elements from $U^{i-1}_{ord}\setminus(\bigcup_{j=1}^{i-1}U^j)$ to $U^i$ such that its elements remain consecutively ordered in $U^{i-1}_{ord}$ and it contains elements from exactly $\lfloor\epsilon k\rfloor$ sets in $({\cal S}'\setminus\{S_1,\ldots,S_{(i-1)\lfloor\epsilon k\rfloor}\}$.
	
	\smallskip
	\item Let $f(i)$ be the largest element in $U^i$.

	\smallskip
	\item Let $U^i_{ord}=\{u'_1,u'_2,\ldots,u'_n\}$ be an ordered copy of $U$ such that the elements in $\bigcup_{j=1}^{i-1}U^j$ appear first, followed by those in $U^i$, where the internal orders between the elements in $\bigcup_{j=1}^{i-1}U^j$, $U^i$ and $U\setminus (\bigcup_{j=1}^{i}U^j)$ are preserved according to their order in $U^{i-1}_{ord}$.
	
	\smallskip
	\item Let $S_{(i-1)\lfloor\epsilon k\rfloor+1},S_{(i-1)\lfloor\epsilon k\rfloor+2},\ldots,S_{i\lfloor\epsilon k\rfloor}$ be the sets in ${\cal S}'\setminus\{S_1,S_2,\ldots,$ $S_{(i-1)\lfloor\epsilon k\rfloor}\}$ that contain elements from $U^i$, such that $\min(S_{j-1})<\min(S_j)$ (according to $U^i_{ord}$).
	
	\smallskip
	\item Note that all the elements in $(\bigcup{\cal S}')\setminus(\bigcup_{j=1}^{i\lfloor\epsilon k\rfloor}S_j)$ are larger than $f(i)$ (according to $U^i_{ord}$).
	
	\smallskip
	\item Let $P_i$ be the set of elements in $\bigcup_{j=1}^{i\lfloor\epsilon k\rfloor}(S_j\setminus\{\min(S_j)\})$ that are smaller or equal to $f(i)$ (according to $U^i_{ord}$). We have that:\\ $\displaystyle{|P_i|\geq |P_{i-1}| \!+\! \left\lceil\frac{|A'\setminus (A'_{min}\cup P_{i-1})|}{x}\right\rceil \geq R(i\!-\!1) \!+\! \left\lceil\frac{2(i\!-\!1)\lfloor\epsilon k\rfloor-R(i\!-\!1)}{\lceil|B'|/\lfloor\epsilon k\rfloor\rceil}\right\rceil}$ $= \displaystyle{R(i\!-\!1) + \left\lceil\frac{2(i-1)\lfloor\epsilon k\rfloor-R(i\!-\!1)}{\lceil3(k-(i-1)\lfloor\epsilon k\rfloor)/\lfloor\epsilon k\rfloor\rceil}\right\rceil}$.
	\end{enumerate}


\smallskip
\item Let $U'=U^{\frac{1}{\epsilon}}_{ord}$.

\smallskip
\item Let $S_{1+\frac{1}{\epsilon}\lfloor\epsilon k\rfloor},\ldots,S_k$ be the sets in  ${\cal S}'\setminus\{S_1,\ldots,S_{\frac{1}{\epsilon}\lfloor\epsilon k\rfloor}\}$, such that $\min(S_{j-1})$ $<\min(S_j)$ where $\min$ and $<$ are computed according to the ordered universe $U'$.

\end{enumerate}

{\noindent We have thus defined an instance $(U',{\cal S},w,W,k,f)$ of {\sc $(3,k)$-CWSP} that is examined by \alg{WSPAlg} in Step \ref{step:wspaccept}. Since this is a yes-instance (by the above arguments, $S_1,S_2,\ldots,S_k$ is a solution to this instance), the algorithm accepts.}\qed
\end{proof}

\section{Solving the $P_2$-Packing Problem}\label{section:p2pack}

We say that a set of $t$ (node-)disjoint simple paths, each on 3 nodes, is a $t$-packing. We first note that the papers \cite{p2packdet,p2packrand} develop an algorithm that given a bipartite graph $H=(A,B,E)$, decides in polynomial time if $G$ contains a $k$-packing such that the end-nodes of its paths belong to $B$. Then, given a graph $G=(V,E)$, they solve {\sc $P_2$-Packing} by examining $8^{k+o(k)}$ options to deterministically partition $V$ into the node-sets $A$ and $B$, and accepting {\em iff} at least one of the resulting bipartite graphs $H=(A,B,\{\{a,b\}\in E: a\in A, b\in B\})$ has a $k$-packing such that the end-nodes of its paths belong to $B$. In the randomized version, they examine $6.75^k$ such options. The goal is to examine an option that captures the end-nodes of a $k$-packing in $G$ (if one exists) in $B$, and its other nodes in $A$. We observe that this can be done by using a $(|V|,3k,k)$-universal set (see Definition \ref{theorem:splitter}). By Theorem \ref{theorem:splitter}, this results in an $O^*(6.75^{k+o(k)})$-time deterministic algorithm for {\sc $P_2$-Packing}.

In the rest of this section, we develop an alternative deterministic FPT algorithm for {\sc $P_2$-Packing}, demonstrating the fifth strategy in Section~\ref{section:strategies}. This result is the least interesting one in the paper, but it shows a different approach, combined with iterative compression, in which one can apply our unbalanced cutting of the universe technique.
We will focus on the following variant of {\sc $P_2$-Packing}:

\myparagraph{Iterative Compression $P_2$-Packing ($P_2$-ICP)} Given an undirected graph $G=(V,E)$ and a parameter $k\in\mathbb{N}$, along with a $(k-1)$-packing in $G$, denoted $S_1=(V_1,E_1),S_2=(V_2,E_2),\ldots,S_{k-1}=(V_{k-1},E_{k-1})$, decide if $G$ has a $k$-packing.

\smallskip
Clearly, if we can solve this variant in time $T$, we can solve {\sc $P_2$-Packing} in time $O^*(T)$ by using the algorithm for {\sc $P_2$-ICP} $k$ times, starting with an empty solution, and then, for $i=1,2,\ldots,k$, constructing a solution that contains exactly $i$ paths. It will be useful to focus on {\sc $P_2$-ICP}, since we can thus use the following result of \cite{p2packraible}:\footnote{We note that our approach does not lead to an algorithm for {\sc Weighted $P_2$-Packing} since Theorem \ref{theorem:p2improve} does not concern weighted graphs.}

\begin{theorem}\label{theorem:p2improve}
Let $G$ be a graph that has a $k$-packing. For any $(k\!-\!1)$-packing in $G$, denoted $S_1=(V_1,E_1),S_2\!=(V_2,E_2),\ldots,S_{k-1}=(V_{k-1},E_{k-1})$, $G$ has a $k$-packing that contains at least $\lceil2.5(k-1)\rceil$ nodes from $\bigcup_{i=1}^{k-1}V_i$.
\end{theorem}

Assume WLOG that $k$ is odd. We next denote $X=\bigcup_{i=1}^{k-1}V_i$ and $Y= V\setminus X$. Moreover, given $p\in\{3,4,\ldots,3k-2.5(k-1)\}$ and $q\in\{\lceil\frac{p}{3}\rceil,\lceil\frac{p}{3}\rceil+1,\ldots,p\}$, let $Sol(p,q)$ be the set of $k$-packings that contain exactly $p$ nodes from $Y$ and  exactly $q$ paths such that each of them includes at least one node from $Y$. By Theorem \ref{theorem:p2improve}, an instance of {\sc $P_2$-ICP} is a yes-instance iff $\bigcup_{p=3}^{3k-2.5(k-1)}\bigcup_{q=\lceil \frac{p}{3}\rceil}^pSol(p,q)\neq\emptyset$.

In the following two sections, we develop two deterministic procedures, \alg{ICPPro1} and \alg{ICPPro2}, for which we prove the following results:

\begin{lemma}\label{lemma:pro1}
Given $p\in\{3,4,\ldots,3k-2.5(k-1)\}$, $q\in\{\lceil\frac{p}{3}\rceil,\lceil\frac{p}{3}\rceil+1,\ldots,p\}$ and an instance of {\sc $P_2$-ICP}, \alg{ICPPro1} returns a family ${\cal F}$ of subsets of size $3q-p$ of $X$ such that [$Sol(p,q)\neq\emptyset$] iff [there exist a set $F\in{\cal F}$ and a $(k-q)$-packing in $G$ whose nodes belong to $X\setminus F$]. \alg{ICPPro1} runs in time $O^*(6.75^{k+o(k)})$.
\end{lemma}

\begin{lemma}\label{lemma:pro2}
Given $p\in\{3,4,\ldots,3k-2.5(k-1)\}$, $q\in\{\lceil\frac{p}{3}\rceil,\lceil\frac{p}{3}\rceil+1,\ldots,p\}$ and an instance of {\sc $P_2$-ICP} along with a family ${\cal F}$ of subsets of size $3q-p$ of $X$, \alg{ICPPro2} decides if there exist a set $F\in{\cal F}$ and a $(k-q)$-packing in $G$ whose nodes belong to $X\setminus F$. \alg{ICPPro2} runs in time $O^*(6.777^k)$.
\end{lemma}

Having such procedures, we can solve {\sc $P_2$-ICP} by trying, for each possible pair $(p,q)$ such that $p\in\{3,4,\ldots,3k-2.5(k-1)\}$ and $q\in\{\lceil\frac{p}{3}\rceil,\lceil\frac{p}{3}\rceil+1,\ldots,p\}$, to call procedure \alg{ICPPro1}, which returns a family $\cal F$, and then accept if \alg{ICPPro2} accepts when called with $\cal F$. If no execution of \alg{ICPPro2} accepts, we reject. We thus obtain the following result, which concludes the correctness of Theorem \ref{theorem:p2Pack} (in Section \ref{section:strategies}):

\begin{theorem}
{\sc $P_2$-ICP} can be solved in determinstic time $O^*(6.777^k)$.
\end{theorem}

\subsection{The Procedure \alg{ICPPro1}: Proof of Lemma \ref{lemma:pro1}}

When we next refer to representative families, suppose that they are computed with respect to $Y$. Moreover, assume an arbitrarily order $Y = \{y_1,y_2,\ldots,y_n\}$. We present a dynamic programming-based procedure to prove the lemma, in which we embed representative sets computations. To this end, we use a matrix M that has an entry $[p',q',m,X']$ for all $p'\in\{1,2,\ldots,p\}$, $q'\in\{\lceil\frac{p'}{3}\rceil,\lceil\frac{p'}{3}\rceil+1,\ldots,\min\{p',q\}\}$, $m\in\{1,2,\ldots,n\}$ and $X'\subseteq X$ of size $3q'-p'$.

We next assume that a reference to an undefined entry returns $\emptyset$. The other entries will store the following families of partial solutions:

\begin{itemize}
\item M$[p',q',m,X']$: A family that $(p-p')$-represents the family $F$ defined as follows. Let $\widetilde{F}$ be the family of every $q'$-packing whose paths contain exactly $p'$ nodes from $Y$ and whose set of other nodes (i.e., nodes in $X$) is exactly $X'$, such that each of its paths contains a node from $Y$ that is smaller than $y_m$, except one path whose smallest node from $Y$ is $y_m$.\\For each such packing in $\widetilde{F}$, $F$ contains a set that includes the nodes of its paths that belong to $Y$, excluding the smallest node from $Y$ in the node-set of each of its paths.
\end{itemize}

\smallskip
{\noindent The entries are computed in the following order:}
\begin{itemize}
\item For $p'=1,2,\ldots,p$:
	\begin{itemize}
	\item For $q'=\lceil\frac{p'}{3}\rceil,\lceil\frac{p'}{3}\rceil+1,\ldots,\min\{p',q\}$:
		\begin{itemize}
				\item Compute all entries of the form M$[p',q',m,X']$.
		\end{itemize}
	\end{itemize}
\end{itemize}

\medskip
{\noindent We now give the recursive formulas using which the entries are computed.}

\begin{enumerate}
\item If $q'=1$: M$[p',q',m,X'] = \{Y'\subseteq Y: \min(Y')=y_m$, there exists a 1-packing in $G$ such that the node-set of its path is $X'\cup Y'\}$.
\smallskip
\item Else: M$[p',q',m,X'] = \{Y'\cup Y'': Y'\cap Y''=\emptyset, \min(Y')=y_m,$ there exists $\widetilde{X}\subseteq X'$ for which [there exists a 1-packing in $G$ such that the node-set of its path is $\widetilde{X}\cup Y'$] and [$Y''\in\bigcup_{m'=1}^{m-1}\mathrm{M}[p'-|Y'|,q'-1,m',X'\setminus\widetilde{X}]$]$\}$.
\end{enumerate}

\vspace{-0.8em}
\begin{itemize}
\item After 2: Replace the result by a family that max $(p-p')$-represents it.
\end{itemize}

{\noindent Finally, we return ${\cal F}=\{X'\subseteq X: \bigcup_{m=1}^n\mathrm{M}[p,q,m,X']\neq\emptyset\}$. Note that the size of each set in the family stored in an entry M$[p',q',m,X']$ is $(p'-q')$. Furthermore, note that, as required, [$Sol(p,q)\neq\emptyset$] iff [there exist a set $F\in{\cal F}$ and a $(k-q)$-packing in $G$ whose nodes belong to $X\setminus F$].}

\bigskip

{\noindent By Theorem \ref{theorem:esarep}, the running time is bounded by $O^*$ of $2^{o(k)}$, multiplied by the following expression:}
\[\displaystyle{\max_{p'=1}^p\max_{q'=\lceil\frac{p'}{3}\rceil}^{\min\{p',q\}}{3(k-1) \choose 3q'-p'}\frac{(c(p-q'))^{2(p-q')-(p'-q')}}{(p'-q')^{p'-q'}(c(p-q')-(p'-q'))^{2(p-q')-2(p'-q')}}}\]

{\noindent We can bound the above expression by $O^*$ of:}

\[\displaystyle{\max_{p'=1}^{\frac{k}{2}}\max_{\frac{p'}{3}\leq q'\leq p'}\frac{(3k)^{3k}}{(3q'-p')^{3q'-p'}(3k-3q'+p')^{3k-3q'+p'}}\cdot\frac{(c(\frac{k}{2}-q'))^{k-p'-q'}}{(p'-q')^{p'-q'}(c(\frac{k}{2}-q')-(p'-q'))^{k-2p'}}}\]

{\noindent Now, we bound this expression by $O^*$ of:}

\[\displaystyle{
\max_{0\leq\alpha\leq 1}\max_{\frac{\alpha \frac{k}{2}}{3}\leq q'\leq \alpha \frac{k}{2}}
\frac{(3k)^{3k}}{(3q'\!-\!\alpha \frac{k}{2})^{3q'\!-\!\alpha \frac{k}{2}}(3k\!-\!3q'\!+\!\alpha \frac{k}{2})^{3k\!-\!3q'\!+\!\alpha \frac{k}{2}}}
\!\cdot\!
\frac{(c(\frac{k}{2}\!-\!q'))^{k\!-\!\alpha \frac{k}{2}\!-\!q'}}{(\alpha \frac{k}{2}\!-\!q')^{\alpha \frac{k}{2}\!-\!q'}(c(\frac{k}{2}\!-\!q')\!-\!(\alpha \frac{k}{2}\!-\!q'))^{k\!-\!\alpha k}}}\]

{\noindent This expression is further bounded by $O^*$ of:}

\[\displaystyle{
\max_{0\leq\alpha\leq 1}\max_{\frac{1}{3}\leq \beta\leq 1}
\left(
\frac{3^3}{(\frac{3\beta\alpha}{2}\!-\!\frac{\alpha}{2})^{\frac{3\beta\alpha}{2}\!-\!\frac{\alpha}{2}}(3\!-\!\frac{3\beta\alpha}{2}\!+\! \frac{\alpha}{2})^{3\!-\!\frac{3\beta\alpha}{2}\!+\!\frac{\alpha}{2}}}
\!\cdot\!
\frac{(c(\frac{1}{2}\!-\!\frac{\beta\alpha}{2}))^{1-\frac{\alpha}{2}-\frac{\beta\alpha}{2}}}{(\frac{\alpha}{2}\!-\!\frac{\beta\alpha}{2})^{\frac{\alpha}{2}\!-\!\frac{\beta\alpha}{2}}(c(\frac{1}{2}\!-\!\frac{\beta\alpha}{2})\!-\!( \frac{\alpha}{2}\!-\!\frac{\beta\alpha}{2}))^{1\!-\!\alpha}}
\right)^k}\]

{\noindent Which is bounded by $O^*$ of:}

\[\displaystyle{
\max_{0\leq\alpha\leq 1}\max_{\frac{1}{3}\leq \beta\leq 1}
\left(
\frac{6^6}{(3\beta\alpha\!-\!\alpha)^{3\beta\alpha\!-\!\alpha}(6\!-\!3\beta\alpha\!+\!\alpha)^{6\!-\!3\beta\alpha\!+\!\alpha}}
\!\cdot\!
\frac{(c(1\!-\!\beta\alpha))^{2\!-\!\alpha\!-\!\beta\alpha}}{(\alpha\!-\!\beta\alpha)^{\alpha\!-\!\beta\alpha}(c(1\!-\!\beta\alpha)\!-\!(\alpha\!-\!\beta\alpha))^{2\!-\!2\alpha}}
\right)^{\frac{k}{2}}}\]

\smallskip
{\noindent For $c=1$, the maximum is obtained when $\alpha=\beta=1$, which results in the bound $O^*(6.75^{k+o(k)})$.}\qed

\subsection{The Procedure \alg{ICPPro2}: Proof of Lemma \ref{lemma:pro2}}

First note that the nodes in $Y$ are not relevant to this section---we can focus on the subgraph of $G$ induced by $X$. Furthermore, we need only know the set of 1-packings in this subgraphs, which can be given as a family of sets on 3 nodes. Thus, to prove the lemma, we will solve following problem in time $O^*(6.777^k)$:

\myparagraph{$P_2$-Pro2} Given a universe $U$ of size $3k$, a family $\cal S$ of subsets of size 3 of $U$, $p\in\{1,2,\ldots,\lfloor\frac{k}{2}\rfloor\}$, $q\in\{\lceil\frac{p}{3}\rceil,\lceil\frac{p}{3}\rceil+1,\ldots,p\}$ and a family ${\cal F}$ of subsets of size $3q-p$ of $U$, decide if there exist a set $F\in{\cal F}$ and a subfamily ${\cal S}'\subseteq {\cal S}$ of $(k-q)$ disjoint sets that do not contain any element from $F$.

\smallskip
We next solve a subcase of {\sc $P_2$-Pro2}, and then show how to translate {\sc $P_2$-Pro2} to this subcase by using unbalanced cutting of the universe. To this end, we follow the ideas introduced in Appendix \ref{section:3kwsp}, where the main differences result from the fact that now, since the universe is of the small size $3k$, we do not need to perform any computation of a representative family.

\subsubsection{Solving a Subcase of $P_2$-Pro2}\label{section:subcaseficp}$\ $\\

{\noindent In this section we solve the subcase {\sc Cut $P_2$-Pro2 ($P_2$-CPro2)} of {\sc $P_2$-Pro2}, defined as follows. Fix $0<\epsilon<0.1$, whose value is determined later, such that $\frac{1}{\epsilon}\in\mathbb{N}$. The input for {\sc $P_2$-CPro2} consists of an {\em ordered} universe $U=\{u_1,u_2,\ldots,u_{3k}\}$, where $u_{i-1}<u_i$ for all $i\in\{2,3,\ldots,3k\}$. We are also given a family $\cal S$ of subsets of size 3 of $U$, $p\in\{1,2,\ldots,\lfloor\frac{k}{2}\rfloor\}$, $q\in\{\lceil\frac{p}{3}\rceil,\lceil\frac{p}{3}\rceil+1,\ldots,p\}$ and a family ${\cal F}$ of subsets of size $3q-p$ of $U$, along with a non-decreasing function $f: \{1,2,\ldots,\frac{1}{\epsilon}\}\rightarrow U$.}

We need to decide if there is a set $F\in{\cal F}$ and an {\em ordered} subfamily ${\cal S}'\subseteq {\cal S}$ of $(k-q)$ disjoint sets, denoted accordingly as ${\cal S}'=\{S_1,S_2,\ldots,S_{k-q}\}$, such that $F\cap(\bigcup{\cal S}')=\emptyset$ and the following ``solution conditions'' are satisfied:
\begin{enumerate}
\item $\forall i\in\{2,3,\ldots,k-q\}$: $\min(S_{i-1})<\min(S_i)$.

\item $\forall i\in\{1,2,\ldots,\frac{1}{\epsilon}\}$: At least
$R(i)$
elements in $\displaystyle{F\cup(\bigcup_{j=1}^{i\lfloor\epsilon (k-q)\rfloor}(S_j\setminus\{\min(S_j)\}))}$ are smaller or equal to $f(i)$, where $R$ is defined according to the following recursion.
	\bigskip
	\begin{itemize}
	\item $R(0)=0$.
	\item For $j=1,\ldots,\frac{1}{\epsilon}$:\\$\displaystyle{R(j) = R(j-1) +
	\left\lceil\frac{(3q-p)+2(j-1)\lfloor\epsilon (k-q)\rfloor
	-R(j-1)}
	{\lceil 3((k-q)-(j-1)\lfloor\epsilon (k-q)\rfloor)/\lfloor \epsilon (k-q)\rfloor\rceil}\right\rceil}$.
	\end{itemize}

\item $\forall i\in\{1,2,\ldots,\frac{1}{\epsilon}\}$: All the elements in $\displaystyle{\bigcup_{j=1+i\lfloor\epsilon (k-q)\rfloor}^{(k-q)}S_j}$ are larger than $f(i)$.
\end{enumerate}

{\noindent We next show that {\sc $P_2$-CPro2} can be efficiently solved using dynamic programming, proving the following lemma:}

\begin{lemma}\label{lemma:cpro2run}
For any fixed $0\!<\!\epsilon\!<\!1$, {\sc $P_2$-CPro2} is solvable in deterministic~time\\
$Y=\displaystyle{O^*(
\max_{i=1}^{\frac{1}{\epsilon}}\max_{j=1+(i-1)\lfloor\epsilon (k-q)\rfloor}^{i\lfloor\epsilon (k-q)\rfloor}
{3k-j-R(i-1) \choose (3q-p)+2j-R(i-1)}  
)}$.
\end{lemma}

\begin{proof}
We now present a dynamic programming-based procedure to prove the lemma. To this end, we use the following two matrices:
\begin{enumerate}
\item M that has an entry $[i,j,s_1,s_2,\ldots,s_{\frac{1}{\epsilon}},m,U']$ for all $i\in\{1,2,\ldots,\frac{1}{\epsilon}+1\}$, $j\in \{1+(i-1)\lfloor\epsilon (k-q)\rfloor,\ldots,J\})$ where [$(i\leq\frac{1}{\epsilon}\rightarrow J=i\lfloor\epsilon (k-q)\rfloor)\wedge(i=\frac{1}{\epsilon}+1\rightarrow J=k-q)$],
[$\forall\ell\in\{1,\ldots,\frac{1}{\epsilon}\}: s_\ell\in\{0,1,\ldots,(3q-p)+2j\}$] such that [$\forall\ell\in\{1,\ldots,i-1\}: s_\ell\geq R(\ell)$],
%
$m\in\{j,\ldots,3k\}$ such that ($i>1\rightarrow u_m>f(i-1)$), and $U'\subseteq \{u_\ell\in U: \ell>j+s_{i-1}\}$ of size $(3q-p)+2j-s_{i-1}$, where $s_0=0$.

\item N has an entry $[i,s_1,s_2,\ldots,s_{\frac{1}{\epsilon}},m,U']$ for all $i\in\{2,3,\ldots,\frac{1}{\epsilon}\}$,
%
[$\forall\ell\in\{1,\ldots,\frac{1}{\epsilon}\}: s_\ell\in\{0,1,\ldots,(3q-p)+2(i\lfloor\epsilon(k-q)\rfloor)\}$] such that [$\forall\ell\in\{1,\ldots,i-1\}: s_\ell\geq R(\ell)$],
%
$m\in\{(i\lfloor\epsilon(k-q)\rfloor),\ldots,3k\}$ such that ($i>1\rightarrow u_m>f(i-1)$), and $U'\subseteq \{u_\ell\in U: \ell>(i\lfloor\epsilon(k-q)\rfloor)+s_{i-1}\}$ of size $(3q-p)+2j-s_{i-1}$.
\end{enumerate}

{\noindent We next assume that a reference to an undefined entry returns FALSE. The other entries will store the following boolean values:}
\begin{itemize}
\item The entry M$[i,j,s_1,\ldots,s_{\frac{1}{\epsilon}},m,U']$ holds TRUE iff there is a set $F\in{\cal F}$ and an {\em ordered} subfamily ${\cal S}'\subseteq {\cal S}$ of $j$ disjoint sets, denoted accordingly as ${\cal S}'=\{S_1,S_2,\ldots,S_j\}$, such that $F\cap(\bigcup{\cal S}')=\emptyset$ and the following ``solution conditions'' are satisfied:
\begin{enumerate}
\smallskip
\item $\forall \ell\in\{2,\ldots,j\}$: $\min(S_{\ell-1})<\min(S_\ell)$.
\smallskip
\item $\min(S_j)=u_m$.

\item $\forall \ell\!\in\!\{1,\ldots,i\!-\!1\}$: At least $R(\ell)$
elements in $\displaystyle{F\cup(\!\!\!\!\bigcup_{p=1}^{\ell\lfloor\epsilon (k-q)\rfloor}\!\!\!\!(S_p\!\setminus\!\{\min(S_p)\}))}$ are smaller or equal to~$f(\ell)$.

\item $\forall \ell\in\{1,\ldots,\frac{1}{\epsilon}\}$: Exactly $s_\ell$ elements in $\displaystyle{F\cup(\bigcup_{p=1}^{j}(S_p\setminus\{\min(S_p)\}))}$ are smaller or equal to $f(\ell)$.

\item Define $f(0)$ as a value smaller than $u_1$. Then, $U'$ is the set elements in $\displaystyle{F\cup(\bigcup_{p=1}^{j}(S_p\!\setminus\!\{\min(S_p)\}))}$ that are larger than $f(i-1)$.

\item $\forall \ell\in\{1,\ldots,i\!-\!1\}$: All the elements in $\displaystyle{\bigcup_{p=1+\ell\lfloor\epsilon (k-q)\rfloor}^{j}\!\!\!S_p}$ are larger than~$f(\ell)$.
\end{enumerate}

\item The entry N$[i,s_1,\ldots,s_{\frac{1}{\epsilon}},m,U']$ holds TRUE iff at least one entry M$[i,i\lfloor\epsilon(k-q)\rfloor,s_1,\ldots,s_{\frac{1}{\epsilon}},m,U'']$, where $U' = U''\setminus\{u\in U: f(i-1)<u\leq f(i)\}$, holds TRUE.
\end{itemize}

\medskip
{\noindent Initialize the entries of N to FALSE. Then, the entries are computed in the following order:}
\begin{enumerate}
\item For $i=1,\ldots,\frac{1}{\epsilon}+1$:
	\begin{enumerate}
	\item For $j=1+(i-1)\lfloor\epsilon k\rfloor,\ldots,J$:
		\begin{enumerate}
				\item Compute all entries of the form M$[i,j,s_1,\ldots,s_{\frac{1}{\epsilon}},m,U']$.
		\end{enumerate}
		\item If $i\leq\frac{1}{\epsilon}$:
			\begin{enumerate}		
			\item For each entry of the form M$[i,i\lfloor\epsilon (k-q)\rfloor,s_1,\ldots,s_{\frac{1}{\epsilon}},m,U']$ that holds TRUE:
				\begin{enumerate}
				\item Assign TRUE to N$[i,s_1,\ldots,s_{\frac{1}{\epsilon}},m,U'\setminus\{u\in U: f(i-1)<u\leq f(i)\}]$.
				\end{enumerate}
			\end{enumerate}
	\end{enumerate}
\end{enumerate}

\smallskip
\medskip
{\noindent We now give the recursive formulas using which the entries of M are computed.}
\begin{enumerate}
\item If $j=1$: M$[1,1,s_1,\ldots,s_{\frac{1}{\epsilon}},m,U']$ holds TRUE {\em iff} there exist disjoint $F\!\in\!{\cal F}$ and $S\!\in\!{\cal S}$ such that $\min(S)\!=\!u_m$, $[\forall\ell\!\in\!\{1,\ldots,\frac{1}{\epsilon}\}\!: |\{u_p\!\in\! F\!\cup\!S\!\setminus\!\{u_m\}: u_p\!\leq\!f(\ell)\}|\!=\!s_\ell]$, and $U'\!=\!F\!\cup\! (S\!\setminus\!\{u_m\})$.

\smallskip
\medskip
\item Else if $j=1+(i-1)\lfloor\epsilon(k-q)\rfloor$: M$[i,j,s_1,\ldots,s_{\frac{1}{\epsilon}},m,U']$ holds TRUE {\em iff} there~exist~disjoint sets $S$ and $A$ such that $S\in{\cal S}$, $\min(S)=u_m$, $\displaystyle{\mathrm{TRUE}\in}$ $\displaystyle{\bigcup_{m'=1}^{m-1}\{\mathrm{N}[i-1,s_1',\ldots,s_{\frac{1}{\epsilon}}',m',A]\}}$, where [$\forall\ell\in\{1,\ldots,\frac{1}{\epsilon}\}: s_\ell' = s_\ell-|\{u_p\in S\setminus\{u_m\}: u_p\leq f(\ell)\}|$], and $U'=A\cup(S\!\setminus\!\{u_m\})$.

\smallskip
\medskip
\item Else: M$[i,j,s_1,\ldots,s_{\frac{1}{\epsilon}},m,U']$ holds TRUE {\em iff} there exist disjoint sets $S$ and $A$ such that $S\in{\cal S}$, $\min(S)=u_m$, $\displaystyle{\mathrm{TRUE}\in\bigcup_{m'=1}^{m-1}\{\mathrm{M}[i,j-1,s_1',\ldots,s_{\frac{1}{\epsilon}}',m',A]\}}$, where [$\forall\ell\in\{1,\ldots,\frac{1}{\epsilon}\}: s_\ell' = s_\ell-|\{u_p\in S\setminus\{u_m\}: u_p\leq f(\ell)\}|$], and $U'=A\cup(S\!\setminus\!\{u_m\})$.
\end{enumerate}

\medskip
{\noindent Finally, we return yes {\em iff} at least one entry of the form M$[i,(k-q),s_1,\ldots,s_{\frac{1}{\epsilon}},m,U']$ holds TRUE.}

\bigskip
{\noindent The matrices M and N contain $\displaystyle{O^*(
\max_{i=1}^{\frac{1}{\epsilon}}\max_{j=1+(i-1)\lfloor\epsilon (k-q)\rfloor}^{i\lfloor\epsilon (k-q)\rfloor}
{3k\!-\!j\!-\!R(i\!-\!1) \choose (3q\!-\!p)\!+\!2j\!-\!R(i\!-\!1)}  
)}$ entries, where each entry of M can be computed in polynomial time, and the total time for computing all of the entries of N is bounded by $O^*$ of the number of entries of M. Thus, the entire procedure can be performed in the desired time.}\qed
\end{proof}

\bigskip
{\noindent We proceed to analyze the bound, denoted $Y$, for the running time in Lemma \ref{lemma:cpro2run}. First, define $T$ according to the following recursion:}
\begin{itemize}
\item $T(0)=0$.
\item For $j=1,\ldots,\frac{1}{\epsilon}$: $\displaystyle{T(j) = T(j-1) +
	\epsilon \left(\frac{\frac{3q-p}{k-q} + 2(j-1)\epsilon
	-T(j-1)}
	{3(1-(j-1)\epsilon)}\right)}$.
\end{itemize}

\smallskip
{\noindent Thus, we get that:}
\[Y=\displaystyle{O^*(
\max_{i=1}^{\frac{1}{\epsilon}}\max_{(i-1)\epsilon(k-q)\leq j\leq i\epsilon(k-q)}
{3k-j-T(i-1)(k-q) \choose (3q-p)+2j-T(i-1)(k-q)}  
)}.\]

{\noindent The worst case is obtained for instances where $q$ is maximal, i.e., when $p=q=\lfloor\frac{k}{2}\rfloor$. Therefore, the running time is bounded by:}

\[
\begin{array}{l}
\bigskip
\displaystyle{O^*(
\max_{i=1}^{\frac{1}{\epsilon}}\max_{(i-1)\epsilon\frac{k}{2}\leq j\leq i\epsilon\frac{k}{2}}
{3k-j-T(i-1)\frac{k}{2} \choose k+2j-T(i-1)\frac{k}{2}}  
)}\\

\bigskip
=\displaystyle{O^*(
\max_{i=1}^{\frac{1}{\epsilon}}\max_{(i-1)\leq\alpha\leq i}
{3k-\alpha\epsilon\frac{k}{2}-T(i-1)\frac{k}{2} \choose k+\alpha\epsilon k-T(i-1)\frac{k}{2}}  
)}\\

\bigskip
=\displaystyle{O^*(
\max_{i=1}^{\frac{1}{\epsilon}}\max_{(i-1)\leq\alpha\leq i}
\frac{(3k-\alpha\epsilon\frac{k}{2}-T(i-1)\frac{k}{2})^{3k-\alpha\epsilon\frac{k}{2}-T(i-1)\frac{k}{2}}}
{(k+\alpha\epsilon k-T(i-1)\frac{k}{2})^{k+\alpha\epsilon k-T(i-1)\frac{k}{2}}
(2k-3\alpha\epsilon\frac{k}{2})^{2k-3\alpha\epsilon\frac{k}{2}}} 
)}\\

\bigskip
=\displaystyle{O^*(
\max_{i=1}^{\frac{1}{\epsilon}}\max_{(i-1)\leq\alpha\leq i}
\left(
\frac{(3-\frac{\alpha\epsilon}{2}-\frac{T(i-1)}{2})^{3-\frac{\alpha\epsilon}{2}-\frac{T(i-1)}{2}}}
{(1+\alpha\epsilon-\frac{T(i-1)}{2})^{1+\alpha\epsilon-\frac{T(i-1)}{2}}
(2-\frac{3\alpha\epsilon}{2})^{2-\frac{3\alpha\epsilon}{2}}}
\right)^k
)}\\

=\displaystyle{O^*(
\max_{i=1}^{\frac{1}{\epsilon}}\max_{(i-1)\leq\alpha\leq i}
\left(
\frac{(6-\alpha\epsilon-T(i-1))^{6-\alpha\epsilon-T(i-1)}}
{(2+2\alpha\epsilon-T(i-1))^{2+2\alpha\epsilon-T(i-1)}
(4-3\alpha\epsilon)^{4-3\alpha\epsilon}}
\right)^{\frac{k}{2}}
)}.
\end{array}
\]

\smallskip
{\noindent Denote the expression above by $Z$. To allow us to compute a bound efficiently, we choose $\epsilon=10^{-5}$. Choosing a smaller $\epsilon$ results in a better bound, but the improvement is negligible. We thus get that the maximum is obtained at $i=\alpha=6377$, where $T(i-1)\cong0.04485$, which results in the bound $O^*(6.77682^k)$. We have proved the following corollary:}

\begin{cor}\label{cor:cpro2}
For $\epsilon=10^{-5}$, {\sc $P_2$-CPro2} can be solved in deterministic time $O^*(6.777^k)$.
\end{cor}

\subsubsection{A Deterministic Procedure for $P_2$-Pro2}\label{section:ficp}$\ $\\

{\noindent Fix $\epsilon=10^{-5}$. The pseudocode of \alg{Procedure2}, a deterministic procedure that solves {\sc $P_2$-Pro2}, is given below. It uses unbalanced cutting of the universe in the same manner as \alg{WSPAlg} (see Appendix \ref{section:wspdet}), obtaining a set of inputs for {\sc $P_2$-CPro2} and accepting {\em iff} at least one of them is a yes-instance.}

\begin{algorithm}[!ht]
\caption{\alg{Procedure2}($U,k,{\cal S},p,q,{\cal F}$)}
\begin{algorithmic}[1]
\STATE\label{step:pro2algorder} order the universe $U$ arbitrarily as $U=\{u_1,u_2,\ldots,u_{3k}\}$.

\FORALL{distinct $\ell_1,\ell_2,\ldots,\ell_{\frac{1}{\epsilon}},r_1,r_2,\ldots,r_{\frac{1}{\epsilon}}\in U$ s.t. $\ell_i<r_i$ for all $i\in\{1,2,\ldots,\frac{1}{\epsilon}\}$}\label{step:pro2orderlr}
	\FOR{$i=1,2,\ldots,\frac{1}{\epsilon}$}
		\STATE\label{step:pro2orderUi} define $U^i=\{{\ell_i},{\ell_i+1},\ldots,{r_i}\}\setminus \bigcup_{j=1}^{i-1}U^j$.
	\ENDFOR
	\STATE\label{step:pro2orderUlast2} $U^{\frac{1}{\epsilon}+1} = U\setminus \bigcup_{j=1}^{\frac{1}{\epsilon}}U^j$.
	\STATE\label{step:pro2orderU} let $U'=\{u'_1,u'_2,\ldots,u'_n\}$ be an ordered copy of $U$ such that the elements in $U^1$ appear first, then those in $U^2$, and so on, where the order between the elements in each $U^i$ is preserved according to their order in $U$.
	\STATE\label{step:pro2orderf} define $f:\{1,2,\ldots,\frac{1}{\epsilon}\}\rightarrow U'$ such that $f(i)$ is the largest element in $U'$ that belongs to $U^i$.
	\STATE\label{step:pro2accept} {\bf if} the procedure of Corollary \ref{cor:cpro2} accepts $(U',k,{\cal S},p,q,{\cal F},f)$ {\bf then} accept. {\bf end if}
\ENDFOR
\STATE reject.
\end{algorithmic}
\end{algorithm}

We next consider the correctness and running time of \alg{Procedure2}.

\begin{lemma}\label{lemma:pro2algcor}
\alg{Procedure2} solves {\sc $P_2$-Pro2} in time $O^*(6.777^k)$.
\end{lemma}

\begin{proof}
The running time of the algorithm follows immediately from the pseudocode and Corollary \ref{cor:cpro2}.

For the easier direction, note that for any ordering of the universe $U$ as $U'$, and for any function $f:\{1,2,\ldots,\frac{1}{\epsilon}\}$, a solution to the instance $(U',k,{\cal S},p,q,{\cal F},f)$ of {\sc $P_2$-CPro2} is also a solution to the instance $(U,k,{\cal S},p,q,{\cal F})$ of {\sc $P_2$-Pro2}. Therefore, if the algorithm accepts, the input is a yes-instance of {\sc $P_2$-Pro2}.

Now, suppose that the input is a yes-instance of {\sc $P_2$-Pro2}, and let $F\in{\cal F}$ and ${\cal S}'\subseteq{\cal S}$ be a corresponding solution. Let $U=\{u_1,u_2,\ldots,u_{3k}\}$ be the order chosen by \alg{Procedure2} in Step \ref{step:pro2algorder}. We define $U^1,U^2,\ldots,U^{\frac{1}{\epsilon}}$, $f$ and an order $S_1,S_2,\ldots,S_k$ of the sets in ${\cal S}'$ as follows.

\begin{enumerate}
\item Initialize:
	\begin{enumerate}
	\smallskip
	\item Let $U^0_{ord}=\{u'_1,u'_2,\ldots,u'_{3k}\}$ be an ordered copy of $U$.
	
	\smallskip
	\item Let $P_0=\emptyset$.	
	\end{enumerate}
	
\medskip	
\item For $i=1,2,\ldots,\frac{1}{\epsilon}$:
	\begin{enumerate}
	\smallskip
	\item Let $A' = F\cup (\bigcup_{j=1}^{(i-1)\lfloor\epsilon (k-q)\rfloor}S_j)$, and $A'_{min} = \bigcup_{j=1}^{(i-1)\lfloor\epsilon (k-q)\rfloor}\{\min(S_j)\}$ where $\min$ is computed according to the ordered universe $U^{i-1}_{ord}$. 
	
	\smallskip
	\item Let $B'=(\bigcup{\cal S}')\setminus(\bigcup_{j=1}^{(i-1)\lfloor\epsilon (k-q)\rfloor}S_j)$.	
	
	\smallskip
	\item Let $B^1,B^2,\ldots,B^x$, where $x=\lceil\frac{|B'|}{\lfloor\epsilon (k-q)\rfloor}\rceil$, be a partition of $U^{i-1}_{ord}\setminus(\bigcup_{j=1}^{i-1}U^j)$ (into $x$ disjoint universes) such that each $B^i$ is a set of elements that are consecutively ordered in $U^{i-1}_{ord}$ and contains exactly $\lfloor\epsilon (k-q)\rfloor$ elements from $B'$, except $B^x$, if $(|B'|\!\ \mathrm{mod}\!\ \lfloor\epsilon (k-q)\rfloor)\neq 0$, which contains exactly $(|B'|\!\ \mathrm{mod}\!\ \lfloor\epsilon (k-q)\rfloor)$ elements from $B'$.
	
	\smallskip
	\item Let $U^i$ be a universe that contains the maximum number of elements~from $A'\setminus (A'_{min}\cup P_{i-1})$ among the universes $B^1,B^2,\ldots,B^x$. If $U^i$ contains elements from less than $\lfloor\epsilon (k-q)\rfloor$ sets in $({\cal S}'\setminus\{S_1,S_2,\ldots,S_{(i-1)\lfloor\epsilon (k-q)\rfloor}\}$, add elements from $U^{i-1}_{ord}\setminus(\bigcup_{j=1}^{i-1}U^j)$ to $U^i$ such that its elements remain consecutively ordered in $U^{i-1}_{ord}$ and it contains elements from exactly $\lfloor\epsilon (k-q)\rfloor$ sets in $({\cal S}'\setminus\{S_1,\ldots,S_{(i-1)\lfloor\epsilon (k-q)\rfloor}\}$.
	
	\smallskip
	\item Let $f(i)$ be the largest element in $U^i$.
	
	\smallskip
	\item Let $U^i_{ord}=\{u'_1,u'_2,\ldots,u'_{3k}\}$ be an ordered copy of $U$ such that the elements in $\bigcup_{j=1}^{i-1}U^j$ appear first, followed by those in $U^i$, where the internal orders between the elements in $\bigcup_{j=1}^{i-1}U^j$, $U^i$ and $U\setminus (\bigcup_{j=1}^{i}U^j)$ are preserved according to their order in $U^{i-1}_{ord}$.
	
	\smallskip
	\item Let $S_{(i-1)\lfloor\epsilon (k-q)\rfloor+1},S_{(i-1)\lfloor\epsilon (k-q)\rfloor+2},\ldots,S_{i\lfloor\epsilon (k-q)\rfloor}$ be the sets in ${\cal S}'\setminus\{S_1,S_2,\ldots,S_{(i-1)\lfloor\epsilon (k-q)\rfloor}\}$ that contain elements from $U^i$, such that $\min(S_{j-1})<\min(S_j)$ (according to $U^i_{ord}$).
	
	\smallskip
	\item Note that all the elements in $(\bigcup{\cal S}')\setminus(\bigcup_{j=1}^{i\lfloor\epsilon (k-q)\rfloor}S_j)$ are larger than $f(i)$ (according to $U^i_{ord}$).
	
	\smallskip
	\item Let $P_i$ be the set of elements in $\bigcup_{j=1}^{i\lfloor\epsilon (k-q)\rfloor}(S_j\setminus\{\min(S_j)\})$ that are smaller or equal to $f(i)$ (according to $U^i_{ord}$). Then, we have that:\\ 
	$\displaystyle{|P_i|\geq |P_{i-1}| + \left\lceil\frac{|A'\setminus (A'_{min}\cup P_{i-1})|}{x}\right\rceil 
	}$\\
	\medskip
	$\geq\displaystyle{R(i\!-\!1) \!+\! \left\lceil\!\frac{(3q\!-\!p) \!+\! 2(i\!-\!1)\lfloor\!\epsilon (k\!-\!q)\!\rfloor-R(i\!-\!1)}{\lceil|B'|/\lfloor\!\epsilon (k\!-\!q)\!\rfloor\rceil}\!\right\rceil} 
	$\\
	$=\displaystyle{R(i\!-\!1) \!+\! \left\lceil\!\frac{(3q\!-\!p) \!+\! 2(i\!-\!1)\lfloor\!\epsilon (k\!-\!q)\!\rfloor-R(i\!-\!1)}{\lceil3((k\!-\!q)\!-\!(i\!-\!1)\lfloor\!\epsilon (k\!-\!q)\!\rfloor)/\lfloor\!\epsilon (k\!-\!q)\!\rfloor\rceil}\!\right\rceil}$.		
	\end{enumerate}	


\smallskip
\item Let $U'=U^{\frac{1}{\epsilon}}_{ord}$.

\smallskip
\item Let $S_{1+\frac{1}{\epsilon}\lfloor\epsilon (k-q)\rfloor},\ldots,S_{k-q}$ be the sets in  ${\cal S}'\setminus\{S_1,\ldots,S_{\frac{1}{\epsilon}\lfloor\epsilon (k-q)\rfloor}\}$, such that $\min(S_{j-1})<\min(S_j)$ where $\min$ and $<$ are computed according to the ordered universe $U'$.

\end{enumerate}

{\noindent We have thus defined an instance $(U',k,{\cal S},p,q,{\cal F},f)$ of {\sc $P_2$-CPro2} that is examined by \alg{Procedure2} in Step \ref{step:pro2accept}. Since this is a yes-instance (by the above arguments, $S_1,S_2,\ldots,S_{k-q}$ is a solution to this instance), the procedure accepts.}\qed
\end{proof}

\end{document}